\documentclass[11pt]{article}
\usepackage{setspace}
\usepackage{mathtools}
\usepackage[utf8]{inputenc}
\usepackage{url}
\usepackage{natbib}     
\setcitestyle{year,open={(},close={)}}
\RequirePackage[colorlinks,citecolor=blue,urlcolor=blue]{hyperref}
\usepackage{doi}
\usepackage{url}
\usepackage{mathtools}
\usepackage{comment}

\usepackage[psamsfonts]{amssymb}
\usepackage{amsthm,amsmath,amsbsy,mathabx}
\usepackage{bbm, bm}
\usepackage[dvipsnames]{xcolor}
\usepackage{subcaption}
\usepackage[shortlabels]{enumitem}
\RequirePackage[colorlinks,citecolor=blue,urlcolor=blue]{hyperref}
\usepackage{graphicx,placeins}

\usepackage{inputenc}
\usepackage{graphicx,graphics,multirow}
\usepackage{subcaption}
\usepackage{supertabular}
\usepackage{amsmath,amscd,amsfonts,amssymb}
\usepackage{fancyhdr}
\usepackage{color}
\usepackage{verbatim}

\usepackage{pdflscape}
\usepackage{array,supertabular}
\usepackage{float}
\usepackage{rotating}
\usepackage{lscape}

\usepackage{booktabs,caption}
\usepackage[flushleft]{threeparttable}
\usepackage{amsthm}

\usepackage{dcolumn}
\newcolumntype{L}{D{.}{.}{2,5}}

\usepackage{chngpage} 

\newtheorem{theorem}{Theorem}[section]
\newtheorem{lemma}[theorem]{Lemma}

\newtheorem{proposition}[theorem]{Proposition}
\newtheorem{corollary}[theorem]{Corollary}

\newtheorem{assumption}{Assumption}
\newtheorem{remark}{Remark}[section]

\def\mds{\medskip}
\def\Rb{{\mathbb R}}

\def\Fc{{\mathcal F}}
\def\Nc{{\mathcal N}}

\long\def\symbolfootnote[#1]#2{\begingroup%
\def\thefootnote{\fnsymbol{footnote}}\footnotetext[#1]{#2}\footnotemark[#1]\endgroup}

\def\ba#1\ea{\begin{align*}#1\end{align*}} 
\def\banum#1\eanum{\begin{align}#1\end{align}} 

\makeatletter
\@addtoreset{equation}{section}

\makeatother

\newcounter{Fig}[figure]

\newcounter{Tab}[table]

  {\refstepcounter{Tab}
   \addtocounter{table}{1}
   \samepage\vspace{0.2cm}
   \centerline{#1} \nobreak
   \begin{verse}{Table \thesection.\arabic{Tab}:}
  }%
  {\end{verse} \vspace{0.2cm}}

\relax

\newcommand{\dd}{\mbox{\boldmath $d$}}

\newcommand{\uu}{\mbox{\boldmath $u$}}
\newcommand{\vv}{\mbox{\boldmath $v$}}

\def \Eb{{\mathbb E}}

\def \Gb{{\mathbb G}}
\def \Hb{{\mathbb H}}

\def \Lb{{\mathbb L}}

\def \Pb{{\mathbb P}}
\def \Rb{{\mathbb R}}

\def \Vb{{\mathbb V}}

\def \Zb{{\mathbb Z}}

\def \Ac{{\mathcal A}}
\def \Rc{{\mathcal R}}

\def \Uc{{\mathcal U}}

\def \Fc{{\mathcal F}}

\def \Sc{{\mathcal S}}

\def \Nc{{\mathcal N}}

\newcommand{\bqa}{\begin{eqnarray*}}
\newcommand{\eqa}{\end{eqnarray*}}
\newcommand{\bqan}{\begin{eqnarray}}
\newcommand{\eqan}{\end{eqnarray}}
\newcommand{\bqt}{\begin{quote}}
\newcommand{\eqt}{\end{quote}}
\newcommand{\bt}{\begin{tabbing}}
\newcommand{\et}{\end{tabbing}}
\newcommand{\bit}{\begin{itemize}}
\newcommand{\eit}{\end{itemize}}
\newcommand{\ben}{\begin{enumerate}}
\newcommand{\een}{\end{enumerate}}
\newcommand{\beq}{\begin{equation}}
\newcommand{\eeq}{\end{equation}}
\newcommand{\beqw}{\begin{equation*}}
\newcommand{\eeqw}{\end{equation*}}
\newcommand{\bdefi}{\begin{definition}}
\newcommand{\edefi}{\end{definition}}
\newcommand{\bpro}{\begin{proposition}}
\newcommand{\epro}{\end{proposition}}
\newcommand{\blem}{\begin{lemma}}
\newcommand{\elem}{\end{lemma}}
\newcommand{\bco}{\begin{corollary}}
\newcommand{\eco}{\end{corollary}}
\newcommand{\bdes}{\begin{description}}
\newcommand{\edes}{\end{description}}
\newcommand{\bre}{\begin{remark}}
\newcommand{\ere}{\end{remark}}

\newcommand{\eps}{\epsilon}

\newcommand{\tr}{\mathrm{tr}}
\usepackage{authblk}

\def\mds{\medskip}
\def\1{{\mathbf 1}}
\usepackage{xr}
\makeatletter
\newcommand*{\addFileDependency}[1]{
  \typeout{(#1)}
  \@addtofilelist{#1}
  \IfFileExists{#1}{}{\typeout{No file #1.}}
}
\makeatother

\makeatletter
\def\@seccntformat#1{\@ifundefined{#1@cntformat}%
   {\csname the#1\endcsname\quad}
   {\csname #1@cntformat\endcsname}
}
\makeatother

\begin{document}

\title{\bf Factor Multivariate Stochastic Volatility Models of High Dimension\thanks{
The authors are most grateful to Yoshihisa Baba for his very helpful comments and suggestions. B.P. was supported by the Japan Society for the Promotion of Science (Grant Number: 25K16617).
M.A. acknowledges the financial support of the Japan Society for the Promotion of Science (Grant Number: 22K01429, 25K05042).
}}

\medskip

\author[a]{ Benjamin Poignard}

\author[b]{Manabu Asai}

 \affil[a]{\it \footnotesize Faculty of Science and Technology, Keio University and Riken-AIP, Japan. bpoignard@math.keio.ac.jp}
 \affil[b]{\it \footnotesize Faculty of Economics, Soka University, Japan. m-asai@soka.ac.jp}

 \maketitle

\baselineskip = 7mm

\begin{abstract}
Building upon factor decomposition to overcome the curse of dimensionality inherent in multivariate volatility processes, we develop a factor model-based multivariate stochastic volatility (fMSV) framework.
We propose a two-stage estimation procedure for the fMSV model: in the first stage, estimators of the factor model are obtained, and in the second stage, the MSV component is estimated using the estimated common factor variables.
We derive the asymptotic properties of the estimators, taking into account the estimation of the factor variables. The prediction performances are illustrated by finite-sample simulation experiments and applications to portfolio allocation.
\end{abstract}


\section{Introduction}

The generalized autoregressive conditional heteroskedasticity (GARCH) and the stochastic volatility (SV) models are two popular families to specify and analyze the uncertainty of economic and financial time series. 
Multivariate GARCH (MGARCH) and Multivariate SV (MSV) models can capture time-varying covariance structures and are thus useful for forecasting variance-covariance matrices.
Within MGARCH models, the dynamic conditional correlation (DCC) models of \cite{engle2002} and \cite{tse2002}, the BEKK model of \cite{baba1985} and \cite{engle1995}, and their variants are commonly used: see the surveys of \cite{bauwens2006} and \cite{boudt2019}.
The MSV process of \cite{harvey1994} serves as the foundational framework and has been extended in various ways, including the factor model of \cite{chib2006} and the dynamic correlation model of \cite{asai2009}.
For a broader discussion on MSV models and Bayesian Markov chain Monte Carlo (MCMC) techniques for their estimation, see \cite{chib2009} and \cite{kastner2017}.

A major drawback of MGARCH and MSV models is the rapid increase in the number of parameters as the dimension of the observed vector grows. This not only complicates estimation in high dimensions but also degrades predictive performance due to overfitting. While enforcing parsimony in parameterization can mitigate this issue, overly restrictive specifications may fail to capture essential data features. Consequently, there exists a trade-off between model flexibility and parsimony.

A solution to the curse of dimensionality in MGARCH and MSV dynamics is to incorporate a factor-based structure, following the approaches of \cite{engle1990} and \cite{jacquier1999}, respectively.
These factor models for time-varying covariance processes align with asset pricing theory in finance.
\textcolor{black}{While \cite{alexander2001} and \cite{weide2002} develop MGARCH models based on principal component analysis, \cite{barigozzi2017} consider general dynamic factor models with conditional volatilities.}

In the MSV literature, various factor MSV (fMSV) models have been proposed to promote parsimonious parameterizations; see, e.g., \cite{aguilar2000}, \cite{chib2006}, \cite{jacquier1999}, \cite{liesenfeld2003}, \cite{lopes2007}, and \cite{pitt1999}.
Notably, most of these studies rely on Bayesian MCMC methods, with the exception of \cite{liesenfeld2003}, which applies a maximum likelihood approach based on efficient importance sampling. However, their fMSV model includes only a single factor.

\cite{poignard2023_jtsa} developed \textcolor{black}{an} OLS-based technique for estimating high-dimensional MSV models.
Unlike existing MSV models, this approach allows for interactions between log-volatilities of different returns across different time periods, making it more flexible than the basic MSV model of \cite{harvey1994}.
Moreover, since the method is non-Bayesian, it avoids \textcolor{black}{the} pitfalls of MCMC-based techniques, such as prior specification challenges \textcolor{black}{and} numerical instability.
Our main idea is to extend this OLS framework to the estimation of fMSV models. Specifically, we consider the following problem: given $T$ observations of a $p$-dimensional random vector $y_t$, we specify a factor decomposition \textcolor{black}{where} the common factors follow an MSV model, as in \cite{jacquier1999}.
The estimation procedure consists of two stages. First, we estimate the factor model \textcolor{black}{by maximum likelihood under suitable restrictions for identification following the framework of \cite{bai2012}}. Then, using the estimated factor scores, we estimate the MSV model for the common factors. 
\textcolor{black}{Although \cite{mucher2025} recently developed a similar two-stage estimation approach based on \cite{bai2012}, our method differs in the second stage. While \cite{mucher2025} employ the efficient method of moments, which is a computationally intensive technique, our approach offers computational efficiency.}

Our main contributions can be summarized as follows: we propose a factor model-based MSV estimation procedure \textcolor{black}{that does not rely on simulated-based techniques}; we establish consistency results of the estimators; the relevance of the method is illustrated by simulated experiments and real data applications.

The remainder of the paper is organized as follows. In Section \ref{sec:framework_setting}, we describe the fMSV model and the procedure to obtain the forecasts of  covariance matrices. 
Section \ref{sec:estimation_procedure} details the two-stage estimation method. The simulation experiments and the real data-based out-of-sample applications are provided in Section \ref{sec:simulations} and Section \ref{sec:real_data}, respectively. All the asymptotic results and secondary technical details \textcolor{black}{are deferred to the Appendices}. The implementation of the fMSV model is available in the Github repository \url{https://github.com/Benjamin-Poignard/fMSV}.  

\mds

\textbf{\emph{Notations.}} Throughout this paper, we denote the cardinality of a set $E$ by $|E|$. For $\vv \in \Rb^{d}$, the $\ell_p$ norm is $\|\vv\|_p = \big(\sum^{\text{d}}_{k=1} |\vv_k|^p \big)^{1/p}$ for $p > 0$, and $\|\vv\|_{\infty} = \max_i|\vv_i|$. We write $A^\top$ (resp. $\vv^\top$) to denote the transpose of the matrix $A$ (resp. the vector $\vv$). For a symmetric matrix $A$, $\lambda_{\min}(A)$ (resp. $\lambda_{\max}(A)$) is the minimum (resp. maximum) eigenvalue of $A$, and $\text{tr}(A)$ is the trace operator. For a matrix $B$, $\|B\|_s = \lambda^{1/2}_{\max}(B^\top B)$ and \textcolor{black}{$\|B\|_F=\{\text{tr}(B^\top B)\}^{1/2}$} are the spectral and Frobenius norms, respectively. We write $\text{vec}(A)$ to denote the vectorization operator that stacks the columns of $A$ on top of one another into a vector. We denote by $\text{vech}(A)$ the $d(d+1)/2$ vector that stacks the columns of the lower triangular part of $A \in \Rb^{d\times d}$. The matrix $I_d$ denotes the $d$-dimensional identity matrix. We denote by $O \in \Rb^{k \times l}$ the $k \times l$ zero matrix. For two matrices $A,B$, $A \otimes B$ is the Kronecker product. For two matrices $A,B$ of the same dimension, $A \odot B$ is the Hadamard product. For $f: \Rb^{d} \rightarrow \Rb$, we denote by $\nabla f$ the gradient or subgradient of $f$ and by $\nabla^2 f$ the Hessian of $f$. 

\section{fMSV Model} \label{sec:framework_setting}

We consider a $p$-dimensional vectorial stochastic process $(y_t)_{t=1,\cdots,T}$ following the fMSV structure: 
\begin{align}
& y_t = \Lambda f_t + \varepsilon_t, \quad \varepsilon_t \sim iid(0, \Sigma_{\varepsilon}), \label{eq:yt} \\
& f_t = D_t^{1/2} \zeta_t, 
\quad \textcolor{black}{D_t = \mbox{diag}\big(\exp( h_{1,t}), \ldots, \exp( h_{m,t}) \big)},
\quad \zeta_t \sim iid(0, I_m), \label{eq:ft}\\
& h_{t+1} = \mu + \Phi (h_t -\mu) + \eta_t, \quad \eta_t \sim \textcolor{black}{iid}(0, \Sigma_{\eta}), \label{eq:ht} 
\end{align}
where $f_t \in \Rb^m$ is the factor variable, $\Lambda \in \Rb^{p \times m}$ is the loading matrix, $\varepsilon_t = (\varepsilon_{1,t},\ldots,\varepsilon_{p,t})^\top \in \Rb^p$ is a random vector, which is independently and identically distributed (i.i.d.) with \textcolor{black}{a diagonal} covariance matrix $\Sigma_{\varepsilon}$, $h_t = (h_{1,t},\ldots,h_{m,t})^\top \in \Rb^m$ is a vector of log-volatilities, $D_t \in \Rb^{m \times m}$ is a diagonal matrix of volatilities, $\zeta_t = (\zeta_{1,t},\ldots,\zeta_{m,t})^\top \in \Rb^m$ is an i.i.d. random vector with covariance matrix $I_m$, $\mu =(\mu_1,\ldots,\mu_m)^\top \in \Rb^m$, $\Phi \in \Rb^{m \times m}$ \textcolor{black}{is non-diagonal}, $\Sigma_{\eta} \in \Rb^{m \times m}$ is the covariance matrix of $\eta_t$, and $(\varepsilon_t)_{t=1,\cdots,T}$, $(\zeta_t)_{t=1,\cdots,T}$, and $(\eta_t)_{t=1,\cdots,T}$ are mutually independent.
\textcolor{black}{The empirical findings of \cite{andersen2001} supports the Gaussian assumption on $\eta_t$ in the context of realized volatility. Therefore, we assume $\eta_t \sim \Nc_{\Rb^m}(0, \Sigma_{\eta})$.}

We assume that all the eigenvalues of $\Phi$ are strictly smaller than one in modulus, which guarantees that $h_t$ and $y_t$ are strictly stationary processes.
\textcolor{black}{We may relax this condition: see the discussion at the end of Subsection \ref{subsec:msv_estimation}.} 
Note that a non-diagonal $\Phi$ allows for cross-effects from other past variables, while the existing literature usually specifies a diagonal $\Phi$, except for \cite{poignard2023_jtsa}.
\textcolor{black}{We do not consider higher-order specifications for $h_t$, as empirical evidence favors the parsimonious MGARCH(1,1) model over higher-order alternatives, despite their theoretical interest.}

Suitable restrictions for the identification of the factor model can be found in, e.g., Table 1 in \cite{bai2012}.
In addition to restricting the factor model as $\Lambda_{kl}=0$ for $l>k$ and $\Lambda_{kk}=1$ for $k \leq m$, \cite{pitt1999} and \cite{chib2006} heavily rely on Bayesian MCMC \textcolor{black}{for estimation purposes}. This is a key difference with our proposed procedure:
our first stage devoted to the factor model estimation builds upon the work of \textcolor{black}{\cite{bai2012}, which lies within the likelihood framework.  
To ensure the identification of the factor model, we rely on condition IC2 of \cite{bai2012}: 
$$\text{IC2:}\;\;M_f:=\Eb[f_tf^\top_t] \; \text{is diagonal with distinct coefficients}, \;\; p^{-1}\Lambda^\top \Sigma^{-1}_{\varepsilon} \Lambda=I_m.$$ 
By equation (\ref{eq:ft}), $\Eb[f_tf^\top_t] = \Eb[\Eb[D^{1/2}_t\zeta_t\zeta^\top_t D^{1/2}_t|h_t]] = \Eb[D_t] \neq I_m$, which renders identification conditions IC3 and IC5 of \cite{bai2012} -- both imposing $M_f=I_m$ -- inappropriate for our setting. While their condition IC1 leaves $M_f$ unrestricted and their condition IC4 requires $M_f$ to be diagonal, both conditions impose a structure on $\Lambda$, similar to \cite{pitt1999} and \cite{chib2006}. Furthermore, under condition IC2, $\Vb(f_{k,t})\neq \Vb(f_{l,t})$ for $k \neq l$.
Moreover, $\Eb[f_t\varepsilon^\top_t]=\Eb[D^{1/2}_t\zeta_t\varepsilon^\top_t]=O\in \Rb^{m \times p}$. By equation (\ref{eq:yt}), the unconditional variance-covariance $\Vb(y_t)$ of $y_t$ is $\Vb(y_t)=\Lambda M_f\Lambda^\top + \Sigma_{\varepsilon}$. 
Note that \cite{bai2012} assume that the covariance $\Sigma_{\varepsilon}$ of $\varepsilon_t$ is diagonal. The time-varying covariance matrix of $y_t$ conditional on $h_t$ in model (\ref{eq:yt})-(\ref{eq:ht}) is defined by 
\begin{equation}
 H_t \equiv \Vb(y_t|h_t) =\Lambda D_t \Lambda^\top + \Sigma_{\varepsilon}.
\end{equation}}

Our procedure to obtain the forecasts of $H_t$, which will be referred as fMSV hereafter, can be broken down as \textcolor{black}{follows}:

\mds

\noindent \textbf{Stage 1(a).} 
Obtain $(\widehat{\Lambda},\widehat{\Sigma}_\varepsilon)$ via \textcolor{black}{\cite{bai2012} under IC2}.

\mds

\noindent \textbf{\textcolor{black}{Stage 1(b)}.} Given, $(\widehat{\Lambda},\widehat{\Sigma}_\varepsilon)$, \textcolor{black}{compute} the generalized least squares (GLS) estimator of $f_t$ as
\begin{equation}
\textcolor{black}{\forall t,} \; \widehat{f}_t = (\widehat{\Lambda}^\top \widehat{\Sigma}_\varepsilon^{-1} \widehat{\Lambda})^{-1} 
\widehat{\Lambda}^\top \widehat{\Sigma}_\varepsilon^{-1} y_t.
\label{eq:gls}
\end{equation}
See Section 14.7 of \cite{anderson2003introduction} for its derivation.

\mds

\noindent \textbf{Stage 2(a).}  Obtain the forecasts $\widehat{D}_{T+l}$ of $D_{T+l}$ using $\widehat{f}_t$ $(t=1,\ldots,T)$, modifying the work of \cite{poignard2023_jtsa}.

\mds

\noindent \textbf{Stage 2(b).}  Obtain the forecasts $\widehat{H}_{T+l}$ of $H_{T+l}$ as $\widehat{H}_{T+l} = \widehat{\Lambda} \widehat{D}_{T+l} \widehat{\Lambda}^\top + \widehat{\Sigma}_{\varepsilon}$ for $l=1,2,\ldots,H$.

In other words, the estimation of the fMSV model is decomposed into two stages. The first stage concerns the estimation of the factor model parameters, \textcolor{black}{denoted by $\theta_{\text{FM}}$, as described in Subsection \ref{subsec:factor_estimation}. The second stage relates to the estimation of the common-factor MSV model parameters, denoted by $\theta_{\text{MSV}}$, as detailed in Subsection \ref{subsec:msv_estimation}. Throughout this work, we assume that the number of factors $m$ is fixed and known. The value $m$ can be selected using, e.g., the method of \cite{bai2002} or \cite{onatski2010}. }
Note that \cite{poignard2023_jtsa} consider the minimum mean square linear estimator (MMSLE) to obtain the forecasts of the log-volatilities.

One point in the procedure of \cite{poignard2023_jtsa} requires a modification. 
The latter work relies on the logarithm of the squared return, $\log y_{i,t}^2$ $(i=1,\ldots,p)$, \textcolor{black}{which excludes the case $y_{i,t}=0$}.
In the context of the fMSV model, \textcolor{black}{as it allows $f_{i,t}=0$, we use $\log (f_{i,t}^2 + c_i) - c_i/(f_{i,t}^2 + c_i)$, where $c_i = 10^{-4} s_i^2$ and $s_i^2$ is the sample mean of $f_{i,t}^2$.
As discussed in Section 9.3 of \cite{fuller1996}, we may use $c_i = 0.02 s_i^2$ if we consider more robust estimation against non-normality.
However, our estimation technique does not require normal approximation for $\log f_{i,t}^2$.
Instead, we need a slight change for $\log f_{i,t}^2$ to allow the case $f_{i,t}=0$.
Our simulation and empirical results are robust to the setting $c_i = 10^{-5} s_i^2$. To perform \textbf{Stage 2}, $f_t$ will be replaced with $\widehat{f}_t$.}
The next section details the two-stage estimation method.

\section{Estimation}\label{sec:estimation_procedure}

\textcolor{black}{We denote by $\theta=(\theta^{\top}_{\text{FM}},\theta^{\top}_{\text{MSV}})^\top$ the vector of the fMSV model parameters, where $\theta_{\text{FM}}$ and $\theta_{\text{MSV}}$ correspond to the factor model parameters and stochastic volatility parameters, respectively. The estimation of $\theta_{\text{FM}}$ is detailed in Subsection \ref{subsec:factor_estimation}. Subsection \ref{subsec:msv_estimation} describes the estimation of $\theta_{\textcolor{black}{\text{MSV}}}$.}

\subsection{Estimation of the factor model}\label{subsec:factor_estimation}

\textcolor{black}{The factor model parameter vector of equation (\ref{eq:yt}) is defined by $\theta_{\text{FM}}= (\theta^\top_{\Lambda},\theta^\top_{\Sigma_{\varepsilon}})^\top$, with $\theta_{\Lambda} = \text{vec}(\Lambda) \in \Rb^{pm},\; \theta_{\Sigma_{\varepsilon}}=\text{vech}(\Sigma_{\varepsilon}) \in \Rb^{p(p+1)/2}$. It belongs to the parameter space $\Theta_{\text{FM}} := \Theta_{\Lambda}\times\Theta_{\Sigma_{\varepsilon}} \subseteq \Rb^{p m}\times \Rb^{p(p+1)/2 }$.} Estimation methods for factor models can be divided into two: likelihood-based ones (see, e.g., \cite{bai2012,bai2016,bailiao2016,poignard2020,poignard_terada_2025}) and PCA-based ones (see, e.g., \cite{stock2002a,fan2011,fan2013}). Likelihood-based methods estimate $\Lambda,\Sigma_{\varepsilon}$ \textcolor{black}{and $M_f$}, and in a second stage, estimate the factors by, e.g., \textcolor{black}{the generalized least squares method}. PCA procedures usually estimate both $\Lambda$ and $f_t,t=1,\ldots,T$, and then obtain the estimator of $\Sigma_{\varepsilon}$. 

\textcolor{black}{Throughout this paper, we will employ the maximum likelihood method under condition IC2 of \cite{bai2012} and under the diagonal assumption on $\Sigma_{\varepsilon}$, so that $\theta_{\Sigma_{\varepsilon}} \in \Rb^p$. Under IC2, $M_f:=\Eb[f_t f^\top_t] \in \Rb^{m \times m}$ is diagonal with distinct coefficients, and $p^{-1}\Lambda^\top \Sigma^{-1}_{\varepsilon}\Lambda=I_m$. 
The factor model parameters $\Lambda,M_f,\Sigma_{\varepsilon}$ are estimated by Gaussian MLE, and the corresponding asymptotic properties are provided by \cite{bai2012}. Now, following Section 8 of \cite{bai2012}, the maximization of the likelihood function is performed by the EM algorithm and under condition IC3 given by: 
$$\text{IC3:}\;\; M_f=I_m,\;\; \text{and}\;\; p^{-1}\Lambda^\top\Sigma^{-1}_{\varepsilon}\Lambda\; \text{diagonal with distinct elements (arranged in decreasing order)}.$$
This implies that the EM algorithm will provide a solution for the variance $\Vb(y_t)= \Lambda\Lambda^\top+\Sigma_{\varepsilon}$, where the corresponding solution $(\widehat{\Lambda}^{\text{IC3}},\widehat{\Sigma}^{\text{IC3}}_{\varepsilon})$ satisfies IC3. More precisely, $\widehat{\Lambda}^{\text{IC3}} = \widetilde{\Lambda}V$, $\widehat{\Sigma}^{\text{IC3}}_{\varepsilon}=\widetilde{\Sigma}_{\varepsilon}$, with $\widetilde{\Lambda}, \widetilde{\Sigma}_{\varepsilon}$ being the values obtained at the final iteration of the EM algorithm, and $V$ is the orthogonal matrix containing the eigenvectors of $p^{-1}\widetilde{\Lambda}^\top\widetilde{\Sigma}^{-1}_{\varepsilon}\widetilde{\Lambda}$ corresponding to decreasing eigenvalues. To obtain the estimators of $\Lambda,\Sigma_{\varepsilon}$ under IC2, we compute $\widehat{\Lambda}^{\text{IC2}}=\widehat{\Lambda}^{\text{IC3}}\big(p^{-1}\widehat{\Lambda}^{\text{IC3}\top}\big(\widehat{\Sigma}^{\text{IC3}}_{\varepsilon}\big)^{-1}\widehat{\Lambda}^{\text{IC3}}\big)^{-1/2}$, and $\widehat{M}^{\text{IC2}}_{f}=p^{-1}\widehat{\Lambda}^{\text{IC3}\top}\big(\widehat{\Sigma}^{\text{IC3}}_{\varepsilon}\big)^{-1}\widehat{\Lambda}^{\text{IC3}}$. Under the diagonal assumption on $\Sigma_{\varepsilon}$, we have $\widehat{\Sigma}^{\text{IC2}}_{\varepsilon}=\widehat{\Sigma}^{\text{IC3}}_{\varepsilon}=\widetilde{\Sigma}_{\varepsilon}$.
Using the estimator $\widehat{\theta}_{\text{FM}}=(\widehat{\theta}^\top_\Lambda,\widehat{\theta}^\top_{\Sigma_{\varepsilon}})^\top$ of $\theta_{\text{FM}}$ under IC2, we compute the GLS estimator $\widehat{f}_t = (\widehat{\Lambda}^\top \widehat{\Sigma}_\varepsilon^{-1} \widehat{\Lambda})^{-1} 
\widehat{\Lambda}^\top \widehat{\Sigma}_\varepsilon^{-1} y_t$ of $f_t$, for all $t$. By Theorem 6.1 of \cite{bai2012}, under large $p$, $\widehat{f}_t$ is asymptotically normally distributed.} 

\textcolor{black}{In Appendix~\ref{appendix_factor_estim}, we show the uniform consistency of $\widehat{f}_t$ under suitable moment conditions and scale conditions on $(p,T)$. This property will facilitate the proofs in the asymptotic analysis of the MSV parameter $\theta_{\text{MSV}}$ detailed in Appendix~\ref{appendix_asymptotic_prop}: it will allow us to control the estimation error resulting from the estimation of the factor variable} when plugged in the loss functions employed for the estimation of $\theta_{\text{MSV}}$. 

\textcolor{black}{As for the choice of the number of factors $m$, one may rely on} \cite{onatski2010}'s method based on an eigenvalue threshold estimator using the sample variance-covariance matrix of $y_t$. To efficiently reduce dimensionality, it is generally preferable to have a small number of factors, typically \textcolor{black}{$m\leq 5$}. As noted in \cite{deNard2021}, there is no consensual selection procedure of the number of factors, making it important to assess the sensitivity of the factor-model-based MSV performance with respect to $m$. Therefore, in our experiments, we set $m \in \{1,2,3\textcolor{black}{, 4, 5}\}$.

\subsection{Estimation of the MSV parameters}\label{subsec:msv_estimation}

\textcolor{black}{We now consider the estimation of $\theta_{\text{MSV}}$, which can be partitioned into two subvectors $\theta_1$ and $\theta_2$ corresponding to the two-stage procedure detailed hereafter in \textbf{Step 1} and \textbf{Step 2}. }

\medskip


\medskip

By transforming $f_t$ to $f_t^\ell =(\log f_{1,t}^2,\ldots,\log f_{m,t}^2)^\top$, we obtain the state space form:
\begin{equation}
\textcolor{black}{\forall t,}\; f_t^\ell = \nu + \alpha_t + \xi_t, \quad
\alpha_{t+1} = \Phi \alpha_t + \eta_t, \label{eq:ssf}
\end{equation}
where $\nu=(\nu_1,\ldots,\nu_m)^\top$, $\xi_t = (\xi_{1,t},\ldots,\xi_{m,t})^\top$, 
and $\alpha_t = h_t - \mu$ with $\nu_i = \mu_i + \Eb[\log \zeta_{i,t}^2]$ and $\xi_{i,t} = \log \zeta_{i,t}^2 - \Eb[\log \zeta_{i,t}^2]$.
We may apply the Kalman filter to estimate the \textcolor{black}{parameters} in (\ref{eq:ssf}), as suggested by Harvey et al. (1994).
\textcolor{black}{However, since we assume that $\Phi$ is non-diagonal, estimation by the Kalman filter can be a challenging task, even in a low-dimensional setting.
Therefore}, in the same spirit as in \cite{poignard2023_jtsa}, we estimate $\theta_{\text{MSV}}$ based on the following steps:
\mds

\noindent \textbf{Step 1.} 
Since $f_t^\ell$ is the sum of a VAR(1) process, $\alpha_t$, and an i.i.d. error $\zeta_t$, case (i) in Section 3 of \cite{granger1976} suggests that $f_t^\ell$ has a VARMA(1,1) representation, which leads a VAR($\infty$) model that will be employed to approximate the unobserved error. More precisely, \textcolor{black}{consider $x_t = (\log (f_{1,t}^2 + c_1) - c_1/(f_{1,t}^2 + c_1),\ldots, \log (f_{m,t}^2 + c_m) - c_m/(f_{m,t}^2 + c_m))^\top$, as detailed in the previous section, in order to obtain the state space form,
\begin{equation}
x_t = \nu^* + \alpha_t + \xi_t^*, \quad
\alpha_{t+1} = \Phi \alpha_t + \eta_t, \label{eq:ssf2}
\end{equation}
where $\nu^*=(\nu_1^*,\ldots,\nu_m^*)^\top$, $\xi_t^* = (\xi_{1,t}^*,\ldots,\xi_{m,t}^*)^\top$,
\begin{align*}
\nu_i^* &= \mu_i + \Eb\left[\log (\zeta_{i,t}^2 + c_i e^{-h_{i,t}}) 
- (\zeta_t^2 + c_i e^{-h_{i,t}})^{-1} c_i e^{-h_{i,t}} \right], \\
\xi_{i,t}^* &= \log (\zeta_{i,t}^2 + c_i e^{-h_{i,t}}) 
- (\zeta_{it}^2 + c_i e^{-h_{i,t}})^{-1} c_i e^{-h_{i,t}} - \nu_i^* + \mu_i.
\end{align*}
Since we set $c_i = 10^{-4}s_i^2$, where $s_i^2$ is the sample mean of $f_{it}^2$, $c_i e^{-h_{i,t}}$ is a negligible quantity compared to $\zeta_{i,t}^2$. Hence we treat $\xi_t^*$ as a white noise process with mean zero and covariance matrix $\Sigma_{\xi} = \Eb[\xi_t \xi_t^\top]$, that is $\xi_t^* \sim WN(0,\Sigma_{\xi})$. It holds when $c_i=0$.
As detailed in Appendix \ref{appendix_varma}, we have the VARMA(1,1) representation for $x_t$ given by
\begin{equation}
x_t = \nu^* + \Phi(x_{t-1} -\nu^*) + u_t - \Upsilon u_{t-1}, \quad u_t \sim {WN}(0,\Sigma_u), 
\label{eq:varma}
\end{equation}
where $\Upsilon$ and $\Sigma_u$ are defined by Appendix \ref{appendix_varma}.}
\textcolor{black}{From the representation, we deduce the VAR($\infty$) form of $x_t$ that we approximate by the VAR($q$) model} $x_t = \psi^* +\sum_{k=1}^q \Psi_k x_{t-k} + u^{(q)}_t$, with $u^{(q)}_t = u_t +\sum_{k>q}\Psi_k x_{t-k}$. 
The term $\sum_{k>q}\Psi_kx_{t-k}$ \textcolor{black}{corresponds to the remainder term resulting} from the truncation of the VAR($\infty$). 
Under suitable parameter decay conditions, it can be shown that $\Eb[\|u^{(q)}_t-u_t\|^r_2]$ is negligible when $q$ is large enough, assuming the existence of the $r$-th moment of $u_t$, as in \cite{chang2002} for univariate processes and \cite{chang2006} for multivariate processes. 
In our framework, the number of factors $m$ is fixed, \textcolor{black}{$q$ may potentially diverge with the sample.} We apply a \textcolor{black}{penalty function} to suitably discard the irrelevant past variables. It aims to estimate the sparse support and identify an optimal $q$ \textcolor{black}{lag} from which the past variables can be considered as negligible. 
\textcolor{black}{Define $\theta_1 := \text{vec}(\underline{\Psi}) \in \Theta_{1} \subseteq \Rb^{d_1}$, $d_1:=m+qm^2$, with $\underline{\Psi} =(\psi^*,\Psi_1,\cdots,\Psi_q)$ the parameters of the VAR($q$) model and $\Theta_1$ the parameter space of $\theta_1$. Define} the non-penalized loss $\Lb_T: \Rb^{m(q+1)T} \times \Theta_{1} \rightarrow \Rb$. \textcolor{black}{Let $\mathbf{F} = (F^\top_1,\ldots,F^\top_T)^\top \in \Rb^{m(q+1)T}$, $F_t = (f^\top_t,\ldots,f^\top_{t-q})^\top \in \Rb^{m(q+1)}$ and let $f_0,\ldots,f_{1-q}$ be initial values (chosen equal to zero). \textcolor{black}{Define $Z_{q,t-1}=(1,x^\top_{t-1},\cdots,x^\top_{t-q})^\top\in \Rb^{1+mq}$, $x_{j,t} = \log (f_{j,t}^2 + c_j) - c_j/(f_{j,t}^2 + c_j)$ for $j=1,\ldots,m$.} When replacing $f_t$ by its estimator $\widehat{f}_t$, \textcolor{black}{the variables} $\mathbf{F}, F_t, \textcolor{black}{Z_{q,t-1},x_t}$ are denoted by $\widehat{\mathbf{F}}$, $\widehat{F}_t,\textcolor{black}{\widehat{Z}_{q,t-1},\widehat{x}_t}$.} The loss $\Lb_T(\widehat{\mathbf{F}};\theta)$ is associated to a continuous function $\ell : \Rb^{m(q+1)} \times \Theta_{1} \rightarrow \Rb$ that can be written as
\begin{equation*}
\Lb_T(\widehat{\mathbf{F}};\theta_1) := \overset{T}{\underset{t=1}{\sum}} \frac{1}{2}\|\textcolor{black}{\widehat{x}}_t-\psi^*-\overset{q}{\underset{k=1}{\sum}}\Psi_k\textcolor{black}{\widehat{x}}_{t-k}\|^2_2 = \overset{T}{\underset{t=1}{\sum}} \frac{1}{2}\|\textcolor{black}{\widehat{x}}_t-\underline{\Psi} \textcolor{black}{\widehat{Z}}_{q,t-1}\|^2_2 =: \overset{T}{\underset{t=1}{\sum}} \ell(\widehat{F}_{t};\theta_1).
\end{equation*}
The true parameter is denoted by $\theta_{01}$ and is assumed sparse with true support $s := \text{card}(\Sc)$, with $\Sc:=\{k=1,\cdots,d_1: \theta_{01,k}\neq 0\}$. To estimate the latter support, we rely on the penalized M-estimation criterion given by
\begin{equation} \label{obj_crit_first_step}
\widehat{\theta}_1 = \underset{\theta_1\in \Theta_{1}}{\arg \; \min} \;  \Lb_T(\widehat{\mathbf{F}};\theta_1) + T\lambda_T \sum^{\textcolor{black}{d_1}}_{k=1}\tau(\widetilde{\theta}_{1,k}) |\theta_{1,k}|.
\end{equation}
The penalty function \textcolor{black}{is} the convex adaptive LASSO penalty of \cite{zou2006adaptive}, where $\lambda_T$ is the tuning parameter and $\tau(\widetilde{\theta}_{1,k})$ is a random weight defined by $\tau(\widetilde{\theta}_{1,k})=|\widetilde{\theta}_{1,k}|^{-\gamma}$, with $\widetilde{\theta}_{1,k}$ a consistent first step estimator. It is defined as the solution:
\begin{equation} \label{obj_crit_first_estimator}
\widetilde{\theta}_1 = \underset{\theta_1 \in \Theta_{1}}{\arg \; \min} \; \Lb_T(\widehat{\mathbf{F}};\theta_{\color{black}{1}}).
\end{equation}
Theorems \ref{bound_proba_first_step_estimator} and \ref{bound_prob_asym} in Appendix \ref{appendix_asymptotic_prop} establish the consistency of $\widetilde{\theta}_1$ and $\widehat{\theta}_1$, respectively. Theorem \ref{sparsistency} provides the \textcolor{black}{``sparsistency'' property of $\widehat{\theta}_1$, that is the recovery of the true zero entries with probability tending to one}: see \cite{lam2009} for further details on sparsistency. 

In our experiments, the selection of an optimal $\lambda_T$ is performed by \textcolor{black}{a cross-validation method that accounts for the dependent nature of the data. We will set $q=10$ when solving (\ref{obj_crit_first_estimator}) and (\ref{obj_crit_first_step}). In the latter problem, we set $\gamma=1$. All the implementation details are provided in Appendix \ref{appendix_sec:implementation_MSV}}.

\noindent\textbf{Step 2.} Define the approximated error $\widehat{u}^{(q)}_t = x_t - \widehat{\psi}^*-\sum^q_{k=1}\widehat{\Psi}_k x_{t-k}$ obtained in \textbf{Step 1}. Given $\widehat{\theta}_1$, we consider the regression
\[
x_t = c^* + \Phi x_{t-1}+ \Xi \widehat{u}^{(q)}_{t-1} + v_t,
\]
where the vector of parameters is $\theta_2:=(c^{*\top},\theta^\top_{\Phi},\theta^\top_{\Xi})^\top \in \Theta_2 \subseteq \Rb^{d_2}$, $d_2 := m(1+2m)$, and since we replace $u_{t-1}$ by $\widehat{u}^{(q)}_{t-1}$, $(v_t)$ is the error term for this auxiliary regression. 

\textcolor{black}{The estimation of $\theta_2$} relies on the second step loss $\Gb_T : \Rb^{m(q+1)T} \times \Theta_1 \times \Theta_2 \rightarrow \Rb$, where $\Gb_T(\widehat{\mathbf{F}};\widehat{\theta}_1;\theta_2)$ is the loss associated to a continuous function $g: \Rb^{m(q+1)} \times \Theta_1 \times \Theta_2 \rightarrow \Rb$. \textcolor{black}{As in \textbf{Step 1}, we replace the latent $x_t$ with $\widehat{x}_t$, so that the second step loss is}
\begin{equation*}
\Gb_T(\widehat{\mathbf{F}};\widehat{\theta}_1;\theta_2) =  \overset{T}{\underset{t=1}{\sum}}\frac{1}{2}\|\textcolor{black}{\widehat{x}}_t-c^*-\Phi \textcolor{black}{\widehat{x}}_{\textcolor{black}{t-1}}-\Xi \textcolor{black}{\widehat{\overline{u}}}^{(q)}_{t-1}\|^2_2  := \overset{T}{\underset{t=1}{\sum}} \frac{1}{2}\|\textcolor{black}{\widehat{x}}_t-\Gamma \widehat{K}_{t-1}(\widehat{\theta}_1)\|^2_2=:
\overset{T}{\underset{t=1}{\sum}} g(\widehat{F}_{t};\widehat{\theta}_1;\theta_2),
\end{equation*}
where $\widehat{K}_{t-1}(\widehat{\theta}_1) = (1,\textcolor{black}{\widehat{x}}^\top_{\textcolor{black}{t-1}},\textcolor{black}{\widehat{\overline{u}}}^{(q)\top}_{t-1})^\top \in \Rb^{1+2m}$ and $\Gamma = (c^*,\Phi,\Xi) \in \Rb^{m\times (1+2m)}$. The dependence on the first step estimator is through \textcolor{black}{$\widehat{\overline{u}}^{(q)}_t=\widehat{x}_t - \widehat{\psi}^*-\sum^q_{k=1}\widehat{\Psi}_k \widehat{x}_{t-k}$}. The problem of interest is
\begin{equation}\label{obj_crit_second}
\widehat{\theta}_2 = \underset{\theta_2\in\Theta_2}{\arg\;\min}\;\Gb_T(\widehat{\mathbf{F}};\widehat{\theta}_1;\theta_2).
\end{equation}
Theorem \ref{bound_proba_second_step_estimator} in Appendix \ref{appendix_asymptotic_prop} establishes the consistency of $\widehat{\theta}_2$.

\noindent\textbf{Step 3.} From the decomposition of the unconditional variance-covariance of $x_t$ given by
\begin{equation}
\Sigma_x = \Sigma_{\alpha} + \textcolor{black}{\Sigma_{\xi}}, 
\label{eq:deco}
\end{equation}
where $\Sigma_x = \Eb[(x_t-\gamma)(x_t-\gamma)^\top]$, $\Sigma_{\xi} = \Eb[\xi_t \xi_t^\top]$, and $\Sigma_{\alpha} = \Eb[\alpha_t \alpha_t^\top]$, the estimators of $\Sigma_{\zeta},\Sigma_{\alpha}$ are deduced as
\begin{equation}\label{adhoc_estim}
\widehat{\Sigma}_{\textcolor{black}{\xi}} = r S_x, \quad \widehat{\Sigma}_{\alpha} = (1-r) S_x,
\end{equation}
where $0 < r < 1$ and $S_x$ is the sample covariance matrix of \textcolor{black}{$\widehat{x}_t$}. This ad hoc method aims to treat the positive-definiteness of the estimators and to deal with the high-dimensional issue of $\widehat{\Sigma}_{\textcolor{black}{\xi}}$.
While we consider a naive decomposition based on equation (\ref{eq:deco}) for the former, we set $r = (\pi^2/2)(m^{-1} \tr (S_{x}))^{-1}$ in (\ref{adhoc_estim}).
Here, $ \pi^2/2$ is the value of $\Eb[\textcolor{black}{\xi}_{i,t}^2]$ when \textcolor{black}{$\zeta_{i,t}$} follows the standard normal distribution. 
The ad hoc estimators yield $\tr (\widehat{\Sigma}_{\textcolor{black}{\xi}}) = r \tr(S_x) = m \pi^2/2$ and $\tr (\widehat{\Sigma}_{\alpha}) =  \tr(S_x) - m \pi^2/2$.
Using this approach, we are able to estimate $\tr({\Sigma}_{\textcolor{black}{\xi}})$ and $\tr({\Sigma}_{\alpha})$ with accuracy and consistency, respectively.
More importantly, the computational cost is negligible, compared with alternative estimators (e.g., the GMM type method) that would require a numerical optimization with constraints on the positive-definiteness of $S_{\xi}$ and $S_x-S_{\xi}$, where $S_{\xi}$ is an alternative estimator. 

\textcolor{black}{We dot not impose the restriction $\rho(\Phi)<1$ apriori.
For economic and financial data, we need to control the case where several eigenvalues
of $\widehat{\Phi}$ exceed 1, using the random walk model.
By replacing such eigenvalues of $\widehat{\Phi}$ by one’s, and setting corresponding values of $\widehat{c}$ to be zero, we may construct a state space model for $x_t$ to obtain out-of-sample and in-sample forecasts by the Kalman filter and smoother, respectively.}

\textcolor{black}{
In Appendix \ref{comp_cost}, we compare the computational efficiency of our two-stage procedure with the MCMC approach proposed by \cite{kastner2017}. The empirical results show that our method is computationally more efficient than the MCMC benchmark.
}


\section{Simulations}\label{sec:simulations}

This section presents numerical experiments on the in-sample comparison between a true $p \times p$ dynamic variance-covariance $(H_t)$ and the conditional forecast $(\widehat{H}_t)$ obtained from a specific model. 

\subsection{Data generating processes}

The simulated $p$-dimensional observations $(y_t)$ satisfy:
\begin{equation*}
\forall t =1,\ldots,T,\; y_t = H^{1/2}_t \eta_t, \; \text{with} \; \forall k=1,\ldots,p, \eta_{k,t} = \sqrt{\gamma-2}\,\nu_{k,t}/\sqrt{\delta}, \; \text{where} \; \nu_{k,t} \sim t(\delta),
\end{equation*}
where $t(\gamma)$ is a centered Student distribution with $\delta$ degrees of freedom, $\delta >2$. Hereafter, $\delta=3$. The true variance-covariance $(H_t)$ is based on two data generating processes (DGPs):
\begin{itemize}
    \item[(i)] \textbf{DGP 1}: $\forall t=1,\ldots,T, \; H_t = \Gamma + A {\color{black}{y_{t-1}y_{t-1}^\top}} A^\top + B H_{t-1} B^\top$, 
    with $\Gamma$ positive-definite and $A, B$ diagonal $p \times p$ matrices. This is a diagonal-BEKK-based process. Denoting by $\Uc(a,b)$ the uniform distribution on $[a,b]$, $\Gamma$ is generated as follows: we simulate a matrix $K \in \Rb^{p\times p}$ with entries in the uniform distribution \textcolor{black}{$\Uc(-0.2,0.2)$}, set $\Gamma = K\,K^\top/p$ and then replace its diagonal elements by coefficients simulated in $\Uc(0.005,0.025)$; to ensure that the resulting matrix is positive-definite, if the simulated $\Gamma$ satisfies $\lambda_{\min}(\Gamma)<0.01$, we apply $\Gamma = \Gamma + (\zeta+|\lambda_{\min}(\Gamma)|)I_p$, where $\zeta$ is the first value in $\{0.005,0.01,0.015,\ldots\}$ such that $\lambda_{\min}(\Gamma)>0.01$. The diagonal elements of $A$ (resp $B$) are drawn in \textcolor{black}{$\Uc(0.1,0.4)$} (resp. \textcolor{black}{$\Uc(0.5,0.8)$}) under the stationarity condition $\max_{1\leq k \leq p}(A^2_{kk}+B^2_{kk})<1$, following Proposition 2.7 of \cite{engle1995}. 
    \item[(ii)] \textbf{DGP 2}: $\forall t=1,\ldots,T, \; H_t = \Gamma + \overset{m^*}{\underset{j=1}{\sum}} \lambda_{j,t} \beta_j\beta^\top_j$,
    where $\lambda_{j,t}$ is the conditional variance of the $j$-th factor \textcolor{black}{denoted here by $r_{j,t}$}, $m^*$ is the true number of factors set as $m^*=2$ and $\Gamma$ is a positive-definite non-diagonal matrix. This is a factor GARCH-based process. The $m^*$ conditional variances $\lambda_{jt},j=1,2$ follow a GARCH(1,1) univariate processes, where the factors \textcolor{black}{$r_{j,t}$} are drawn in $\Nc_{\Rb}(0,\lambda_{jt})$, and the elements of $\Gamma$ are generated as in DGP 1. To generate the $m^*$ univariate GARCH(1,1) processes $\sigma^2_{j,t} = \varsigma_j + \kappa_j \textcolor{black}{r}^2_{j,t-1} + \tau_j \sigma^2_{j,t-1}$, we choose randomly the corresponding $3m^*$ parameters, where we simulate $\varsigma_j \sim \Uc(0.005,0.01)$, $\kappa_j \sim \Uc(0.05,0.15)$ and $\tau_j \sim \Uc(0.7,0.9)$, under the stationarity constraint $\kappa_j + \tau_j < 1$ for any $j$. Finally, the $p$-dimensional vectors $\beta_j$ are generated in \textcolor{black}{$\Uc(-1,1)$}. More details on Factor GARCH models can be found in Section 2.1 of \cite{bauwens2006}.
\end{itemize}
We set the sample size $T=2000$ and dimension \textcolor{black}{$p=20, 100, 500$} and we consider $100$ independent batches for $(y_t)$. Once a series is simulated, we estimate the model using the full sample and compare this estimator with the true variance-covariance to measure its accuracy.

\subsection{Measures for statistical accuracy}\label{sec:loss_accuracy}

\textcolor{black}{We measure the in-sample statistical accuracy of the variance-covariance models based on the following matrix distances proposed by \cite{laurent2012}:
\begin{equation}
\begin{split}
\text{D}_E(H_t,\widehat{H}_t) &= \text{vech}(H_t-\widehat{H}_t)^\top\text{vech}(H_t-\widehat{H}_t), \\
\text{D}_F(H_t,\widehat{H}_t) &= \text{trace}((H_t-\widehat{H}_t)^\top(H_t-\widehat{H}_t)),\\
\text{D}_S(H_t,\widehat{H}_t) &= \text{trace}(\widehat{H}^{-1}_tH_t)-\log(|\widehat{H}^{-1}_tH_t|)-p,\\
\text{D}_b(H_t,\widehat{H}_t) &= \frac{1}{b(b-1)}\text{trace}(H^b_t-\widehat{H}^b_t) - \frac{1}{b-1}\text{trace}(\widehat{H}^{b-1}_t(H_t-\widehat{H}_t)), \, b \geq 3,
\end{split}
\label{ranking_metrics}
\end{equation}
where $H_t$ is the true variance-covariance and $\widehat{H}_t$ is the estimated variance-covariance matrix at time $t$.
$\text{D}_E$ and $\text{D}_F$ are the Euclidean distance, the squared Frobenius norm, respectively, and are closely related. $\text{D}_E$ evaluates how close the individual covariance entries are, whereas $\text{D}_F$ is a matrix-based MSE distance. $\text{D}_S$ is the Stein loss: it is the scale-invariant Gaussian likelihood, which is asymmetric with respect to over and/or under predictions, where under predictions are significantly penalized. In the same spirit, the measure $\text{D}_b$ proposed by \cite{laurent2012} is also asymmetric with respect to over and/or under predictions, but significantly penalizes over predictions. The coefficient of asymmetry $b$ is set as $b=3$.}

\subsection{Competing models}\label{sec:competing_models}

Equipped with these DGPs, we generate the observations $y_t,t=1,\ldots,T$ and estimate the variance-covariance using the models:
\begin{itemize}
    \item[(i)] \textbf{DCC}: the scalar DCC of \cite{engle2002} with GARCH(1,1) dynamics for the univariate conditional marginals. The estimation is carried out by a standard two-step Gaussian QML, following, e.g., Section 3.2 of \cite{bauwens2006}, with full likelihood when $p\leq 100$ and composite likelihood when $p>100$. More details on the DCC and its estimation, in particular the composite likelihood method, are provided in Appendix~\ref{dcc_model}.
    \item[(ii)] \textbf{sBEKK}:  the scalar BEKK model with variance-targeting. The estimation is carried out by Gaussian QML. Full likelihood and composite likelihood are employed as in the estimation of the scalar DCC model. More details on the scalar BEKK are provided in Appendix~\ref{bekk_process}.
    \item[(iii)] \textbf{$\text{fGARCH}_m$}: the factor GARCH model of \cite{alexander2002} with $m$ factors. The variance-covariance matrix process $(H_t)$ is defined as $H_t = \Lambda F_t \Lambda^\top + \Sigma_{\varepsilon}, t =1,\ldots,T$, where $\Lambda \in \Rb^{p\times m}$ is the matrix of factor loadings, $\Sigma_{\varepsilon} \in \Rb^{p \times p}$ is diagonal, and $F_t$ is a diagonal matrix with diagonal elements corresponding to the GARCH(1,1) variances of the estimated factors $\widehat{f}_t$. The factor loadings and $\Sigma_{\varepsilon}$ are obtained by PCA. More precisely, the estimator of the factor loading matrix $\widehat{\Lambda}$ is equal to the eigenvectors of $\mathbf{Y}^\top\mathbf{Y}/T$, corresponding to its $m$ largest eigenvalues, with $\mathbf{Y}\in \Rb^{T \times p}$ the data matrix with $t$-th row $y^\top_t$. This allows to filter the factors $\widehat{\mathbf{F}} = \mathbf{Y}^\top\widehat{\Lambda}/p$, with $\widehat{\mathbf{F}} \in \Rb^{T \times m}$ the matrix of the estimated factors with $t$-th row $\widehat{f}^\top_t$. Equipped with these estimators, define $\widehat{\varepsilon}_t = y_t - \widehat{\Lambda}\widehat{f}_t$. This is a method employed in, e.g., \cite{stock2002a}. The estimator $\widehat{\Sigma}_{\varepsilon}$ of the variance-covariance of the idiosyncratic error variables is given by $\widehat{\Sigma}_{\varepsilon} = \text{diag}(\sum^T_{t=1}\widehat{\varepsilon}_t\widehat{\varepsilon}^\top_t/T)$.
    \item[(iv)] \textbf{$\text{fMSV}_m$}: the fMSV model with $m$ factors, where $(\Lambda,\Sigma_{\varepsilon})$ are estimated under \textcolor{black}{the condition for identification IC2 of \cite{bai2012}}. 
\end{itemize}
\textcolor{black}{We set $m=1,2,3, \textcolor{black}{4,5},$ for $\text{fGARCH}_m$, $\text{fMSV}_m$ to assess the sensitivity of the method with respect to the factor dimension. }

\subsection{Results}

\textcolor{black}{Tables \ref{simulations_results_DGP1}–\ref{simulations_results_DGP2} report the in-sample accuracy measures $\text{D}_E,\text{D}_F,\text{D}_S,\text{D}_3$, averaged over 100 independent simulation batches, for DGP 1 and DGP 2, respectively. These matrix distances emphasize different aspects of covariance forecast accuracy and therefore highlight distinct strengths and weaknesses of the competing models.
The distances $\text{D}_E$ and $\text{D}_F$ measure entry-wise covariance accuracy. Under DGP 1, the sBEKK model performs best according to these criteria, indicating an excellent fit of individual covariance entries. The proposed fMSV model becomes increasingly competitive as the dimension grows. However, the sBEKK model exhibits substantial distortions in the inverse covariance structure, as emphasized by the Stein loss $\text{D}_S$, for which the DCC model (resp. factor-based models) achieves the best performance when $p \leq 100$ (resp. $p=500$). In terms of the asymmetric $\text{D}_3$, both sBEKK and fMSV avoid systematic volatility overestimation, whereas the fGARCH specifications tend to inflate dominant eigenvalues.
Under DGP 2 (Table \ref{simulations_results_DGP2}), the sBEKK model becomes severely misspecified across all metrics. In contrast, factor-based models such as fMSV and fGARCH provide a substantially better fit with the true variance-covariance, particularly when $m \geq 2$. In terms of $\text{D}_S$, both fGARCH and fMSV models with $m \geq 2$ deliver the best performances. Moreover, the $m=2$ factor specifications yield the smallest values of $\text{D}_3$, indicating superior control of volatility overprediction. Interestingly, when $m=1$, both fGARCH and fMSV perform poorly across all criteria, underscoring the importance of selecting a sufficiently rich factor structure to capture the key features of the data.
Overall, the fMSV model exhibits the most balanced performance across the different accuracy measures and DGPs.}

\begin{table}[h]
\caption{\textcolor{black}{In-sample accuracy based on $\text{D}_E, \text{D}_F,\text{D}_S,\text{D}_3$ measures with $T=2000$, averaged over $100$ batches, under DGP 1.}}\label{simulations_results_DGP1}
\scalebox{0.85}{{\color{black}\begin{tabular}{c|cccc|cccc|cccc}\hline\hline
&  &  & & & & & & & & & & \\
& \multicolumn{4}{c|}{\textbf{$p=20$}} & \multicolumn{4}{c|}{\textbf{$p=100$}} & \multicolumn{4}{c}{\textbf{$p=500$}} \\
& $\text{D}_E$ & $\text{D}_F$ & 
$\text{D}_S$ & $\text{D}_3$ & 
$\text{D}_E$ & $\text{D}_F$ & 
$\text{D}_S$ & $\text{D}_3$ &
$\text{D}_E$ & $\text{D}_F$ & 
$\text{D}_S$ & $\text{D}_3$\\              
                     
\hline 
&  &  & & & & & & & & & &   \\
                      
\textbf{DCC} & 2.76  &2.78 &  \textbf{0.69} & 230.82 & 11.81 & 12.01 & \textbf{7.10} & 973.43 & 94.39   & 98.41 & 152.09   & 14010.31   \\

&  &  & & & & & & & & & &   \\

\textbf{sBEKK} &  \textbf{0.07}& \textbf{0.09}  &  2.93&   \textbf{0.30} & \textbf{0.69} &\textbf{0.92} &45.79  & \textbf{4.75} & \textbf{6.70}  & \textbf{10.69}   & 67019.10    & \textbf{70.50}  \\

&  &  & & & & & & & & & &   \\

$\textbf{fGARCH}_1$ & 2.64 &  3.00  & 2.27& 319.9  & 10.02& 12.13 & 14.46   & 1342.87 & 124.25   &  153.35   & 87.65   & 49339.29   \\

&  &  & & & & & & & & & &   \\

$\textbf{fGARCH}_2$ &  2.96 & 3.36 & 2.30  & 352.28 & 10.97& 13.48&  14.47  & 1434.83 & 160.42   &   199.03    &  87.34    & 60142.11   \\

&  &  & & & & & & & & & &   \\

$\textbf{fGARCH}_3$ &  3.02 & 3.46 &  2.44&356.64 & 11.03& 13.56 & 14.53  &  1436.71 & 166.27  & 207.7   &  \textbf{87.25}   & 60856.09   \\

&  &  & & & & & & & & & &  \\

$\textbf{fGARCH}_4$ &  3.02  &3.46 &2.60 & 356.63 &  11.11 &13.71 &  14.62 & 1441.78 & 168.68   & 211.45   & 87.27   & 61097.31   \\

&  &  & & & & & & & & & &  \\

$\textbf{fGARCH}_5$ & 3.02 & 3.45&2.77&356.63 & 11.13 & 13.73&  14.77&  1442.01 & 171.15  & 215.22 & 87.33 & 61444.25  \\

&  &  & & & & & & & & & &  \\

\textbf{$\text{fMSV}_1$} & 0.18 & 0.21 & 2.13 &  0.45 & 1.64&  2.12 & 14.48 &  9.91 & 33.83   & 46.05  &  87.94  & 2028.46  \\

&  &  & & & & & & & & & &  \\

\textbf{$\text{fMSV}_2$} &  0.18 & 0.21 &2.04   &  0.46 & 1.73 & 2.25  & 14.41& 11.98 & 28.14  & 40.36   & 87.77  &978.79   \\

&  &  & & & & & & & & & & \\

\textbf{$\text{fMSV}_3$} &  0.18 &0.21 & 1.96 & 0.46 & 1.91 & 2.54& 14.36 &15.86 & 27.36   & 40.48 &  87.65  & 2395.68  \\

&  &  & & & & & & & & & &   \\

\textbf{$\text{fMSV}_4$} &  0.18 & 0.21  & 1.92& 0.45 & 1.71&  2.21 & 14.36   & 9.98 & 34.03  & 50.59  & 87.64  &2547.49  \\

&  &  & & & & & & & & & &   \\

\textbf{$\text{fMSV}_5$} &   0.18 & 0.21  &  1.93 & 0.46 & 1.70 & 2.21& 14.32 &  9.81    & 19.09  &30.11   &87.58  &    580.86  \\

&  &  & & & & & & & & & & \\

\hline\hline
\end{tabular}}
}
\end{table}

\begin{table}[h]
\caption{\textcolor{black}{In-sample accuracy based on $\text{D}_E, \text{D}_F,\text{D}_S,\text{D}_3$ measures with $T=2000$, averaged over $100$ batches, under DGP 2.}}\label{simulations_results_DGP2}
\scalebox{0.8}{{\color{black}\begin{tabular}{c|cccc|cccc|cccc}\hline\hline
&  &  & & & & & & & & & & \\
& \multicolumn{4}{c|}{\textbf{$p=20$}} & \multicolumn{4}{c|}{\textbf{$p=100$}} & \multicolumn{4}{c}{\textbf{$p=500$}} \\
& $\text{D}_E$ & $\text{D}_F$ & 
$\text{D}_S$ & $\text{D}_3$ & 
$\text{D}_E$ & $\text{D}_F$ & 
$\text{D}_S$ & $\text{D}_3$ &
$\text{D}_E$ & $\text{D}_F$ & 
$\text{D}_S$ & $\text{D}_3$\\

\hline 
&  &  & & & & & & & & & &   \\
       
\textbf{DCC} &  1.89 &  2.56 &  \textbf{0.58} &  81.00 & 16.62&29.24&\textbf{5.29}&939.42& 227.91&397.9&124.22&58605.73  \\

&  &  & & & & & & & & & &   \\

\textbf{sBEKK} &  9.22  &  12.93 &  139.68 & 30.68 & 242.04&327.62&1296.04&2809.14 & 4169.92&5539.27&623749.21&138910.59  \\

&  &  & & & & & & & & & &   \\

$\textbf{fGARCH}_1$ &  1.54 & 2.96 &  12.92 & 11.62 & 44.47&88.43&87.44&1160.98 &  943.15&1884.89&463.67&30114.81  \\

&  &  & & & & & & & & & &   \\

$\textbf{fGARCH}_2$ &  0.86 & 1.55   &  1.38   &  22.78 & 13.16&25.82&6.76&1068.04 &\textbf{153.57}&\textbf{305.7}&34.17&\textbf{19875.41}   \\

&  &  & & & & & & & & & &   \\

$\textbf{fGARCH}_3$ & 1.2  &   1.92 &  1.45  & 55.93 & \textbf{14.80}&\textbf{27.66}&6.73&1259.91 & 180.08&334.16&\textbf{34.12}&33096.21  \\

&  &  & & & & & & & & & &  \\

$\textbf{fGARCH}_4$ & 1.21 & 1.93 & 1.62 &  55.95 & 14.89&27.77&6.72&1262.62 & 182.08&336.53&34.16&33352.24   \\

&  &  & & & & & & & & & &  \\

$\textbf{fGARCH}_5$ & 1.21& 1.93 & 1.86  & 55.95 & 14.89&27.78&6.76&1262.62 & 182.59&337.21&34.24&33390.95   \\

&  &  & & & & & & & & & &  \\

\textbf{$\text{fMSV}_1$} & 1.64 & 3.16     &12.59 &  5.79 & 51.21&101.75&85.52&921.33 & 1002.26&2002.70&462.23&40620.36   \\

&  &  & & & & & & & & & &  \\

\textbf{$\text{fMSV}_2$} &  0.78 & 1.43   & 1.21 &  5.27 & 17.39 & 34.12&6.70&\textbf{807.57} & 195.65&389.56&34.13&23826.70  \\

&  &  & & & & & & & & & & \\

\textbf{$\text{fMSV}_3$} &  \textbf{0.73}&  \textbf{1.33}  &  1.04 & \textbf{4.46} & 18.55 & 36.39 &    6.62     &  860.68 & 208.37&414.88&34.16&27022.14  \\

&  &  & & & & & & & & & &   \\

\textbf{$\text{fMSV}_4$} &  0.80& 1.46& 0.96 & 5.18 & 17.47  &34.27  &  6.52 & 809.12 & 228.18&453.94&34.22&27575.16   \\

&  &  & & & & & & & & & &   \\

\textbf{$\text{fMSV}_5$} &  0.80 & 1.46 &  0.93 & 5.08 & 37.76 &  73.89 &6.46  & 13831.90 & 251.83&501.39&34.40&33981.85  \\

&  &  & & & & & & & & & & \\

\hline\hline
\end{tabular}}
}
\end{table}

\section{Out-of-sample analysis with real data}\label{sec:real_data}

In this section, the relevance of our method is compared with existing competing models through out-of-sample forecasts based on real financial data. The performances will be assessed through measures of statistical accuracy and economic performances. The accuracy of the forecasts will be ranked by the model confidence set (MCS) procedure of \cite{hansen2011}, which provides a testing framework for the null hypothesis of forecast equivalence across subsets of models.

\subsection{Data}

We consider hereafter the stochastic process $(y_t)_{t\in\Zb}$ in $\Rb^p$ of the log-stock returns, where $y_{j,t} = 100\times\log(P_{j,t}/P_{j,t-1}), 1 \leq j \leq p$ with $P_{j,t}$ the stock price of the $j$-th index at time $t$. We study three financial portfolios: a low-dimensional portfolio of daily log-returns composed of the MSCI stock index based on the sample December 1998 -- March 2018, which yields a sample size $T = 5006$, and for 23 countries\footnote{Australia, Austria, Belgium, Canada, Denmark, Finland, France, Germany, Greece, Hong Kong, Ireland, Italy, Japan, Netherlands, New Zealand, Norway, Portugal, Singapore, Spain, Sweden, Switzerland, the United-Kingdom, the United-States}; a mid-dimensional portfolio of daily log-returns listed in the S\&P 100, where we selected $94$ firms\footnote{S\&P 100 indices excluding  AbbVie Inc., Dow Inc., General Motors, Kraft Heinz, Kinder Morgan and PayPal Holdings} that have been continuously included in the index over the period February 2010 -- January 2020, i.e., with a sample size $T=2500$; a high-dimensional portfolio of daily log-returns listed in the S\&P 500, where we selected $480$ firms\footnote{S\&P 500 indices excluding Etsy Inc., Solaredge Technologies Inc., PayPal, Hewlett Packard Enterprise, Under Armour (class C), Fortive, Lamb Weston, Ingersoll Rand Inc., Ceridian HCM, Linde PLC, Moderna, Fox Corporation (class A and B), Dow Inc., Corteva Inc., Amcor, Otis Worldwide, Carrier Global, Match Group, Viatris Inc.} that have been continuously included in the index over the period September 2014 -- January 2022, i.e., with a sample size $T=1901$. Hereafter, the portfolios are denoted by MSCI, S\&P 100 and S\&P 500, respectively\footnote{The MSCI, S\&P 100 and S\&P 500 data can be found on https://www.msci.com/, https://finance.yahoo.com and https://macrobond.com, respectively. MSCI and S\&P 100 data are publicly available in \url{https://github.com/Benjamin-Poignard/fMSV}; the S\&P 500 data requires a license and is thus not publicly available}. 

\subsection{Statistical accuracy and portfolio allocation}

We measure the statistical accuracy of the variance-covariance models based on the losses defined in (\ref{ranking_metrics}). Since $H_t$ is the underlying variance-covariance of $y_t$ and is unobserved, we replace it 
with a proxy for the realized covariance. This proxy is defined as $(1-a) y_ty^\top_t + a T^{*-1} \sum_{s=t-T^*+1}^t y_s y^\top_s$ ($t>T^*$), where $a=0.01$ and $T^*$ denotes the length of the in-sample period. This proxy represents a weighted average of $y_ty^\top_t$, which serves as a consistent (though noisy) estimator of $H_t$, and the rolling sample covariance matrix up to time $t$. The variance-covariance $\widehat{H}_t$ is estimated in-sample and the losses $\text{D}_E, \text{D}_F$, \textcolor{black}{$\text{D}_S$, $\text{D}_3$} are evaluated over the out-of-sample periods. The economic performances are assessed through the Global Minimum Variance Portfolio (GMVP) investment problem. The latter problem at time $t$, in the absence of short-sales constraints, is defined as
\begin{equation}\label{portprob}
\min_{w_t} \; w^{\top}_t \; \widehat{H}_t \; w_t, \;\;\text{s.t.} \;\;\iota^{\top} w_t = 1,
\end{equation}
where $\widehat{H}_t$ is the $p \times p$ one-step ahead forecast of the conditional variance-covariance matrix built at time $t-1$ and $\iota$ is a $p \times 1$ vector of $1$'s. The explicit solution is given by $\omega_t = \widehat{H}^{-1}_t \iota/\iota^{\top}\widehat{H}^{-1}_t\iota$: as a function depending only on $\widehat{H}_t$, the GMVP performance essentially depends on the precise measurement of the variance-covariance matrix. The following out-of-sample performance metrics (annualized) will be reported: \textbf{AVG} as the average of the out-of-sample portfolio returns, multiplied by $252$; \textbf{SD} as the standard deviation of the out-of-sample portfolio returns, multiplied by $\sqrt{252}$; \textbf{IR} as the information ratio computed as $\textbf{AVG}/\textbf{SD}$. The key performance measure is the out-of-sample \textbf{SD}. The GMVP problem essentially aims to minimize the variance rather than to maximize the expected return (that is high \textbf{IR}). 
As an alternative to the GMVP, we examine the performance of the risk parity portfolio (RPP), described in Appendix \ref{portfolio}. The RPP aims to equalize risk contributions across all asset classes in order to maximize diversification benefits. We evaluate its performance, expecting a high \textbf{IR}.

\mds

The out-of-sample analysis evaluates statistical accuracy using $\text{D}_E$, $\text{D}_F$, $\text{D}_S$, and $\text{D}_3$ as well as GMVP and RPP performances for the competing models described in Subsection \ref{sec:competing_models}.  
In addition, we consider \textcolor{black}{four} benchmark procedures: \textcolor{black}{the equally weighted portfolio $1/p$, the sample covariance matrix (SCov)}, the geometric-inverse shrinkage (GIS) estimator of \cite{ledoit2022}, a nonlinear shrinkage estimator based on the symmetrized Kullback-Leibler loss, and the linear shrinkage towards one-parameter covariance estimator (Cov1Para) of \cite{ledoit2004}. \textcolor{black}{The last three models are estimated using the in-sample period.}
The MCS test for ranking the statistical accuracy takes $\text{D}_E,\text{D}_F$, $\text{D}_S$, and $\text{D}_3$ as loss functions. The \textcolor{black}{GMVP} performances are ranked out-of-sample by the MCS test which takes the distance based on the difference of the \textcolor{blue}{squared} returns of portfolios $i$ and $j$, defined as
$u_{ij,t} = \big(\mu_{i,t}-\overline{\mu}_i\big)^2 - \big(\mu_{j,t}-\overline{\mu}_j\big)^2$, with $\mu_{k,t}=w^{\top}_{k,t} y_t$ the portfolio return at time $t$, where $w_{k,t}$ represents the GMVP weight deduced from variance-covariance model $k$, and $\overline{\mu}_k$ is the average portfolio return over the period. 
For the RPP performances, we set $u_{ij,t} = - \mu_{i,t}/s_i + \mu_{j,t}/s_j$ using the RPP weights with $s_k$ the standard deviation of the RPP return over the period.
The MCS test is evaluated at the $10\%$ level based on the range statistic and with block bootstrap with $10,000$ replications: see \cite{hansen2003} for further technical procedures. We report below
the out-of-sample periods used for the test accuracy, together with the number of factors $\widehat{m}$ selected in-sample by the variance-covariance eigenvalue-based procedure of \cite{onatski2010} on an indicative basis:
\begin{itemize}
    \item[(i)] MSCI: in sample 1999/01/01--2013/12/12, $\widehat{m}=1$; out-of-sample 2013/12/13--2018/03/12.
    \item[(ii)] S\&P 100: in sample 2010/02/19--2014/07/02, $\widehat{m}=3$; out-of-sample 2014/07/03--2020/01/23.
    \item[(iii)] S\&P 500: in sample 2014/09/25--2018/05/14, $\widehat{m}=1$; out-of-sample 2018/05/15--2022/01/27.
\end{itemize}
Table \ref{Table_metrics_stat_acc_performance} reports the average annualized values of $\text{D}_E$, $\text{D}_F$, \textcolor{black}{$\text{D}_S$, and $\text{D}_3$}, along with the MCS results. 
\textcolor{black}{The $\text{fMSV}_1$ model achieves the lowest distances for $\text{D}_E$, $\text{D}_F$, and $\text{D}_3$, whereas the results for $\text{D}_S$ vary, favoring the DCC for the low/mid dimensional portfolio and the $\text{fMSV}_5$ model for the high-dimensional portfolio. As detailed in Subsection \ref{sec:loss_accuracy}, $\text{D}_S$ and $\text{D}_3$ are asymmetric with respect to over- and under-predictions. $\text{D}_3$ yields rankings similar to those obtained from $\text{D}_E,\text{D}_F$, whereas the Stein loss $\text{D}_S$ leads to different conclusions. This indicates that relative and inverse-covariance errors drive the differences across models. 
For the MSCI and S\&P 500 portfolios, the selection of fMSV$_1$ is in line with the number of factors determined by the method of \cite{onatski2010}. Regarding the S\&P 100 portfolio, while the approach of \cite{onatski2010} chose $m=3$, the distance measures tend to favor simpler models for prediction. 
This discrepancy implies that more parsimonious models are favored for forecasting as they exclude overfitting noise in the estimation period.
In several cases, fMSV models with $m \geq 2$ are included in the MCS.} 

\mds

Table \ref{Table_metrics_econ_performance} displays the GMVP performance results, which can be summarized as follows (note that, unless stated otherwise, the results refer to \textbf{SD}). Overall, the fMSV models provide the best economic performance. Specifically, the $\text{fMSV}_5$ model achieves the lowest standard deviations for the MSCI and S\&P 100 portfolios, whereas the S\&P 500 portfolio favors the $\text{fMSV}_3$ model. \textcolor{black}{In the latter portfolio, among the benchmark procedures, the GIS is competitive with GARCH-class models but exhibits a higher SD than the fMSV models.} Furthermore, $\text{fMSV}_m$ $(m \geq 2)$ models are included in the MCS for the S\&P 500 portfolio. Unlike the statistical accuracy results in Table \ref{Table_metrics_stat_acc_performance}, these findings indicate that factor models with more factors yield better performances. In portfolio construction, the flexibility inherent in factor models with several factors may have contributed to optimizing asset allocation, particularly during periods of high market risk.
Focusing on the $\text{D}_S$ measure and GMVP performance for the fMSV models, $\text{D}_S$ and \textbf{SD} select the same model ($\text{fMSV}_5$) for both MSCI and S\&P100 in Tables \ref{Table_metrics_stat_acc_performance} and \ref{Table_metrics_econ_performance}. Regarding S\&P500, while Table \ref{Table_metrics_stat_acc_performance} shows that $\text{fMSV}_3$ and $\text{fMSV}_5$ are included in the MCS under $\text{D}_S$, Table \ref{Table_metrics_econ_performance} indicates that fMSV$_3$ achieves the smallest \textbf{SD}. Overall, these results suggest that $\text{D}_S$ and \textbf{SD} tend to select similar models.

\mds

The RPP results are displayed in Table \ref{Table_metrics_econ_performance2}, and indicate that most of \textbf{IR}s are lower than those of the GMVP in Table \ref{Table_metrics_econ_performance}. 
In other words, the GMVP outperforms the RPP for the datasets.
Reflecting the result, there are no major differences among the models. 

\mds

Tables \ref{Table_metrics_stat_acc_performance} and \ref{Table_metrics_econ_performance} illustrate that fMSV models generally outperform competing models, suggesting several key insights: (i) the stochastic nature of fMSV models captures market dynamics more effectively than the deterministic structure of GARCH class models; (ii) fMSV models with a larger number of factors provide the necessary flexibility to capture dynamic inter-asset relationships; and (iii) the strong performance of DCC models in certain cases puts a stress on the importance of incorporating dynamic correlations in covariance estimation.

\begin{table}[h]
\caption{\textcolor{black}{Annualized performance metrics $\text{D}_E$, $\text{D}_F$, $\text{D}_S$, and $\text{D}_3$ for various estimators.}}\label{Table_metrics_stat_acc_performance}
\scalebox{0.8}{\color{black}\begin{tabular}{c|cccc||cccc||cccc}\hline\hline
& \multicolumn{4}{c||}{MSCI} & \multicolumn{4}{c||}{S\&P 100} &  \multicolumn{4}{c}{S\&P 500} \\
&    \multicolumn{4}{c||}{\scalebox{0.75}{2013/12/13--2018/03/12}} &  \multicolumn{4}{c||}{\scalebox{0.75}{2014/07/03--2020/01/23}} &  \multicolumn{4}{c}{\scalebox{0.75}{2018/05/15--2022/01/27}} \\
  &  $\text{D}_E$ & $\text{D}_F$ & $\text{D}_S$ & $\text{D}_3$
  & $\text{D}_E$ & $\text{D}_F$ & $\text{D}_S$ & $\text{D}_3$ & $\text{D}_E$ & $\text{D}_F$ & $\text{D}_S$ & $\text{D}_3$\\ 
\hline

&  &  & & & & & & & & & &   \\ 

\textbf{DCC} &  1.06$^\star$ & 1.76$^\star$ & \textbf{2.18}$^\star$ & 0.35& 9.70& 17.76& \textbf{11.48}$^\star$ & 2.57& 5874.13& 11644.53& 89.20& 141862.05\\

&  &  & & & & & & & & & &   \\ 

\textbf{sBEKK} & 1.11& 1.84& 2.26& 0.35&12.51& 22.91& 18.33& 2.89& 5809.25& 11513.12&730.80& 141432.03\\
   
&  &  & & & & & & & & & &   \\

$\textbf{fGARCH}_1$ & 1.08& 1.78& 2.43& 0.35& 9.62& 17.61& 11.95& 2.56& 5575.14& 11048.43& 75.55& 139615.16\\
   
&  &  & & & & & & & & & &   \\
  
$\textbf{fGARCH}_2$ & 1.08& 1.78& 2.39& 0.35& 9.62& 17.61& 11.94& 2.56& 5573.03& 11044.41& 74.93& 139614.91\\
   
&  &  & & & & & & & & & &   \\
  
$\textbf{fGARCH}_3$ & 1.08& 1.78& 2.37& 0.35& 9.61& 17.60& 11.83& 2.56& 5572.71& 11043.79& 75.09& 139614.93\\
   
&  &  & & & & & & & & & &   \\

$\textbf{fGARCH}_4$ & 1.07& 1.78& 2.35& 0.35& 9.61& 17.59& 11.79& 2.56& 5571.90& 11042.08& 74.20$^\star$ & 139614.88\\
   
&  &  & & & & & & & & & &   \\

$\textbf{fGARCH}_5$ & 1.07& 1.78& 2.33& 0.35& 9.61& 17.59& 11.77& 2.56& 5571.74& 11041.64& 74.31$^\star$ & 139614.85\\
   
&  &  & & & & & & & & & &   \\

\textbf{$\text{fMSV}_1$} & \textbf{0.94}$^\star$ & \textbf{1.51}$^\star$ & 2.40& \textbf{0.33}$^\star$ & \textbf{8.74}$^\star$ & \textbf{15.88}$^\star$ & 11.97& \textbf{2.37}$^\star$ & \textbf{4502.96}$^\star$ &  \textbf{8910.01}$^\star$ & 75.65& \textbf{119546.97}$^\star$\\

&  &  & & & & & & & & & &   \\

\textbf{$\text{fMSV}_2$} & 1.01$^\star$ & 1.65$^\star$ & 2.36& 0.34$^\star$ & 9.10& 16.59& 11.83& 2.49& 5663.49 & 11230.80& 75.39& 189584.25\\
   
&  &  & & & & & & & & & &   \\

\textbf{$\text{fMSV}_3$} & 1.02$^\star$ & 1.66$^\star$ & 2.35& 0.34$^\star$ & 9.56& 17.50& 11.78& 2.70& 4922.50$^\star$ &  9747.55$^\star$ & 74.36$^\star$ & 128833.93$^\star$\\
   
&  &  & & & & & & & & & &   \\

\textbf{$\text{fMSV}_4$} & 1.00$^\star$ & 1.65$^\star$ & 2.33& 0.34$^\star$ & 9.47& 17.34& 11.78& 2.70& 5956.95& 11808.20& 74.40& 142592.08\\
   
&  &  & & & & & & & & & &   \\

\textbf{$\text{fMSV}_5$} & 1.09& 1.85& 2.31& 0.36&10.05& 18.46& 11.68& 3.12& 5999.93& 11893.18& \textbf{74.17}$^\star$ & 142861.62 \\
   
&  &  & & & & & & & & & &   \\

\hline\hline
\end{tabular}}
\begin{minipage}{16.4cm}
\footnotesize \textcolor{black}{Note: $\text{D}_E$, $\text{D}_F$, $\text{D}_S$, and $\text{D}_3$ represent the average (annualized) values of the distances defined in equation (\ref{ranking_metrics}).
For readability, $\text{D}_E$ and $\text{D}_F$ are scaled by $10^2$, $\text{D}_S$ by $10^{-4}$, and $\text{D}_3$ by $10^3$. The minimum values for each metric are indicated in bold.
`$\star$' denotes that the models are included in the MCS for a $90\%$-level ($p$-values $\geq 0.10$).
The out-of-sample periods are indicated above $\text{D}_E$, $\text{D}_F$, $\text{D}_S$, and $\text{D}_3$.}
\end{minipage}
\end{table}

\begin{table}[h!]
\caption{Annualized GMVP performance metrics for various estimators.}\label{Table_metrics_econ_performance}
\scalebox{0.8}{\color{black}\begin{tabular}{c|cccc||cccc||cccc}\hline\hline
& \multicolumn{4}{c||}{MSCI} & \multicolumn{4}{c||}{S\&P 100} &  \multicolumn{4}{c}{S\&P 500} \\
&    \multicolumn{4}{c||}{\scalebox{0.75}{2013/12/13--2018/03/12}} &  \multicolumn{4}{c||}{\scalebox{0.75}{2014/07/03--2020/01/23}} &  \multicolumn{4}{c}{\scalebox{0.75}{2018/05/15--2022/01/27}} \\
  &  \textbf{AVG} & \textbf{SD} & \textbf{IR} & \textbf{MCS}
  & \textbf{AVG} & \textbf{SD} & \textbf{IR} & \textbf{MCS} & \textbf{AVG} & \textbf{SD} & \textbf{IR} & \textbf{MCS}\\ 
\hline        

&  &  & & & & & & & & & &   \\
   
\textbf{DCC} & 6.196 & 8.921 & 0.695 &0.087& 17.293 & 11.941 & 1.448 &0    & 11.532 & 19.805 & 0.582 &\textbf{0.476} \\
  
&  &  & & & & & & & & & &   \\
   
\textbf{sBEKK}  &10.095 & 9.367 & 1.078 &0.002& 23.489 & 18.815 & 1.248 &0    &  1.008 & 45.652 & 0.022 &0 \\
  
&  &  & & & & & & & & & &   \\
   
$\textbf{fGARCH}_1$   & 6.031 & 9.875 & 0.611 &0    &  8.807 & 12.649 & 0.696 &0    &  9.994 & 20.948 & 0.477 &0.007 \\
  
&  &  & & & & & & & & & &   \\
   
$\textbf{fGARCH}_2$   & 7.800 & 9.403 & 0.830 &0.001&  9.143 & 12.691 & 0.720 &0    &  2.074 & 24.727 & 0.084 &0.002 \\
  
&  &  & & & & & & & & & &   \\
   
$\textbf{fGARCH}_3$   & 5.042 & 9.878 & 0.511 &0    &  7.835 & 12.634 & 0.620 &0    &  5.218 & 24.167 & 0.216 &0.002 \\
  
&  &  & & & & & & & & & &   \\
   
$\textbf{fGARCH}_4$   & 3.474 &10.315 & 0.337 &0    & 13.190 & 12.273 & 1.075 &0    & 12.871 & 18.083 & 0.712 &\textbf{0.219} \\
  
&  &  & & & & & & & & & &   \\
   
$\textbf{fGARCH}_5$   & 5.543 & 9.747 & 0.569 &0    & 14.176 & 12.453 & 1.138 &0    & 13.570 & 18.938 & 0.717 &0.024 \\
  
&  &  & & & & & & & & & &   \\
   
$\textbf{fMSV}_1$   & 6.305 & 9.646 & 0.654 &0.002& 10.923 & 11.987 & 0.911 &0    &  7.942 & 20.874 & 0.381 &0.007 \\
  
&  &  & & & & & & & & & &   \\
   
$\textbf{fMSV}_2$   & 7.340 & 8.615 & 0.852 &\textbf{0.183}& 12.553 & 11.708 & 1.072 &0    & 10.811 & 16.241 & 0.666 &\textbf{0.476} \\
  
&  &  & & & & & & & & & &   \\
   
$\textbf{fMSV}_3$   & 5.745 & 8.809 & 0.652 &0.087& 15.260 & 11.853 & 1.288 &0    & 11.987 & \textbf{15.602} & 0.768 &\textbf{1.000} \\
  
&  &  & & & & & & & & & &   \\
   
$\textbf{fMSV}_4$   & 4.641 & 8.996 & 0.516 &0.002& 14.993 & 11.590 & 1.294 &0    & 10.736 & 15.848 & 0.677 &\textbf{0.549} \\
  
&  &  & & & & & & & & & &   \\
   
$\textbf{fMSV}_5$   & 6.514 & \textbf{8.356} & 0.780 &\textbf{1.000}& 12.206 &  \textbf{9.593} & 1.273 &\textbf{1.000}& 10.668 & 15.894 & 0.671 &\textbf{0.549} \\
  
&  &  & & & & & & & & & &   \\
   
\textbf{$1/p$}   & 3.859 & 13.593 & 0.284 &0    & 8.676 & 13.028 & 0.666 &0    &  9.708 & 23.163 & 0.419 &0.024 \\

&  &  & & & & & & & & & &   \\

\textbf{SCov}   & 6.723 & 9.469 & 0.710 &0    & 13.187 & 11.412 & 1.156 &0    & $-$0.600 & 22.686 & $-$0.026 &0.002 \\

&  &  & & & & & & & & & &   \\

\textbf{GIS}   & 6.755 & 9.454 & 0.714 &0    & 13.073 & 11.285 & 1.158 &0    &  6.323 & 18.809 & 0.336 &\textbf{0.219} \\
  
&  &  & & & & & & & & & &   \\
   
\textbf{Cov1Para}   & 6.735 & 9.454 & 0.712 &0    & 13.052 & 11.326 & 1.152 &0    &  2.539 & 20.940 & 0.121 &0.007 \\

&  &  & & & & & & & & & &   \\

\hline\hline
\end{tabular}}
\begin{minipage}{15.3cm}
\footnotesize Note: The lowest \textbf{SD} figure is in bold face. \textbf{MCS} column contains the p-values of the MCS tests, where the bold figures are the models included in the MCS for a $90\%$-level (p-values $\geq 0.10$). The out-of-sample periods are indicated above \textbf{AVG}, \textbf{SD}, \textbf{IR} and \textbf{MCS}.
\end{minipage}
\end{table}

\begin{table}[h!]
\caption{Annualized RPP performance metrics for various estimators.}\label{Table_metrics_econ_performance2}
\scalebox{0.8}{\color{black}\begin{tabular}{c|cccc||cccc||cccc}\hline\hline
& \multicolumn{4}{c||}{MSCI} & \multicolumn{4}{c||}{S\&P 100} &  \multicolumn{4}{c}{S\&P 500} \\
&    \multicolumn{4}{c||}{\scalebox{0.75}{2013/12/13--2018/03/12}} &  \multicolumn{4}{c||}{\scalebox{0.75}{2014/07/03--2020/01/23}} &  \multicolumn{4}{c}{\scalebox{0.75}{2018/05/15--2022/01/27}} \\
  &  \textbf{AVG} & \textbf{SD} & \textbf{IR} & \textbf{MCS}
  & \textbf{AVG} & \textbf{SD} & \textbf{IR} & \textbf{MCS} & \textbf{AVG} & \textbf{SD} & \textbf{IR} & \textbf{MCS}\\ 
\hline

&  &  & & & & & & & & & &   \\
   
\textbf{DCC}   & 4.424 & 11.919 & 0.371 &\textbf{0.722}& 8.881 & 12.230 & \textbf{0.726} &\textbf{1.000}& 9.390 & 20.811 & 0.451 &\textbf{0.934}\\
   
&  &  & & & & & & & & & &   \\
   
\textbf{sBEKK}    & 4.693 & 12.063 & 0.389 &\textbf{0.722}& 8.383 & 12.073 & 0.694 &\textbf{0.705}& 9.675 & 20.553 & \textbf{0.471} &\textbf{1.000}\\
   
&  &  & & & & & & & & & &   \\

$\textbf{fGARCH}_1$     & 4.425 & 12.398 & 0.357 &\textbf{0.253}& 8.565 & 12.049 & 0.711 &\textbf{0.908}& 9.517 & 21.493 & 0.443 &\textbf{0.934}\\
   
&  &  & & & & & & & & & &   \\
   
$\textbf{fGARCH}_2$     & 4.492 & 12.450 & 0.361 &\textbf{0.253}& 8.617 & 12.061 & 0.715 &\textbf{0.996}& 9.439 & 21.747 & 0.434 &\textbf{0.649}\\
   
&  &  & & & & & & & & & &   \\
   
$\textbf{fGARCH}_3$     & 4.172 & 12.707 & 0.328 &\textbf{0.218}& 8.652 & 12.102 & 0.715 &\textbf{0.996}& 9.338 & 21.851 & 0.427 &0.009\\
   
&  &  & & & & & & & & & &   \\
   
$\textbf{fGARCH}_4$     & 4.123 & 12.726 & 0.324 &\textbf{0.218}& 8.726 & 12.172 & 0.717 &\textbf{0.996}& 9.358 & 21.853 & 0.428 &0.009\\
   
&  &  & & & & & & & & & &   \\
   
$\textbf{fGARCH}_5$     & 4.164 & 12.707 & 0.328 &\textbf{0.218}& 8.731 & 12.176 & 0.717 &\textbf{0.996}& 9.351 & 21.855 & 0.428 &0.009\\
   
&  &  & & & & & & & & & &   \\
   
$\textbf{fMSV}_1$     & 4.321 & 12.236 & 0.353 &\textbf{0.253}& 8.597 & 12.120 & 0.709 &\textbf{0.908}& 9.456 & 21.636 & 0.437 &\textbf{0.710}\\
   
&  &  & & & & & & & & & &   \\
   
$\textbf{fMSV}_2$     & 4.388 & 12.406 & 0.354 &\textbf{0.253}& 8.704 & 12.205 & 0.713 &\textbf{0.908}& 9.554 & 21.836 & 0.438 &\textbf{0.649}\\
   
&  &  & & & & & & & & & &   \\
   
$\textbf{fMSV}_3$     & 4.322 & 12.415 & 0.348 &\textbf{0.218}& 8.758 & 12.216 & 0.717 &\textbf{0.996}& 9.252 & 21.792 & 0.425 &0.009\\
   
&  &  & & & & & & & & & &   \\
   
$\textbf{fMSV}_4$     & 4.319 & 12.293 & 0.351 &\textbf{0.722}& 8.768 & 12.260 & 0.715 &\textbf{0.993}& 9.389 & 21.760 & 0.432 &0.009\\
   
&  &  & & & & & & & & & &   \\
   
$\textbf{fMSV}_5$     & 4.805 & 12.003 & \textbf{0.400} &\textbf{1.000}& 8.614 & 12.110 & 0.711 &\textbf{0.996}& 9.962 & 21.784 & 0.457 &\textbf{0.934}\\
   
&  &  & & & & & & & & & &   \\
   
\textbf{$1/p$}   & 3.859 & 13.593 & 0.284 &\textbf{0.218}& 8.676 & 13.028 & 0.666 &\textbf{0.878}&  9.708 & 23.163 & 0.419 &\textbf{0.649} \\

&  &  & & & & & & & && &   \\

\textbf{SCov}   & 4.471 & 12.435 & 0.360 &\textbf{0.253}& 8.671 & 12.167 & 0.713 &\textbf{0.908}&  9.655 & 21.962 & 0.440 &\textbf{0.710} \\

&  &  & & & & & & & && &   \\
\textbf{GIS}     & 4.471 & 12.434 & 0.360 &\textbf{0.253}& 8.669 & 12.164 & 0.713 &\textbf{0.908}& 9.654 & 21.952 & 0.440 &\textbf{0.710}\\
   
&  &  & & & & & & & & & &   \\

\textbf{Cov1Para}     & 4.470 & 12.436 & 0.359 &\textbf{0.253}& 8.671 & 12.168 & 0.713 &\textbf{0.908}& 9.655 & 21.963 & 0.440 &\textbf{0.710}\\
   
&  &  & & & & & & & && &   \\

\hline\hline
\end{tabular}}
\begin{minipage}{15.3cm}
\footnotesize Note: The highest \textbf{IR} figure is in bold face. \textbf{MCS} column contains the p-values of the MCS tests, where the bold figures are the models included in the MCS for a $90\%$-level (p-values $\geq 0.10$). The out-of-sample periods are indicated above \textbf{AVG}, \textbf{SD}, \textbf{IR} and \textbf{MCS}.
\end{minipage}
\end{table}

\section{Conclusion}

The paper considers a two-stage estimation where the factor model is estimated in the first stage, while the second stage estimates the MSV processes of the estimated factors. We provide some asymptotic results for the second stage estimators.
Monte Carlo results using the DGPs based on MGARCH models indicate that the forecasts of the fMSV model perform well compared with non-factor models.
The empirical results based on real data show that the performances of the forecasts of the fMSV are better than those of the competing MGARCH models. 

Several directions relating to the theoretical and the empirical analysis on fMSV models can be considered.
The first topic concerns the variance-covariance matrix of the idiosyncratic errors. We may extend it to allow stochastic volatility, as in \cite{chib2006}.
The second one is to accommodate realized covariance and asymmetric effects, as in \cite{asai2015}.
\textcolor{black}{The third direction is to incorporate dynamic factors, as explored by \cite{barigozzi2017} and \cite{lam2012}. Fourth, while our current framework is based on an additive factor structure, we could consider multiplicative volatility factors by extending the frameworks of \cite{ding2025} and \cite{ray2000}.} 
We shall leave these issues for the future research.

\bibliography{biblio} 

\appendix

\section{Asymptotic properties}\label{appendix_asymptotic_prop}

In this section, we provide the large sample properties of $\widetilde{\theta}_1, \widehat{\theta}_1$ and $\widehat{\theta}_2$, which are defined in (\ref{obj_crit_first_step}), (\ref{obj_crit_first_estimator}) and (\ref{obj_crit_second}) in Subsection \ref{subsec:msv_estimation}, respectively. Hereafter, we denote by $\theta_{01}=\text{vec}(\underline{\Psi}_0)\in \Rb^{d_1}$ and $\theta_{02}=(c^{*\top}_0,\theta^\top_{0,\Phi},\theta^\top_{0,\Xi})^\top\in \Rb^{d_2}$ the ``true values'' of the parameters $\theta_1=\text{vec}(\underline{\Psi}) \in \Rb^{d_1}$ and $\theta_2 = (c^{*\top},\theta^\top_{\Phi},\theta^\top_{\Xi})^\top\in\Rb^{d_2}$, $d_1=m+qm^2$, and $d_2=m(1+2m)$. \textcolor{black}{We assume the following uniform consistency on the estimated factors:
\begin{assumption}\label{assumption_unif_consistency_factors}
Assume $T=o(p^2)$. The GLS estimator $\widehat{f}_t = (\widehat{\Lambda}^\top \widehat{\Sigma}_\varepsilon^{-1} \widehat{\Lambda})^{-1} 
\widehat{\Lambda}^\top \widehat{\Sigma}_\varepsilon^{-1} y_t$ of $f_t$ satisfies the uniform consistency $\underset{t\leq T}{\max}\|\widehat{f}_t-f_t\|_2 = O_p(\frac{1}{T^{1/4}}+\frac{T^{1/4}}{\sqrt{p}})$.
\end{assumption}}
\textcolor{black}{From Theorem 6.1 of \cite{bai2012}, $\|\widehat{f}_t-f_t\|_2=O_p(\frac{1}{\sqrt{p}})$ for any $t$. Following the proofs of Proposition 6.1 and Lemma D.1 of \cite{bai2012}, under suitable moment conditions on $f_t$ and $\varepsilon_t$, it can actually be shown that $\underset{t\leq T}{\max}\|\widehat{f}_t-f_t\|_2 = O_p(\frac{1}{T^{1/4}}+\frac{T^{1/4}}{\sqrt{p}})$  is satisfied. In Appendix \ref{appendix_factor_estim}, we specify these moment conditions and derive the uniform consistency of the estimated factor variables. Hereafter, we will denote $b_{T,p}=\frac{1}{T^{1/4}}+\frac{T^{1/4}}{\sqrt{p}}$. In our analysis, $q$ can potentially grow with $T$, whereas the number of factors $m$ is fixed, so that $d_1$ will grow with $q$.} We first consider the existence of a consistent first step estimator $\widetilde{\theta}_1$. We rely on the following conditions.

\begin{assumption}\label{assumption_regularity_matrix}
\textcolor{black}{The set $\Theta_{1}$ is a borelian subset of $\Rb^{d_1}$, $d_1 = m+qm^2$.} For any $F_t$, the map $\ell(F_t;\theta_{01})$ is twice differentiable on $\Theta_{1}$. Let $\nabla_{\theta_1}\ell(F_t;\theta_{01}) = -\text{vec}\big(\big(x_t-\underline{\Psi}_0 Z_{q,t-1}\big)Z^\top_{q,t-1}\big)$ and $\nabla^2_{\theta_{1}\theta^\top_{1}}\ell(F_t;\theta_{01}) = \big(Z_{q,t-1}Z^\top_{q,t-1} \otimes I_m\big)$. Moreover, the maps $F_t\mapsto \nabla_{\theta_1} \ell(F_t;\theta_{01})$ and $F_t\mapsto \nabla^2_{\theta_1\theta^\top_1} \ell(F_t;\theta_{01})$ are differentiable on $\Rb^{m(q+1)}$. $\Hb(\theta_{01}) := \Eb[\nabla^2_{\theta_1 \theta^\top_1}\ell(F_t;\theta_{01})] \in \Rb^{d_1 \times d_1}$ exists, $\exists \alpha_1,\alpha_2$ with $0 < \alpha_1 < \alpha_2 < \infty$ such that $\alpha_1 < \lambda_{\min}(\Hb(\theta_{01})) < \lambda_{\max}(\Hb(\theta_{01})) < \alpha_2$.
\end{assumption}

\begin{assumption}\label{assumption_gradient_second_dev}
\begin{itemize}
\item[(i)] \textcolor{black}{$\sup_q\max_{1 \leq j \leq m,1\leq l \leq qm}\Eb[|x_{j,t-k}Z_{l,q,t-1}|]<\infty$.
    \item[(ii)] There are some functions $\Psi_1(\cdot),\Psi_2(\cdot)$ such that for any $T$: 
{\small{\begin{align*}
&\sup_q\Eb[\underset{1 \leq l \leq qm}{\underset{1 \leq j \leq m}{\max}}|u_{j,t}Z_{l,q,t-1}|\underset{1 \leq l \leq qm}{\underset{1 \leq j \leq m}{\max}}|u_{j,t'}Z_{l,q,t'-1}|] \leq \Psi_1(|t-t'|), \; \text{and} \; \underset{T>0}{\sup}\; \frac{1}{T}\overset{T}{\underset{t,t'=1}{\sum}} \Psi_1(|t-t'|)<\infty,\\
&\sup_q\underset{k,k'>q}{\sup}\Eb[\underset{1 \leq l \leq qm}{\underset{1 \leq j \leq m}{\max}}|x_{j,t-k}Z_{l,q,t-1}-\Eb[x_{j,t-k}Z_{l,q,t-1}]|\underset{1 \leq l \leq qm}{\underset{1 \leq j \leq m}{\max}}|x_{j,t'-k'}Z_{l,q,t'-1}-\Eb[x_{j,t'-k'}Z_{l,q,t'-1}]|]\leq \Psi_2(|t-t'|), \\
& \text{and} \;\; \underset{T>0}{\sup}\; \frac{1}{T}\overset{T}{\underset{t,t'=1}{\sum}} \Psi_2(|t-t'|)<\infty.
\end{align*}}}}
\item[(iii)] \textcolor{black}{Let $\zeta_{kl,t} = \big(Z_{q,t-1}Z^\top_{q,t-1} \otimes I_m\big)_{kl} - \Eb[\big(Z_{q,t-1}Z^\top_{q,t-1} \otimes I_m\big)_{kl}]$. Then there exists some function $\chi(\cdot)$ such that for any $T$: $|\Eb[\zeta_{kl,t} \zeta_{kl,t'}]| \leq \chi(|t-t'|), \;\; \text{and} \;\; \underset{T>0}{\sup} \; \frac{1}{T}\overset{T}{\underset{t,t'=1}{\sum}} \chi(|t-t'|) <\infty$.}
\end{itemize}
\end{assumption}
\begin{assumption}\label{assumption_ell_F_control_gradient_hessian}
\begin{itemize} \textcolor{black}{Let $b_{T,p}=\frac{1}{T^{1/4}}+\frac{T^{1/4}}{\sqrt{p}}$. Then:}
    \item[(i)] \textcolor{black}{There exists some measurable function $\rho(\cdot)$ such that for some $\eps>0$:
    {\small{\begin{equation*}
    \underset{\underset{1\leq l\leq m(q+1)}{1 \leq k\leq d_1}}{\sup}\underset{U_t: \|U_t-F_t\|_2\leq L\sqrt{q}b_{T,p}}{\sup} |\partial^2_{\theta_{1,k}F_{l}}\ell(U_t;\theta_{01})| \leq \rho(F_t), \; \text{and} \;\underset{T>0}{\sup} \; \frac{1}{T}\overset{T}{\underset{t,t'=1}{\sum}} \Eb[\rho(F_t)\rho(F_{t'})] <\infty, \; \text{with}\;L>0.
    \end{equation*}}}}
    \item[(ii)]  \textcolor{black}{There exists some measurable function $\zeta(\cdot)$ such that for some $\eps>0$:
    {\small{\begin{equation*}
    \underset{\underset{1\leq l\leq m(q+1)}{1 \leq k,k'\leq d_1}}{\sup}  \underset{U_t: \|U_t-F_t\|_2\leq L\sqrt{q}b_{T,p}}{\sup}|\partial^3_{\theta_{1,k}\theta_{1,k'}F_{l}}\ell(U_t;\theta_{01})| \leq \zeta(F_t), \; \text{and}
\;  \underset{T>0}{\sup} \; \frac{1}{T}\overset{T}{\underset{t,t'=1}{\sum}} \Eb[\zeta(F_t)\zeta(F_{t'})] <\infty,  \; \text{with}\;L>0.
    \end{equation*}}}}
\end{itemize}
\end{assumption}

\textcolor{black}{Assumption \ref{assumption_regularity_matrix} is} standard for M-estimation. The conditions stated in Assumption \ref{assumption_gradient_second_dev} are moment conditions in the same spirit as \cite{poignard2020aism}. Assumption \ref{assumption_ell_F_control_gradient_hessian} will allow us to control the estimation error originating from the estimation of the factor variable.

\begin{theorem}\label{bound_proba_first_step_estimator}
Under Assumptions \ref{assumption_unif_consistency_factors}-\ref{assumption_ell_F_control_gradient_hessian}, \textcolor{black}{$\sqrt{q}b_{T,p} \rightarrow 0$, $q^4b^2_{T,p}=o(T)$}, there exists a sequence $(\widetilde{\theta}_1)$ of solutions of (\ref{obj_crit_first_estimator}) that satisfies $\|\widetilde{\theta}_1-\theta_{01}\|_2 = O_p\big(\textcolor{black}{\sqrt{q/T} + \sqrt{q/T}\big(\|\mathbf{\Psi}_ {0q}\|_{s,1}+q\,b_{T,p}\big)+\sqrt{q}\sum_{k>q}\|\Psi_ {0k}\|_{s}e_{k,\max}}\big)$, \textcolor{black}{where $\|\mathbf{\Psi}_
{0q}\|_{s,1}:=\sum_{k>q}\|\Psi_{0k}\|_s$ and $e_{k,\max}=\underset{q}{\sup}
\underset{1\leq j \leq m}{\underset{1 \leq l \leq 1+qm}{\max}}\Eb[|x_{j,t-k}Z_{l,q,t-1}|]$, with $Z_{q,t-1}=(1,x^\top_{t-1},\ldots,x^\top_{t-q})^\top$}.
\end{theorem}
\textcolor{black}{\begin{remark}
The asymptotic bound of $\widetilde{\theta}_1$ depends on several terms: $\sqrt{q/T}$ is the standard rate we would have obtained in the context of linear regression with a diverging number of parameters; $\sqrt{q/T}\|\mathbf{\Psi}_ {0q}\|_{s,1}$ represents the truncation bias in the approximation by VAR($q$); $\sqrt{q}\sum_{k>q}\|\Psi_ {0k}\|_{s}e_{k,\max}$ also relates to the truncation bias, \textcolor{black}{and is connected to the population level score function of the non-penalized loss at the true parameter}; $\sqrt{q/T}qb_{T,p}$ relates to the estimation error originating from the factor variables. For $q$ sufficiently large and under exponential parameter decay, $\|\mathbf{\Psi}_ {0q}\|_{s,1}$ will become negligible. Indeed, if one assumes that $\exists \kappa>0$ and $0 \leq \beta <1$ such that $\|\Psi_{0k}\|_s \leq \kappa \beta^k$ for any $k \geq 1$, then $\sum_{k>q}\|\Psi_{0k}\|_s \leq \kappa \frac{\beta^q}{1-\beta}$. Take $q = L_1\log(T)$ with $L_1>0$ constant, then $\sum_{k>q}\|\Psi_{0k}\|_s \leq \frac{\kappa}{1-\beta}\exp(-L_1\log(T)\log(1/\beta)) \rightarrow 0$ as $T\rightarrow \infty$. The exponential decay follows from the stability and invertibility of the VARMA(1,1).
Moreover, under $\log(T)^2T=o(p^2)$, then $\sqrt{q} b_{T,p} \rightarrow 0$ and $q^4b^2_{T,p}=o(T)$. 
An alternative choice for $q$ is $q=L_1 T^{c}$, with $c>0$ a suitable constant so that $\sqrt{q} b_{T,p} \rightarrow 0$ and $q^4b^2_{T,p}=o(T)$.
\end{remark}}

\begin{proof}[Proof of Theorem \ref{bound_proba_first_step_estimator}]
Let \textcolor{black}{$\nu_T = \textcolor{black}{\sqrt{q/T} + \sqrt{q/T}\|\mathbf{\Psi}_ {0q}\|_{s,1} + \sqrt{q/T}\,q\,b_{T,p} + \sqrt{q}\sum_{k>q}\|\Psi_{0k}\|_se_{k,\max}}$, with $e_{k,\max}=\sup_q\max_{1\leq j \leq m,1 \leq l \leq qm}\Eb[|x_{j,t-k}Z_{l,q,t-1}|]$}. We work with $m$ fixed. Now, we would like to prove that for any $\eps > 0$, there exists $C_{\eps} > 0$ such that $\Pb(\cfrac{1}{\nu_T}\|\widetilde{\theta}_1 - \theta_{01}\|_2 > C_{\eps}) < \eps$.
We have
\begin{equation*}
\Pb(\cfrac{1}{\nu_T} \|\widetilde{\theta}_1 - \theta_{01}\|_2 > C_{\eps}) \leq \Pb(\exists \uu, \|\uu\|_2 = C_{\eps}: \Lb_T(\widehat{\mathbf{F}};\theta_{01}+\nu_T\uu) \leq \Lb_T(\widehat{\mathbf{F}};\theta_{01})),
\end{equation*}
which implies that there is a minimum in the ball $\{\theta_{01}+\nu_T\uu, \|\uu\|_2 \leq C_{\eps}\}$ so that the minimum $\widetilde{\theta}_1$ satisfies $\|\widetilde{\theta}_1 - \theta_{01}\|_2 = O_p(\nu_T)$. Now by a Taylor expansion of the loss, we obtain
\begin{eqnarray*}
\frac{1}{T}\Lb_T(\widehat{\mathbf{F}};\theta_{01}+\nu_T\uu)-\frac{1}{T}\Lb_T(\widehat{\mathbf{F}};\theta_{01})= \nu_T \uu^\top \nabla_{\theta_1} \frac{1}{T}\Lb_T(\widehat{\mathbf{F}};\theta_{01}) + \frac{\nu^2_T}{2} \uu^\top \nabla^2_{\theta_1 \theta^\top_1} \frac{1}{T}\Lb_T(\widehat{\mathbf{F}};\theta_{01})\uu=:T_1+T_2.
\end{eqnarray*}
We want to prove
\begin{equation} \label{bound_obj}
\Pb(\exists \uu, \|\uu\|_2 = C_{\eps} : T_1+T_2 \leq 0) < \eps.
\end{equation}
Let us consider \textcolor{black}{$T_1$}. By Cauchy-Schwarz inequality, we get
\begin{eqnarray*}
\underset{\uu:\|\uu\|_2=C_{\eps}}{\sup}|\uu^{\top}\nabla_{\theta} \frac{1}{T}\Lb_T(\widehat{\mathbf{F}};\theta_{01})| \leq \underset{\uu:\|\uu\|_2=C_{\eps}}{\sup}\|\uu\|_2\|\nabla_{\theta} \frac{1}{T}\Lb_T(\widehat{\mathbf{F}};\theta_{01})\|_2 .
\end{eqnarray*}
Recalling that $F_t = (f^\top_t,\ldots,f^\top_{t-q})^\top$, by a Taylor expansion, for some random vector $f^*_t$ such that $\|f^*_t-f_t\|_2 \leq \|\widehat{f}_t-f_t\|_2, t = 1,\ldots,T$, then we have element-by-element $k=1,\ldots,d_1$:
\begin{eqnarray*}
\partial_{\theta_{1,k}} \frac{1}{T}\Lb_T(\widehat{\mathbf{F}};\theta_{01})=\partial_{\theta_{1,k}} \frac{1}{T}\Lb_T(\mathbf{F};\theta_{01})+ \frac{1}{T}\overset{T}{\underset{t=1}{\sum}} \overset{m(q+1)}{\underset{l=1}{\sum}}\partial^2_{\theta_{1,k}F_{l}}\ell(F^*_t;\theta_{01})(\widehat{F}_{t,l}-F_{t,l}),
\end{eqnarray*}
where $F_{t,l} \in \Rb$ denotes the $l$-th element of $F_t \in \Rb^{m(q+1)}$, $F^*_t=(f^{*\top}_t,\ldots,f^{*\top}_{t-q})^\top$. Therefore
\begin{eqnarray*}
\textcolor{black}{\|\nabla_{\theta_{1}} \frac{1}{T}\Lb_T(\widehat{\mathbf{F}};\theta_{01})\|_2 \leq\|\nabla_{\theta_{1}} \frac{1}{T}\Lb_T(\mathbf{F};\theta_{01})\|_2+ \|\frac{1}{T}\overset{T}{\underset{t=1}{\sum}} \nabla^2_{\theta_{1}F^\top}\ell(F^*_t;\theta_{01})\|_F\sqrt{q+1}\underset{t\leq T}{\max}\|\widehat{f}_{t}-f_{t}\|_{\textcolor{black}{2}}}.
\end{eqnarray*}
For $a >0$, we have
\begin{eqnarray*}
\lefteqn{\Pb(\underset{\uu:\|\uu\|_2=C_{\eps}}{\sup}|\nu_T\uu^{\top}\nabla_{\theta_1} \frac{1}{T}\Lb_T(\widehat{\mathbf{F}};\theta_{01})| > a) }\\
&\leq &\Pb(\nu_T\|\nabla_{\theta_1} \frac{1}{T}\Lb_T(\widehat{\mathbf{F}};\theta_{01})\|_2 > \frac{a}{C_{\eps}}) \\
& \leq & \textcolor{black}{\Pb(\nu_T\|\nabla_{\theta_{1}} \frac{1}{T}\Lb_T(\mathbf{F};\theta_{01})\|_2> \frac{a}{2C_{\eps}})+ \Pb(\|\frac{1}{T}\overset{T}{\underset{t=1}{\sum}} \nabla^2_{\theta_{1}F^\top}\ell(F^*_t;\theta_{01})\|_F\sqrt{q+1}\underset{t\leq T}{\max}\|\widehat{f}_{t}-f_{t}\|_{\textcolor{black}{2}}>\frac{a}{2C_{\eps}})}\\
&=:& \textcolor{black}{L_1+L_2.}
\end{eqnarray*}
\textcolor{black}{For any given $q$, recall that $u^{(q)}_t=u_t+\sum_{k>q}\Psi_{0k} x_{t-k}$. We have:
\begin{equation*}
\nabla_{\theta_1}\frac{1}{T}\Lb_T(\mathbf{F};\theta_0)=-\frac{1}{T}\text{vec}\big(\sum^T_{t=1}(x_t-\underline{\Psi}_0Z_{q,t-1})Z^\top_{q,t-1}\big) = -\frac{1}{T}\text{vec}\big(\sum^T_{t=1}(u_t+\sum_{k>q}\Psi_{0k}x_{t-k})Z^\top_{q,t-1}\big).
\end{equation*}
We deduce
\begin{eqnarray*}
\lefteqn{\|\nabla_{\theta_1}\frac{1}{T}\Lb_T(\mathbf{F};\theta_0)\|_2 \leq \|\frac{1}{T}\sum^T_{t=1}u_tZ^\top_{q,t-1}\|_F}\\
&&+ \|\frac{1}{T}\sum^T_{t=1}(\sum_{k>q}\Psi_{0k}x_{t-k})Z^\top_{q,t-1}-\Eb[(\sum_{k>q}\Psi_{0k}x_{t-k})Z^\top_{q,t-1}]\|_F +\|\Eb[(\sum_{k>q}\Psi_{0k}x_{t-k})Z^\top_{q,t-1}]\|_F,
\end{eqnarray*}
and 
\begin{eqnarray*}
\lefteqn{L_1 \leq \Pb(\nu_T\|\frac{1}{T}\sum^T_{t=1}u_tZ^\top_{q,t-1}\|_F> \frac{a}{4C_{\eps}}) }\\
& & +  \Pb(\nu_T\|\frac{1}{T}\sum^T_{t=1}(\sum_{k>q}\Psi_{0k}x_{t-k})Z^\top_{q,t-1}-\Eb[(\sum_{k>q}\Psi_{0k}x_{t-k})Z^\top_{q,t-1}]\|_F> \frac{a}{4C_{\eps}})=:L_{11}+L_{12}.
\end{eqnarray*}
Since $\Eb[u_{j,t}Z_{l,q,t-1}]=0$ for any $j=1,\ldots,m$ and $l=1,\ldots,qm$, we have
\begin{equation*}
L_{11}\leq d_1 \frac{\nu^2_T16C^2_{\eps}}{a^2}\frac{1}{T^2}\sum^{T}_{t,t'=1} \Eb[\max_{1 \leq j \leq m,1 \leq l \leq qm}|u_{j,t}Z_{l,q,t-1}|\max_{1 \leq j \leq m,1 \leq l \leq qm}|u_{j,t'}Z_{l,q,t'-1}|].
\end{equation*}
Furthermore, using $\|AB\|_F \leq \|A\|_s\|B\|_F$, we have
\begin{equation*}
\|\frac{1}{T}\sum^T_{t=1}(\sum_{k>q}\Psi_{0k}x_{t-k})Z^\top_{q,t-1}-\Eb[(\sum_{k>q}\Psi_{0k}x_{t-k})Z^\top_{q,t-1}]\|_F \leq \sum_{k>q}\|\Psi_{0k}\|_s \|\frac{1}{T}\sum^T_{t=1}x_{t-k}Z^\top_{q,t-1}-\Eb[x_{t-k}Z^\top_{q,t-1}]\|_F.
\end{equation*}
We deduce
\begin{eqnarray*}
L_{12} \leq \frac{\nu^2_T16C^2_{\eps}d_1}{a^2}\frac{1}{T^2}\Eb[\Big(\sum_{k>q}\|\Psi_{0k}\|_s\sum^T_{t=1}\underset{1 \leq j \leq m,1 \leq l \leq qm}{\max}|x_{j,t-k}Z_{l,q,t-1}-\Eb[x_{j,t-k}Z_{l,q,t-1}]|\Big)^2].
\end{eqnarray*}
Therefore, under Assumption \ref{assumption_gradient_second_dev}-(ii), we conclude
{\small{\begin{eqnarray*}
\lefteqn{L_{1} \leq d_1 \frac{\nu^2_T16C^2_{\eps}}{a^2}\frac{1}{T^2}\sum^{T}_{t,t'=1} \Eb[\max_{1 \leq j \leq m,1 \leq l \leq qm}|u_{j,t}Z_{l,q,t-1}|\max_{1 \leq j \leq m,1 \leq l \leq qm}|u_{j,t'}Z_{l,q,t'-1}|]}\\
&  & + d_1\frac{\nu^2_T16C^2_{\eps}}{a^2}\frac{1}{T^2} \sum_{k>q}\sum_{k'>q}\|\Psi_{0k}\|_s\|\Psi_{0k'}\|_s \\
& & \times \sum^T_{t,t'=1}\Eb[\underset{1 \leq j \leq m,1 \leq l \leq qm}{\max}|x_{t-k,j}Z_{q,t-1,l}-\Eb[x_{j,t-k}Z_{l,q,t-1}]|\underset{1 \leq j \leq m,1 \leq l \leq qm}{\max}|x_{j,t'-k'}Z_{l,q,t'-1}-\Eb[x_{j,t'-k'}Z_{l,q,t'-1}]|] \\
& \leq & K_0q\frac{\nu^2_T C^2_{\eps}}{a^2}\frac{1}{T} \Big(K_1+K_2\|\mathbf{\Psi}_{0q}\|^2_{s,1}\Big),
\end{eqnarray*}}}
with $K_0,K_1,K_2>0$ finite, $\|\mathbf{\Psi}_ {0q}\|_{s,1}:=\sum_{k>q}\|\Psi_{0k}\|_s$.
 Moreover, under Assumption \ref{assumption_ell_F_control_gradient_hessian}-(i), we have} 
\begin{eqnarray*}
L_2 \leq \textcolor{black}{d_1m(q+1) \frac{\nu^2_T4 C^2_{\eps}}{a^2}  (q+1)C^2_{\kappa}b^2_{T,p}\frac{1}{T^2}\overset{T}{\underset{t,t'=1}{\sum}}\Eb[ \rho(F_t)\rho(F_{t'})]+\Pb(\underset{t\leq T}{\max}\|\widehat{f}_{t}-f_{t}\|_{\textcolor{black}{2}}>C_{\kappa}b_{T,p})},
\end{eqnarray*}
with $C_\kappa >0$ finite. 
Therefore, putting the pieces together, we obtain
\begin{eqnarray*}
\lefteqn{\Pb(\underset{\uu:\|\uu\|_2=C_{\eps}}{\sup}\nu_T|\uu^{\top}\nabla_{\theta_1} \frac{1}{T}\Lb_T(\widehat{\mathbf{F}};\theta_{0,T})| > a) }\\
&\leq & K_0q\frac{\nu^2_T C^2_{\eps}}{a^2}\frac{1}{T}\Big(K_1+K_2\|\mathbf{\Psi}_{0q}\|^2_{s,1}\Big)+ K_3\frac{\nu^2_TC^2_{\eps}}{a^2} q m\textcolor{black}{\frac{1}{T}(q+1)^2C^2_{\kappa}b^2_{T,p}\frac{1}{T}\overset{T}{\underset{t,t'=1}{\sum}}\Eb[ \rho(F_t)\rho(F_{t'})]} + \kappa,
\end{eqnarray*}
where $\kappa \rightarrow 0$, \textcolor{black}{$K_3>0$ finite}. \textcolor{black}{Moreover, since $\|\Eb[(\sum_{k>q}\Psi_{0k}x_{t-k})Z^\top_{q,t-1}]\|_F \leq \sqrt{m(1+qm)}\sum_{k>q}\|\Psi_{0k}\|_se_{k,\max}$ with $e_{k,\max}=\sup_q\max_{1\leq j \leq m,1 \leq l \leq 1+qm}\Eb[|x_{j,t-k}Z_{l,q,t-1}|]$}, then, with $m$ fixed, we have $$T_1 = O_p(\big(\textcolor{black}{\sqrt{q/T}+\sqrt{q/T}\|\mathbf{\Psi}_ {0q}\|_{s,1}+\sqrt{q/T} \,q\,b_{T,p}+\sqrt{q}\sum_{k>q}\|\Psi_{0k}\|_se_{k,\max}}\big)\nu_T) \|\uu\|_2.$$
The \textcolor{black}{term $T_2$} that can be rewritten as $\uu^\top \nabla^2_{\theta_1\theta^\top_1} \frac{1}{T}\Lb_T(\widehat{\mathbf{F}};\theta_{01})\uu = \uu^\top \Eb[\nabla^2_{\theta_1 \theta^\top_1} \ell(F_t;\theta_{01})]\uu + \Rc_T(\theta_{01})$,
where $\Rc_T(\theta_{0,T})= \overset{d_1}{\underset{k,k'=1}{\sum}}\uu_k\uu_{k'}\big\{\partial^2_{\theta_{1,k}\theta_{1,k'}}\frac{1}{T}\Lb_T(\widehat{\mathbf{F}};\theta_{01})-\Eb[\partial^2_{\theta_{1,k}\theta_{1,k'}}\ell(F_t;\theta_{01})]\big\}$. By a Taylor expansion, for any $k,k'=1,\ldots,\textcolor{black}{d_1}$:
\textcolor{black}{\begin{eqnarray*}
\lefteqn{\overset{d_1}{\underset{k,k'=1}{\sum}}\uu_k\uu_{k'}\partial^2_{\theta_{1,k}\theta_{1,k'}}\frac{1}{T}\Lb_T(\widehat{\mathbf{F}};\theta_{01}) }\\
& = & \overset{d_1}{\underset{k,k'=1}{\sum}}\uu_k\uu_{k'}\Big(\partial^2_{\theta_{1,k}\theta_{1,k'}}\frac{1}{T}\Lb_T(\mathbf{F};\theta_{01}) + \frac{1}{T}\overset{T}{\underset{t=1}{\sum}} \overset{m(q+1)}{\underset{l=1}{\sum}}\partial^3_{\theta_{1,k}\theta_{1,k'}F_{l}}\ell(F^*_t;\theta_{01})(\widehat{F}_{t,l}-F_{t,l})\Big)\\
& =: & K_1+K_2,
\end{eqnarray*}}
where $F^*_t=(f^{*\top}_t,\ldots,f^{*\top}_{t-q})^\top$ some random vector s.t. $\|f^*_t-f_t\|_2\leq \|\widehat{f}_t-f_t\|_2$ for any $t$. 
First, consider $K_1$. For any $b>0$:
\begin{eqnarray*}
\Pb(\underset{\uu:\|\uu\|_2=C_{\eps}}{\sup}\nu^2_T/2|\overset{d_1}{\underset{k,k'=1}{\sum}}\uu_k\uu_{k'}\big\{\partial^2_{\theta_{1,k}\theta_{1,k'}}\frac{1}{T}\Lb_T(\mathbf{F};\theta_{01})-\Eb[\partial^2_{\theta_{1,k}\theta_{1,k'}}\ell(F_t;\theta_{01})]\big\}|>b/2) \leq \frac{C_0 \nu^4_TC^4_{\eps} d^2_1}{T b^2},
\end{eqnarray*}
for some constant $C_0>0$ under Assumption \ref{assumption_gradient_second_dev}. Moreover, for $K_2$, by Cauchy-Schwarz inequality, we get for any $b>0$:
\textcolor{black}{
\begin{align*}
\Pb(\underset{\uu:\|\uu\|_2=C_{\eps}}{\sup}\nu^2_T/2 \|\uu\|^2_2 \|\frac{1}{T}\overset{T}{\underset{t=1}{\sum}}  \nabla_{F^\top}\big[\nabla^2_{\theta_{1}\theta^\top_{1}}\ell(F^*_t;\theta_{01})\Big](\widehat{F}_{t}-F_{t})\|_2>b)&\\
\leq C_1\frac{\nu^4_TC^4_{\eps}d^2_1q}{b^2T} C^2_{\kappa}(q+1)^2b^2_{T,p} + \Pb(\max_{t \leq T}\|\widehat{f}_t-f_t\|_2>C_{\kappa}b_{T,p})&,
\end{align*}
with $C_1>0$ a finite constant, under Assumption \ref{assumption_ell_F_control_gradient_hessian}-(ii). Under $q^4b^2_{T,p}=o(T)$}, we conclude:
\begin{eqnarray*}
\lefteqn{\nu^2_T\uu^\top \nabla^2_{\theta_1\theta^\top_1} \frac{1}{T}\Lb_T(\widehat{\mathbf{F}};\theta_{01})\uu}\\
&= &\nu^2_T \uu^\top\Eb[\nabla^2_{\theta_{1}\theta^\top_{1}}\ell(F_t;\theta_{01})]\uu + O_p(\textcolor{black}{\big(\textcolor{black}{\sqrt{q/T}+\sqrt{q/T}\|\mathbf{\Psi}_ {0q}\|_{s,1}+\sqrt{q/T}\,q\,b_{T,p}}+\sqrt{q}\sum_{k>q}\|\Psi_{0k}\|_{s}e_{k,\max}\big)})\|\uu\|^2_2 \nu^2_T .
\end{eqnarray*}
Now $T_1+T_2 = \frac{\nu^2_T}{2}\uu^\top \Hb(\theta_{01}) \uu (1+o_p(1))$ and since the latter term is larger than $C^2_{\eps}\lambda_{\min}(\Hb(\theta_{01}))\nu^2_T/2>0$, we deduce (\ref{bound_obj}) and so $\|\widetilde{\theta}_1-\theta_{01}\|=O_p(\nu_T)$.
\end{proof}
This first-step consistent estimator is plugged in (\ref{obj_crit_first_step}). We now show the existence of a consistent penalized estimator $\widehat{\theta}_1$ of (\ref{obj_crit_first_step}).

\begin{theorem}\label{bound_prob_asym}
\textcolor{black}{Under the conditions of Theorem \ref{bound_proba_first_step_estimator},} there exists a sequence $(\widehat{\theta}_1)$ of solutions of problem (\ref{obj_crit_first_step}) that satisfies 
$$\|\widehat{\theta}_1-\theta_{01}\|_2 = \textcolor{black}{O_p\big(\sqrt{q/T} + \sqrt{q/T}\big(\|\mathbf{\Psi}_ {0q}\|_{s,1}+q\,b_{T,p}\big)+\sqrt{q}\sum_{k>q}\|\Psi_{0k}\|_{s}e_{k,\max}+\sqrt{|\Sc|}\lambda_T\big)},$$
\textcolor{black}{where $\|\mathbf{\Psi}_ {0q}\|_{s,1}:=\sum_{k>q}\|\Psi_{0k}\|_s$, $\Sc:=\{1\leq k \leq d_1: \theta_{01,k}\neq 0\}$, $e_{k,\max}=\underset{q}{\sup}
\underset{1\leq j \leq m}{\underset{1 \leq l \leq qm}{\max}}\Eb[|x_{j,t-k}Z_{l,q,t-1}|]$.}
\end{theorem}

\begin{proof}[Proof of Theorem \ref{bound_prob_asym}.]
The proof follows similar steps as in the proof of Theorem \ref{bound_proba_first_step_estimator}. Define $\textcolor{black}{\alpha_T} = \textcolor{black}{\sqrt{q/T} + \sqrt{q/T}\|\mathbf{\Psi}_ {0q}\|_{s,1}+\sqrt{q/T}\,q\,b_{T,p} + \sqrt{q}\sum_{k>q}\|\Psi_{0k}\|_{s}e_{k,\max}+\sqrt{|\Sc|}\lambda_T}$, and $\Lb^{\text{pen}}_T(\widehat{\mathbf{F}};\theta_1) = \Lb_T(\widehat{\mathbf{F}};\theta_1)+T\lambda_T \sum^{d_1}_{k=1}\tau(\widetilde{\theta}_{1,k}) |\theta_{1,k}| $. We would like to prove that for any $\eps > 0$, there exists $C_{\eps} > 0$ such that $\Pb(\cfrac{1}{\textcolor{black}{\alpha_T}}\|\widehat{\theta}_1 - \theta_{01}\|_2 > C_{\eps}) < \eps$.
We have
\begin{equation*}
\Pb(\cfrac{1}{\textcolor{black}{\alpha_T}} \|\widehat{\theta}_1 - \theta_{01}\|_2 > C_{\eps}) \leq \Pb(\exists \uu, \|\uu\|_2 = C_{\eps}: \Lb^{\text{pen}}_T(\widehat{\mathbf{F}};\theta_{01}+\textcolor{black}{\alpha_T}\uu) \leq \Lb^{\text{pen}}_T(\widehat{\mathbf{F}};\theta_{01})),
\end{equation*}
which implies that there is a minimum in the ball $\{\theta_{01}+\textcolor{black}{\alpha_T}\uu, \|\uu\|_2 \leq C_{\eps}\}$ so that the minimum $\widehat{\theta}_1$ satisfies $\|\widehat{\theta}_1 - \theta_{01}\|_2 = O_p(\textcolor{black}{\alpha_T})$. Now by a Taylor expansion of the penalized loss, we obtain
\begin{eqnarray*}
\lefteqn{\frac{1}{T}\Lb^{\text{pen}}_T(\widehat{\mathbf{F}};\theta_{01}+\textcolor{black}{\alpha_T}\uu)-\frac{1}{T}\Lb^{\text{pen}}_T(\widehat{\mathbf{F}};\theta_{01}) }\\
& \geq &  \textcolor{black}{\alpha_T} \uu^\top \nabla_{\theta_1} \frac{1}{T}\Lb_T(\widehat{\mathbf{F}};\theta_{01}) + \frac{\textcolor{black}{\alpha^2_T}}{2} \uu^\top \nabla^2_{\theta_1 \theta^\top_1} \frac{1}{T}\Lb_T(\widehat{\mathbf{F}};\theta_{01})\uu + \lambda_T \underset{k \in \Sc}{\sum}\tau(\widetilde{\theta}_{1,k}) \big\{|\theta_{01,k}+\textcolor{black}{\alpha_T}\uu_k|-|\theta_{01,k}|\big\}\\
& =: & T_1+T_2+T_3.
\end{eqnarray*}
We want to prove
\begin{eqnarray} \label{bound_obj_pen}
\Pb(\exists \uu, \|\uu\|_2 = C_{\eps} :T_1+T_2+T_3\leq 0) < \eps.
\end{eqnarray}
$T_1,T_2$ can be treated as in the proof of Theorem \ref{bound_proba_first_step_estimator}. As for \textcolor{black}{$T_3$}, note that we have $|\widetilde{\theta}_{1,k}|^{-\gamma}= O_p(1)$ for $k\in\Sc$ by consistency in Theorem \ref{bound_proba_first_step_estimator} and $\theta_{01,k}\neq 0$, so we get
\textcolor{black}{\begin{eqnarray*}
|\lambda_T \underset{k \in \Sc}{\sum}|\widetilde{\theta}_{1,k}|^{-\gamma} \big\{|\theta_{01,k}+\textcolor{black}{\alpha_T}\uu_k|-|\theta_{01,k}|\big\}|\leq\lambda_T \textcolor{black}{\alpha_T} \sqrt{|\Sc|}\|\uu\|_2\max_{k \in\Sc}|\tilde{\theta}_{1,k}|^{-\gamma} \leq \lambda_T\textcolor{black}{\alpha_T}\sqrt{|\Sc|}\|\uu\|_2O_p(1).
\end{eqnarray*}
Therefore, $|T_3| = \|\uu\|_2 O_p(\lambda_T\textcolor{black}{\alpha_T}\sqrt{|\Sc|})\leq \|\uu\|_2 O_p(\alpha^2_T)$.} We deduce
\begin{equation*}
T_1+T_2+T_3 = \frac{\nu^2_T}{2}\uu^\top \Hb(\theta_{01})\uu (1+o_p(1)),
\end{equation*}
and since the latter term is larger than $C^2_{\eps}\lambda_{\min}(\Hb(\theta_{01}))\textcolor{black}{\alpha^2_T}/2>0$, we deduce (\ref{bound_obj_pen}) and so we conclude $\|\widehat{\theta}_1-\theta_{01}\|_2=O_p(\textcolor{black}{\alpha_T})$.
\end{proof}

\textcolor{black}{Hereafter, we will denote $\nu_T=\sqrt{q/T} + \sqrt{q/T}\big(\|\mathbf{\Psi}_ {0q}\|_{s,1}+q\,b_{T,p}\big)+\sqrt{q}\sum_{k>q}\|\Psi_{0k}\|_{s}e_{k,\max}$. We now consider the recovery of the true zero parameters. } 
\begin{theorem}\label{sparsistency}
Under the conditions of Theorem \ref{bound_prob_asym}, assume \textcolor{black}{$\min_{k \in \Sc}|\theta_{01,k}|/\lambda_T\rightarrow \infty$} and that \textcolor{black}{$\sqrt{q}\lambda_T=o(\nu_T)$,} $\textcolor{black}{\nu^{-(1+\gamma)}_T}\lambda_T \rightarrow \infty$ hold, then with probability tending to one, the estimator $\widehat{\theta}_1$ of Theorem \ref{bound_prob_asym} satisfies $\widehat{\theta}_{1\Sc^c}=0$. 
\end{theorem}


\begin{proof}[Proof of Theorem \ref{sparsistency}.]
The proof is performed in the same spirit as in \cite{fan2001}. Consider an estimator $\widehat{\theta}_1=(\widehat{\theta}^\top_{1\Sc},\widehat{\theta}^\top_{1\Sc^c})^\top$ of $\theta_{01}$ such that $\|\widehat{\theta}_1-\theta_{01}\|_2 = O_p(\textcolor{black}{\nu_T})$. The correct identification of the zero entries holds as $T \rightarrow \infty$ when
\begin{equation}\label{sparsistency_obj}
\Lb^{\text{pen}}_T(\widehat{\mathbf{F}};(\widehat{\theta}^\top_{1\Sc},0^\top_{\Sc^c})^\top) = \underset{\|\theta_{1\Ac^c}\|_2\leq C\textcolor{black}{\nu_T}}{\arg\,\min} \; \Lb^{\text{pen}}_T(\widehat{\mathbf{F}};(\widehat{\theta}^\top_{1\Sc},\theta^\top_{1\Sc^c})^\top),
\end{equation}
for any constant $C>0$ with probability tending to one. Define $\eps_T := C\textcolor{black}{\nu_T}$. To prove (\ref{sparsistency_obj}), it is sufficient to show that for any $\theta_1 \in \Theta_1$ such that $\|\theta_1-\theta_{01}\|_2\leq \eps_T$, we have with probability tending to one 
\begin{equation*}
\partial_{\theta_{1,k}}\Lb^{\text{pen}}_T(\widehat{\mathbf{F}};\theta_1)>0 \; \text{when} \; 0 < \theta_{1,k} < \eps_T; \; \partial_{\theta_{1,k}}\Lb^{\text{pen}}_T(\widehat{\mathbf{F}};\theta_1)<0 \; \text{when} \; -\eps_T < \theta_{1,k} < 0,
\end{equation*}
for any $k \in \Sc^c$. By a Taylor expansion of the partial derivative around $\theta_{01}$, we obtain
\begin{eqnarray*}
\lefteqn{\partial_{\theta_{1,k}}\Lb^{\text{pen}}_T(\widehat{\mathbf{F}};\theta_1) = \partial_{\theta_{1,k}}\Lb_T(\widehat{\mathbf{F}};\theta_{1})+T\lambda_T\tau(\widetilde{\theta}_{1,k})\text{sgn}(\theta_{1,k})}\\ 
& = & \partial_{\theta_{1,k}}\Lb_T(\widehat{\mathbf{F}};\theta_{01})+\overset{d_1}{\underset{k'=1}{\sum}}\partial^2_{\theta_{1,k} \theta_{1,k'}}\Lb_T(\widehat{\mathbf{F}};\theta_{01}) (\theta_{1,k'}-\theta_{01,k'})+T\lambda_T\tau(\widetilde{\theta}_{1,k})\text{sgn}(\theta_{1,k})\\
& =: & K_1+K_2+K_3.
\end{eqnarray*}
Using the same steps as in the proof of Theorem \ref{bound_prob_asym}, we get $|\frac{1}{T}K_1| = O_p(\textcolor{black}{\nu_T})$. 
\textcolor{black}{The term $K_2$ can be expanded as:
\begin{align*}
\frac{1}{T}K_2&=\overset{d_1}{\underset{k'=1}{\sum}}\Big(\partial^2_{\theta_{1,k} \theta_{1,k'}}\frac{1}{T}\Lb_T(\widehat{\mathbf{F}};\theta_{01})-\Eb[\partial^2_{\theta_{1,k}\theta_{1,k'}}\ell(F_t;\theta_{01})] \Big)(\theta_{1,k'}-\theta_{01,k'})\\
& +\overset{d_1}{\underset{k'=1}{\sum}}\Eb[\partial^2_{\theta_{1,k}\theta_{1,k'}}\ell(F_t;\theta_{01})](\theta_{1,k'}-\theta_{01,k'})=:K_{21}+K_{22}.
\end{align*}
We have the expansion
\begin{eqnarray*}
\lefteqn{\partial^2_{\theta_{1,k} \theta_{1,k'}}\frac{1}{T}\Lb_T(\widehat{\mathbf{F}};\theta_{01})-\Eb[\partial^2_{\theta_{1,k}\theta_{1,k'}}\ell(F_t;\theta_{01})] }\\
& = & \partial^2_{\theta_{1,k} \theta_{1,k'}}\frac{1}{T}\Lb_T(\mathbf{F};\theta_{01})-\Eb[\partial^2_{\theta_{1,k}\theta_{1,k'}}\ell(F_t;\theta_{01})] +\frac{1}{T}\sum^T_{t=1}\sum^{m(q+1)}_{l=1}\partial^3_{\theta_{1,k}\theta_{1,k'}F_l}\ell(F^*_t;\theta_{01})(\widehat{F}_{t,l}-F_{t,l}).
\end{eqnarray*}
Now by Cauchy–Schwarz inequality and using $\|\theta_1-\theta_{01}\|_2=O_p(\nu_T)$,
\begin{align*}
|K_{21}| &\leq O_p(\nu_T)\Big[\sum^{d_1}_{k'=1}\Big(\partial^2_{\theta_{1,k} \theta_{1,k'}}\frac{1}{T}\Lb_T(\widehat{\mathbf{F}};\theta_{01})-\Eb[\partial^2_{\theta_{1,k}\theta_{1,k'}}\ell(F_t;\theta_{01})] \Big)^2\Big]^{1/2}\\
& + O_p(\nu_T)\Big[\sum^{d_1}_{k'=1}\Big(\frac{1}{T}\sum^T_{t=1}\sum^{m(q+1)}_{l=1}\partial^3_{\theta_{1,k}\theta_{1,k'}F_l}\ell(F^*_t;\theta_{01})(\widehat{F}_{t,l}-F_{t,l})\Big)^2\Big]^{1/2}
\end{align*}
and
\begin{equation*}
|K_{22}| \leq O_p(\nu_T)\Big[\sum^{d_1}_{k'=1}\Big(\Eb[\partial^2_{\theta_{1,k}\theta_{1,k'}}\ell(F_t;\theta_{01})] \Big)^2\Big]^{1/2}.
\end{equation*}
Under Assumption \ref{assumption_regularity_matrix}, $|K_{22}|=O_p(\nu_T)$, and by Assumption \ref{assumption_ell_F_control_gradient_hessian}, $|K_{21}|=O_p(\nu_T)$. Moreover, we have that $|\widetilde{\theta}_{1,k}|\leq \|\widetilde{\theta}_1-\theta_{01}\|_2\leq \eps_T$. Therefore, we deduce 
\begin{equation*}
\partial_{\theta_{1,k}}\Lb^{\text{pen}}_T(\widehat{\mathbf{F}};\theta_1) = T\Big\{O_p(\textcolor{black}{\nu_T}) + \lambda_T\nu^{-\gamma}_TC^{-\gamma}\text{sgn}(\theta_{1,k})\Big\}.
\end{equation*}
Under $\nu_T^{-(1+\gamma)}\lambda_T\rightarrow \infty$, the sign of the derivative is determined by the sign of $\theta_{1,k}$. Thus all the zero components of $\theta_{01}$ will be estimated as zero with probability tending to one.}
\end{proof}

We now consider the second step estimator $\widehat{\theta}_2$. We rely on the following conditions.

\begin{assumption}\label{assumption_regularity_matrix_2}
\textcolor{black}{The parameter space $\Theta_{2}$ is a borelian subset of $\Rb^{d_2}$.} For any $F_t$, the map $g(F_t;\theta_{01};\theta_{02})$ is twice differentiable on $\Theta_{2}$. Let $\nabla^2_{\theta_{2}\theta^\top_{2}}g(F_t;\theta_{01};\theta_{02}) = \big(K_{t-1}(\theta_{01})K_{t-1}(\theta_{01})^\top \otimes I_m\big)$, \textcolor{black}{with $K_{t-1}(\theta_{01})=(1,x^\top_{t-1},u^{(q)\top}_{t-1})^\top$, $u^{(q)}_t = u_t +\sum_{m>q}\Psi_{0k}x_{t-k}$, and $\Gamma_0 = (c^*_0,\Phi_0,\Xi_0) \in \Rb^{m\times (1+2m)}$.}
Moreover, for any $k,k'=1,\ldots,d_2$ and $l=1,\ldots,d_1$, the maps $F_t\mapsto \partial_{\theta_{2,k}} g(F_t;\theta_{01};\theta_{02}), F_t\mapsto \partial^2_{\theta_{2,k}\theta_{2,k'}} g(F_t;\theta_{01};\theta_{02})$, $F_t\mapsto \partial^2_{\theta_{2,k}\theta_{1,j}} g(F_t;\theta_{01};\theta_{02})$ and $F_t\mapsto \partial^3_{\theta_{2,k}\theta_{2,k'}\theta_{1,j}} g(F_t;\theta_{01};\theta_{02})$ are differentiable on $\Rb^{m(q+1)}$. The $d_2 \times d_2$ matrix $\Hb(\theta_{01};\theta_{02}) := \Eb[\nabla^2_{\theta_2 \theta^\top_2}g(F_t;\theta_{01};\theta_{02})]$ exists and $ \exists \gamma_1,\gamma_2: 0 < \gamma_1 < \gamma_2 < \infty$ such that $\gamma_1 < \lambda_{\min}(\Hb(\theta_{01};\theta_{02})) < \lambda_{\max}(\Hb(\theta_{01};\theta_{02})) < \gamma_2$.
\end{assumption}

\begin{assumption}\label{assumption_gradient_second_dev_2}
\begin{itemize}
\item[(i)] \textcolor{black}{There are functions $\chi_r(\cdot),r=1,\ldots,6$, such that for any $T$: 
{\small{\begin{align*}
&\Eb[\underset{1 \leq l,j \leq m}{\max}|u_{l,t}x_{j,t-1}|\underset{1 \leq l,j \leq m}{\max}|u_{l,t'}x_{j,t'-1}|] \leq \chi_1(|t-t'|), \\
& \underset{k>q}{\sup}\,\Eb[\underset{1 \leq j,l \leq m}{\max}|x_{j,t}x_{l,t-k'}-\Eb[x_{j,t}x_{l,t-k'}]|\underset{1 \leq j,l \leq m}{\max}|x_{j,t'}x_{l,t'-k'}-\Eb[x_{j,t'}x_{l,t'-k'}]|]\leq \chi_2(|t-t'|)\\
&\Eb[\underset{1 \leq l,j \leq m}{\max}|u_{l,t}u_{j,t-1}|\underset{1 \leq l,j \leq m}{\max}|u_{l,t'}u_{j,t'-1}|] \leq \chi_3(|t-t'|), \\
&\underset{k,k'>q}{\sup}\Eb[\underset{1 \leq j,l \leq m}{\max}|x_{j,t-k}x_{l,t-k'}-\Eb[x_{j,t-k}x_{l,t-k'}]|\underset{1 \leq j,l \leq m}{\max}|x_{j,t'-k}x_{l,t'-k'}-\Eb[x_{j,t'-k}x_{l,t'-k'}]|]\leq \chi_4(|t-t'|), \\
& \sup_{k,k'>q}\Eb[\underset{1 \leq l,j \leq m}{\max}|u_{l,t}x_{j,t-1-k}|\underset{1 \leq l,j\leq m}{\max}|u_{l,t'}x_{j,t'-1-k'}|] \leq \chi_5(|t-t'|),\\
& \sup_{k,k'>q}\Eb[\underset{1 \leq l,j \leq m}{\max}|u_{l,t-1}x_{j,t-1-k}|\underset{1 \leq l,j \leq m}{\max}|u_{l,t-1'}x_{j,t'-1-k'}|] \leq \chi_6(|t-t'|),
\end{align*}}}
and $\underset{T>0}{\sup}\; \frac{1}{T}\overset{T}{\underset{t,t'=1}{\sum}} \chi_r(|t-t'|)<\infty$, $r=1,\ldots,6$.}
\item[(ii)] \textcolor{black}{Let $\partial^2_{\theta_{2,k}\theta_{2,l}}g(F_t;\theta_{01};\theta_{02}) = \big(K_{t-1}(\theta_{01})K_{t-1}(\theta_{01})^\top \otimes I_m\big)_{kl}$, $k,l=1,\cdots,d_2$, and $\nu_{kl,t} = \partial^2_{\theta_{2,k} \theta_{2,l}} g(F_t;\theta_{01};\theta_{02}) - \Eb[\partial^2_{\theta_{2,k} \theta_{2,l}} g(F_t;\theta_{01};\theta_{02})]$. Then there exists some function $\Phi_1(\cdot)$ such that for any $T$: $|\Eb[\nu_{kl,t} \nu_{kl,t'}]| \leq \Phi_1(|t-t'|), \;\; \text{and} \;\; \underset{T>0}{\sup} \; \frac{1}{T}\overset{T}{\underset{t,t'=1}{\sum}} \Phi_1(|t-t'|) <\infty$.}
\end{itemize}
\end{assumption}

\begin{assumption}\label{assumption_ell_F_control_gradient_2}
\begin{itemize}
    \item[(i)] Let $v_t(C) = \underset{\underset{1\leq l \leq d_1}{1\leq k \leq d_2}}{\sup}\underset{\theta_1:\|\theta_1-\theta_{01}\|\leq C\textcolor{black}{\nu_T}}{\sup}\partial^2_{\theta_{2,k}\theta_{1,l}}g(F_t;\theta_1;\theta_{02})$, with $C>0$ a finite constant. Then $\frac{1}{T}\overset{T}{\underset{t,t'=1}{\sum}} \Eb[v_t(C)v_{t'}(C)] <\infty$.
    \item[(ii)] There exists some measurable function $\kappa_1(\cdot)$ such that:
    {\small{\begin{equation*}
    \underset{\underset{1\leq l\leq m(q+1)}{1 \leq k\leq d_2}}{\sup}\underset{U_t: \|U_t-F_t\|_2\leq \textcolor{black}{L\sqrt{q}b_{T,p}}}{\sup} |\partial^2_{\theta_{2,k}F_{l}}g(U_t;\theta_{01};\theta_{02})| \leq \kappa_{1}(F_t), \; \text{and} \;\underset{T>0}{\sup} \; \frac{1}{T}\overset{T}{\underset{t,t'=1}{\sum}} \Eb[\kappa_{1}(F_t)\kappa_{1}(F_{t'})] <\infty, \; \textcolor{black}{\text{with} \; L>0}.
    \end{equation*}}}
    \item[(iii)] Let $C>0, L>0$. There exists some measurable function $\kappa_2(\cdot)$ such that:  
    \begin{equation*}
    \underset{\underset{1\leq j \leq d_1,1\leq l\leq m(q+1)}{1\leq k \leq d_2}}{\sup}\underset{\theta_1:\|\theta_1-\theta_{01}\|_2\leq C\textcolor{black}{\nu_T}}{\sup}\underset{U_t: \|U_t-F_t\|_2\leq \textcolor{black}{L\sqrt{q}b_{T,p}}}{\sup} |\partial^3_{\theta_{2,k}\theta_{1,j}F_{l}}g(U_t;\theta_{1};\theta_{02})| \leq \kappa_{2}(F_t),  
    \end{equation*}
    and $\underset{T>0}{\sup} \; \frac{1}{T}\overset{T}{\underset{t,t'=1}{\sum}} \Eb[\kappa_{2}(F_t)\kappa_{2}(F_{t'})] <\infty$.
\end{itemize}
\end{assumption}

\begin{assumption}\label{assumption_ell_F_control_hessian_2}
\begin{itemize}
    \item[(i)] Let $h_t(C) = \underset{\underset{1\leq l \leq d_1}{1\leq k,k' \leq d_2}}{\sup}\underset{\theta_1:\|\theta_1-\theta_{01}\|_2\leq C\textcolor{black}{\nu_T}}{\sup}\partial^3_{\theta_{2,k}\theta_{2,k'}\theta_{1,l}}g(F_t;\theta_1;\theta_{02})$, with $C>0$ a finite constant. Then $\frac{1}{T}\overset{T}{\underset{t,t'=1}{\sum}} \Eb[h_t(C)h_{t'}(C)] <\infty$.
    \item[(ii)] Let $L>0$. There exists some measurable function $\omega_1(\cdot)$ such that:  
    \begin{equation*}
    \underset{\underset{1\leq l\leq m(q+1)}{1 \leq k,k'\leq d_2}}{\sup}\underset{U_t: \|U_t-F_t\|_2\leq \textcolor{black}{L\sqrt{q}b_{T,p}}}{\sup} |\partial^3_{\theta_{2,k}\theta_{2,k'}F_{l}}g(U_t;\theta_{01};\theta_{02})| \leq \omega_1(F_t), \; \text{and} \;\underset{T>0}{\sup} \; \frac{1}{T}\overset{T}{\underset{t,t'=1}{\sum}} \Eb[\omega_1(F_t)\omega_1(F_{t'})] <\infty. 
    \end{equation*}
    \item[(iii)] Let $C>0, L>0$. There exists some measurable function $\omega_2(\cdot)$ such that:  
    \begin{equation*}
    \underset{\underset{1\leq j \leq d_1,1\leq l\leq m(q+1)}{1\leq k,k' \leq d_2}}{\sup}\underset{\theta_1:\|\theta_1-\theta_{01}\|_2\leq C\textcolor{black}{\nu_T}}{\sup}\underset{U_t: \|U_t-F_t\|_2\leq \textcolor{black}{L\sqrt{q}b_{T,p}}}{\sup} |\partial^4_{\theta_{2,k}\theta_{2,k'}\theta_{1,j}F_{l}}g(U_t;\theta_{1};\theta_{02})| \leq \omega_2(F_t),
    \end{equation*}
    and $\underset{T>0}{\sup} \; \frac{1}{T}\overset{T}{\underset{t,t'=1}{\sum}} \Eb[\omega_2(F_t)\omega_2(F_{t'})] <\infty$.
\end{itemize}
\end{assumption}

\begin{theorem}\label{bound_proba_second_step_estimator}
Under the conditions of Theorem \ref{sparsistency}, under Assumptions \ref{assumption_regularity_matrix_2}-\ref{assumption_ell_F_control_hessian_2}, there exists a sequence $(\widehat{\theta}_2)$ of (\ref{obj_crit_second}) that satisfies 
\textcolor{black}{$$\|\widehat{\theta}_2-\theta_{02}\|_2 = O_p(T^{-1/2}+T^{-1/2}(\|\mathbf{\Psi}_{0q}\|_{s,1}+\|\mathbf{\Psi}_{0q}\|^2_{s,1})+B(\mathbf{\Psi}_{0q})+\nu_T),$$ 
with $\nu_T$ defined in Theorem \ref{sparsistency} and $B(\mathbf{\Psi}_{0q})$ defined by
\begin{align*}
B(\mathbf{\Psi}_{0q}):=&\underset{k>q}{\sum}\|\Psi_{0k}\|_s(\underset{1 \leq j \leq m}{\max}|\Eb[x_{j,t-1-k}]|\vee\underset{1 \leq j,l \leq m}{\max}|\Eb[x_{j,t-1-k}x_{l,t-1}]|)\\
&+\sum_{k>q}\sum_{k'>q}\|\Psi_{0k}\|_s\|\Psi_{0k'}\|_s\underset{1\leq j,l\leq m}{\max}|\Eb[x_{j,t-1-k}x_{l,t-1-k'}]|.
\end{align*}
}
\end{theorem}

\begin{proof}[Proof of Theorem \ref{bound_proba_second_step_estimator}.]
We follow the same steps as in the proof of Theorem \ref{bound_proba_first_step_estimator}. Let us define \textcolor{black}{$\pi_T = T^{-1/2}+T^{-1/2}(\|\mathbf{\Psi}_{0q}\|_{s,1}+\|\mathbf{\Psi}_{0q}\|^2_{s,1})+B(\mathbf{\Psi}_{0q})+\nu_T$, with $\nu_T$ defined in Theorem \ref{sparsistency}, and 
\begin{align*}
B(\mathbf{\Psi}_{0q}):=&\sum_{k>q}\|\Psi_{0k}\|_s(\underset{1 \leq j \leq m}{\max}|\Eb[x_{j,t-1-k}]|\vee\underset{1 \leq j,l \leq m}{\max}|\Eb[x_{j,t-1-k}x_{l,t-1}]|)\\
&+\sum_{k>q}\sum_{k'>q}\|\Psi_{0k}\|_s\|\Psi_{0k'}\|_s\underset{1\leq j,l\leq m}{\max}|\Eb[x_{j,t-1-k}x_{l,t-1-k'}]|.
\end{align*}}
\textcolor{black}{We aim to show that for any $\pi>0$, there exists $C_{\pi}>0$ such that $\Pb(\cfrac{1}{\pi_T} \|\widehat{\theta}_2 - \theta_{02}\|_2 > C_{\pi})<\pi$. We have
\begin{equation*}
\Pb(\cfrac{1}{\pi_T} \|\widehat{\theta}_2 - \theta_{02}\|_2 > C_{\pi}) \leq \Pb(\exists \uu, \|\uu\|_2 = C_{\pi}: \Gb_T(\widehat{\mathbf{F}};\widehat{\theta}_1;\theta_{02}+\pi_T\uu) \leq \Gb_T(\widehat{\mathbf{F}};\widehat{\theta}_1;\theta_{02})),
\end{equation*}
which implies that there is a minimum in the ball $\{\theta_{02}+\pi_T\uu, \|\uu\|_2 \leq C_{\pi}\}$} so that the minimum $\widehat{\theta}_2$ satisfies $\|\widehat{\theta}_2 - \theta_{02}\|_2 = O_p(\pi_T)$. Now by a Taylor expansion of the loss, we obtain
\begin{eqnarray*}
\lefteqn{\frac{1}{T}\Gb_T(\widehat{\mathbf{F}};\widehat{\theta}_1;\theta_{02}+\pi_T\uu)-\frac{1}{T}\Gb_T(\widehat{\mathbf{F}};\widehat{\theta}_1;\theta_{02})}\\
& = & \pi_T \uu^\top \nabla_{\theta_2} \frac{1}{T}\Gb_T(\widehat{\mathbf{F}};\widehat{\theta}_1;\theta_{02}) + \frac{\pi^2_T}{2} \uu^\top \nabla^2_{\theta_2 \theta^\top_2} \frac{1}{T}\Gb_T(\widehat{\mathbf{F}};\widehat{\theta}_1;\theta_{02})\uu=:T_1+T_2.
\end{eqnarray*}
We want to prove
\begin{equation} \label{bound_obj_second_step}
\Pb(\exists \uu, \|\uu\|_2 = C_{\textcolor{black}{\pi}} : T_1+T_2 \leq 0) < \pi.
\end{equation}
Let us consider the term \textcolor{black}{$T_1$}. For any $k=1,\ldots,d_2$, it can be expanded as follows:
{\small{\begin{eqnarray*}
\lefteqn{\partial_{\theta_{2,k}} \frac{1}{T}\Gb_T(\widehat{\mathbf{F}};\widehat{\theta}_{1};\theta_{02}) = \partial_{\theta_{2,k}} \frac{1}{T}\Gb_T(\widehat{\mathbf{F}};\theta_{01};\theta_{02}) + \frac{1}{T}\overset{d_1}{\underset{j=1}{\sum}}\partial^2_{\theta_{2,k}\theta_{1,j}} \Gb_T(\widehat{\mathbf{F}};\theta^*_1;\theta_{02})(\widehat{\theta}_{1,j}-\theta_{01,j})}\\
& = & \partial_{\theta_{2,k}}\frac{1}{T}\Gb_T(\mathbf{F};\theta_{01};\theta_{02}) + \frac{1}{T}\overset{T}{\underset{t=1}{\sum}}\overset{m(q+1)}{\underset{l=1}{\sum}} \partial^2_{\theta_{2,k}F_l} g(F^*_t;\theta_{01};\theta_{02})(\widehat{F}_{t,l}-F_{t,l})  + \frac{1}{T}\overset{d_1}{\underset{j=1}{\sum}}\partial^2_{\theta_{2,k}\theta_{1,j}} \Gb_T(\mathbf{F};\theta^*_1;\theta_{02})(\widehat{\theta}_{1,j}-\theta_{01,j}) \\
& + & \frac{1}{T}\overset{T}{\underset{t=1}{\sum}}\overset{d_1}{\underset{j=1}{\sum}}\overset{m(q+1)}{\underset{l=1}{\sum}}\partial^3_{\theta_{2,k}\theta_{1,j}F_l} g(F^*_t;\theta^*_1;\theta_{02})(\widehat{\theta}_{1,j}-\theta_{01,j})(\widehat{F}_{t,l}-F_{t,l}),
\end{eqnarray*}}}
with $\|\theta^*_1-\theta_{01}\|_2\leq\|\widehat{\theta}_1-\theta_{01}\|_2$ and $\|f^*_t-f_t\|_2 \leq \|\widehat{f}_t-f_t\|_2, t = 1,\ldots,T$. 
For the first term of the expansion, we have $\nabla_{\theta_2}g(F_t;\theta_{01};\theta_{02}) = -\text{vec}(w_tK_{t-1}(\theta_{01})^\top)$ with $w_t=x_t-c^*-\Phi x_{t-1}-\Xi u^{(q)}_{t-1}$. Moreover, $w_t = x_t-c^*_0-\Phi_0 x_{t-1}-\Xi_0 u_{t-1} - \Xi_0 (u^{(q)}_{t-1}-u_{t-1})= u_t-\Xi_0(u^{(q)}_{t-1}-u_{t-1})$ using the VARMA(1,1) representation. 
Therefore, we have
\textcolor{black}{\begin{eqnarray*}
\lefteqn{\|\nabla_{\theta_2} \frac{1}{T}\Gb_T(\mathbf{F};\theta_{01};\theta_{02})\|_2\leq \|\nabla_{\theta_2} \frac{1}{T}\Gb_T(\mathbf{F};\theta_{01};\theta_{02})-\Eb[\nabla_{\theta_2} \frac{1}{T}\Gb_T(\mathbf{F};\theta_{01};\theta_{02})]\|_2+\|\Eb[\nabla_{\theta_2} \frac{1}{T}\Gb_T(\mathbf{F};\theta_{01};\theta_{02})]\|_2}\\
&\leq & \|\frac{1}{T}\sum^T_{t=1}w_t-\Eb[w_t]\|_2+ \|\frac{1}{T}\sum^T_{t=1}w_tx^\top_{t-1}-\Eb[w_tx^\top_{t-1}]\|_F + \|\frac{1}{T}\sum^T_{t=1}w_tu^{(q)\top}_{t-1}-\Eb[w_tu^{(q)\top}_{t-1}]\|_F\\
&&+ \|\Eb[\nabla_{\theta_2} \frac{1}{T}\Gb_T(\mathbf{F};\theta_{01};\theta_{02})]\|_2.
\end{eqnarray*}}
We have
\begin{align*}
\|\frac{1}{T}\sum^T_{t=1}w_t-\Eb[w_t]\|_2&\leq  \|\frac{1}{T}\sum^T_{t=1}u_{t}\|_2+\|\Xi_0\|_s\|\frac{1}{T}\sum^T_{t=1}\big(\sum_{k>q}\Psi_{0k}\big(x_{t-1-k}-\Eb[x_{t-1-k}]\big)\big)\|_2.
\end{align*}
Furthermore, 
\begin{eqnarray*}
\lefteqn{\|\frac{1}{T}\sum^T_{t=1}w_tx^\top_{t-1}-\Eb[w_tx^\top_{t-1}]\|_F}\\
& \leq& \|\frac{1}{T}\sum^T_{t=1}u_{t}x^\top_{t-1}-\Eb[u_{t}x^\top_{t-1}]\|_F+\|\Xi_0\|_s\|\frac{1}{T}\sum^T_{t=1}\big(\sum_{k>q}\Psi_{0k}\big(x_{t-1-k}x^\top_{t-1}-\Eb[x_{t-1-k}x^\top_{t-1}]\big)\|_F,
\end{eqnarray*}
and
\begin{eqnarray*}
\lefteqn{\|\frac{1}{T}\sum^T_{t=1}w_tu^{(q)\top}_{t-1}-\Eb[w_tu^{(q)\top}_{t-1}]\|_F \leq \|\frac{1}{T}\sum^T_{t=1}u_{t}u^\top_{t-1}-\Eb[u_{t}u^\top_{t-1}]\|_F} \\
& & + \|\frac{1}{T}\sum^T_{t=1}u_t(\sum_{k>q}\Psi_{0k}x_{t-1-k})^\top-\Eb[u_t(\sum_{k>q}\Psi_{0k}x_{t-1-k})^\top]\|_F\\
&&+ \|\Xi_0\|_s\|\frac{1}{T}\sum^T_{t=1}\sum_{k>q}\Psi_{0k}x_{t-1-k}u^\top_{t-1}-\Eb[\sum_{k>q}\Psi_{0k}x_{t-1-k}u^\top_{t-1}]\|_F \\
&& + \|\Xi_0\|_s\|\frac{1}{T}\sum^T_{t=1}\sum_{k>q}\Psi_{0k}x_{t-1-k}(\sum_{k'>q}\Psi_{0k'}x_{t-1-k'})^\top-\Eb[\sum_{k>q}\Psi_{0k}x_{t-1-k}(\sum_{k'>q}\Psi_{0k'}x_{t-1-k'})^\top]\|_F.
\end{eqnarray*}
\textcolor{black}{We now consider $\|\Eb[\nabla_{\theta_2}g(F_t;\theta_{01};\theta_{02})]\|_F$. We have
\begin{equation*}
\|\Eb[\nabla_{\theta_2}g(F_t;\theta_{01};\theta_{02})]\|_F \leq \|\Eb[w_t]\|_2+\|\Eb[w_tx^\top_{t-1}]\|_F+\|\Eb[w_tu^{(q)}_{t-1}]\|_F=A_1+A_2+A_3.
\end{equation*}
Since $\Eb[u_t]=0$, we get $A_1 \leq \|\Xi_0\|_s m^{1/2}\sum_{k>q}\|\Psi_{0k}\|_s\max_{1 \leq j \leq m}|\Eb[x_{j,t-1-k}]|$.
As for $A_2$, we obtain $A_2 \leq \|\Xi_0\|_s m\sum_{k>q}\|\Psi_{0k}\|_s\max_{1 \leq j,l \leq m}|\Eb[x_{j,t-1-k}x_{l,t-1}]|$.
Finally, we have
\begin{align*}
A_3 &\leq \|\Xi_0\|_sm\sum_{k>q}\|\Psi_{0k}\|_s\max_{1\leq j,l\leq m}|\Eb[x_{j,t-1-k}u_{l,t-1}]|\\
&+\|\Xi_0\|_sm\sum_{k>q}\sum_{k'>q}\|\Psi_{0k}\|_s\|\Psi_{0k'}\|_s\max_{1\leq j,l\leq m}|\Eb[x_{j,t-1-k}x_{l,t-1-k'}]|.
\end{align*}
Since $\|\Xi_0\|_s=O(1)$, and $\Eb[x_{j,t-1-k}u_{l,t-1}]=0$ for $k>0$, we deduce:
\begin{align*}
\|\Eb[\nabla_{\theta_2}g(F_t;\theta_{01};\theta_{02})]\|_F &\leq C_0\Big[\sum_{k>q}\|\Psi_{0k}\|_s\max_{1 \leq j \leq m}|\Eb[x_{j,t-1-k}]|+\sum_{k>q}\|\Psi_{0k}\|_s\max_{1 \leq j,l \leq m}|\Eb[x_{j,t-1-k}x_{l,t-1}]|\\
&+
\sum_{k>q}\sum_{k'>q}\|\Psi_{0k}\|_s\|\Psi_{0k'}\|_s\max_{1\leq j,l\leq m}|\Eb[x_{j,t-1-k}x_{l,t-1-k'}]|\Big],
\end{align*}
with $C_0>0$ sufficiently large.
Under Assumption \ref{assumption_gradient_second_dev_2}-(i), we deduce
\begin{align*}
\pi_T \|\nabla_{\theta_2} \frac{1}{T}\Gb_T(\mathbf{F};\theta_{01};\theta_{02})\|_2\|\uu\|_2&=O_p(\pi_T \Big\{T^{-1/2}+T^{-1/2}\|\mathbf{\Psi}_{0q}\|_{s,1}+T^{-1/2}\|\mathbf{\Psi}_{0q}\|^2_{s,1}\\
&+\sum_{k>q}\|\Psi_{0k}\|_s(\max_{1 \leq j \leq m}|\Eb[x_{j,t-1-k}]|\vee\max_{1 \leq j,l \leq m}|\Eb[x_{j,t-1-k}x_{l,t-1}]|)\\
&+\sum_{k>q}\sum_{k'>q}\|\Psi_{0k}\|_s\|\Psi_{0k'}\|_s\max_{1\leq j,l\leq m}|\Eb[x_{j,t-1-k}x_{l,t-1-k'}]|\Big\} )\|\uu\|_2.
\end{align*}}
As for the second term, for any $a>0$, we get
\begin{eqnarray*}
\Pb(\underset{\uu:\|\uu\|_2=C_{\pi}}{\sup}\pi_T\|\uu\|_2 \|\frac{1}{T}\overset{T}{\underset{t=1}{\sum}} \nabla^2_{\theta_2F^\top}g(F^*_t;\theta_{01};\theta_{02})(\widehat{F}_{t}-F_{t})\|_2>a/4)
& \leq & C_1\frac{C^2_{\pi}\pi^2_T}{a^2} d_2\frac{1}{T}(q+1)^2 b^2_{T,p},
\end{eqnarray*}
with $C_1>0$.
By assumption \ref{assumption_ell_F_control_gradient_2}, under $\|\widehat{\theta}_1-\theta_{01}\|_2 \leq C_\eps \nu_T$ with probability tending to one, the third term can be bounded as:
\begin{eqnarray*}
\Pb(\underset{\uu:\|\uu\|_2=C_{\pi}}{\sup}\nu_T\|\uu\|_2 \|\frac{1}{T}\overset{T}{\underset{t=1}{\sum}} \nabla^2_{\theta_2\theta^\top_1}g(F_t;\theta^*_{1};\theta_{02})(\widehat{\theta}_1-\theta_{01})\|_2>a/4)
& \leq & C_2\frac{C^2_{\pi}\pi^2_T}{a^2} d_2\frac{q}{T} C^2_{\epsilon}\nu^2_T.
\end{eqnarray*}
with $C_2>0$ finite. Finally, let $A_{k,t}:=\overset{d_1}{\underset{j=1}{\sum}}\overset{m(q+1)}{\underset{l=1}{\sum}}\partial^3_{\theta_{2,k}\theta_{1,j}F_l} g(F^*_t;\theta^*_1;\theta_{02})(\widehat{\theta}_{1,j}-\theta_{01,j})(\widehat{F}_{t,l}-F_{t,l})$, for $k=1,\ldots,d_2$ and let $A_t =(A_{k,t})_{1 \leq k \leq d_2}$. Then, for the last quantity of the expansion, we have
\begin{eqnarray*}
\Pb(\underset{\uu:\|\uu\|_2=C_{\pi}}{\sup}\pi_T\|\uu\|_2 \|\frac{1}{T}\sum^T_{t=1}A_{t}\|_2>a/4) \leq C_3\frac{C^2_{\pi}\pi^2_T}{a^2} d_2\frac{q}{T}(q+1)b^2_{T,p} C^2_{\epsilon}\nu^2_T,
\end{eqnarray*}
with $C_3>0$ finite.
Consequently, we deduce $|T_1|=O_p(\|\uu\|_2\pi^2_T)$.

\textcolor{black}{Similarly, to study the term $T_2$, we consider the expansion
\begin{eqnarray*}
\lefteqn{\partial^2_{\theta_{2,k} \theta_{2,k'}} \frac{1}{T}\Gb_T(\widehat{\mathbf{F}};\widehat{\theta}_{1};\theta_{02}) - \Eb[\partial^2_{\theta_{2,k} \theta_{2,k'}} g(F_t;\theta_{01};\theta_{02})]}\\
& = & \partial^2_{\theta_{2,k} \theta_{2,k'}} \frac{1}{T}\Gb_T(\widehat{\mathbf{F}};\theta_{01};\theta_{02})-\Eb[\partial^2_{\theta_{2,k} \theta_{2,k'}} g(F_t;\theta_{01};\theta_{02})]+ \frac{1}{T}\overset{d_1}{\underset{j=1}{\sum}}\partial^3_{\theta_{2,k} \theta_{2,k'}\theta_{1,j}}\Gb_T(\widehat{\mathbf{F}};\theta^*_1;\theta_{02})(\widehat{\theta}_{1,j}-\theta_{01,j})\\
& = & \partial^2_{\theta_{2,k} \theta_{2,k'}} \frac{1}{T}\Gb_T(\mathbf{F};\theta_{01};\theta_{02})-\Eb[\partial^2_{\theta_{2,k} \theta_{2,k'}} g(F_t;\theta_{01};\theta_{02})]\\
&+& \frac{1}{T}\overset{T}{\underset{t=1}{\sum}}\overset{m(q+1)}{\underset{l=1}{\sum}} \partial^3_{\theta_{2,k}\theta_{2,k'} F_l} g(F^*_t;\theta_{01};\theta_{02})(\widehat{F}_{t,l}-F_{t,l})+ \frac{1}{T}\overset{d_1}{\underset{j=1}{\sum}}\partial^3_{\theta_{2,k}\theta_{2,k'}\theta_{1,j}} \Gb_T(\mathbf{F};\theta^*_1;\theta_{02})(\widehat{\theta}_{1,j}-\theta_{01,j})\\
& + & \frac{1}{T}\overset{T}{\underset{t=1}{\sum}}\overset{d_1}{\underset{j=1}{\sum}}\overset{m(q+1)}{\underset{l=1}{\sum}}\partial^4_{\theta_{2,k}\theta_{2,k'}\theta_{1,j}F_l} g(F^*_t;\theta^*_1;\theta_{02})(\widehat{\theta}_{1,j}-\theta_{01,j})(\widehat{F}_{t,l}-F_{t,l}).
\end{eqnarray*}}
Using the same reasoning as in \textcolor{black}{$T_1$}, under Assumption \ref{assumption_ell_F_control_hessian_2}, we obtain
\begin{equation*}
T_2 = \frac{\pi^2_T}{2}\uu^\top\Eb[\nabla^2_{\theta_2\theta^\top_2}g(F_t;\theta_{01};\theta_{0,2})]\uu +O_p(\pi^2_T\|\uu\|^2_2).
\end{equation*}
Hence, we conclude $T_1+T_2 = \frac{\pi^2_T}{2}\uu^\top \Hb(\theta_{01};\theta_{02})\uu(1+o_p(1))$. The latter term is larger than $C^2_\eps\lambda_{\min}(\Hb(\theta_{01};\theta_{02}))/2>0$, which implies $\|\widehat{\theta}_{2}-\theta_{02}\|=O_p(\pi_T)$.
\end{proof}

\section{Factor estimation}\label{appendix_factor_estim}

\textcolor{black}{Throughout this paper, we work with the GLS estimator $\widehat{f}_t$ of the factor variable $f_t$ that satisfies $\widehat{f}_t=(\widehat{\Lambda}^\top \widehat{\Sigma}_\varepsilon^{-1} \widehat{\Lambda})^{-1} 
\widehat{\Lambda}^\top \widehat{\Sigma}_\varepsilon^{-1} y_t$ for any $t$. By Theorem 5.1 of \cite{bai2012}, $\|\widehat{\Lambda}-\Lambda_0\|_F=O_p(\sqrt{\frac{p}{T}})$ and $\|\widehat{\Sigma}_{\varepsilon}-\Sigma_{0\varepsilon}\|_F=O_p(\sqrt{\frac{p}{T}})$, and by Theorem 6.1, $\|\widehat{f}_t-f_t\|_2=O_p(\sqrt{\frac{1}{p}})$ for any $t$. Recall that the parameter $\Sigma_{\varepsilon}$ is a diagonal matrix, with diagonal element being bounded (as in Assumption C in \cite{bai2012}). Let us consider the consistency of $\widehat{f}_t$ in a uniform sense, which is specified in Assumption \ref{assumption_unif_consistency_factors}. Assume the following conditions:
\begin{assumption}\label{assumption_factor_moments}
\begin{itemize}
    \item[(i)] $\exists M>0$ such that $\sup_t \Eb[m^{-1}\|f_t\|^4_2]<M$.
    \item[(ii)] $\exists M>0$ such that $\sup_t \sup_{1 \leq k \leq p}\Eb[|\varepsilon_{k,t}|^4] < M$
    \item[(iii)] $\exists M>0$ such that: $\sup_t\Eb[\|\frac{1}{\sqrt{p}}\sum^p_{k=1}\Lambda_{k,0}\varepsilon_{k,t}\|^4_2]<M$, with $\Lambda_{k,0}$ the $k$-th row of $\Lambda_0$.
\end{itemize}
\end{assumption}
Condition (i) implies the existence of the fourth order moments of the factor variables. Based on equation (\ref{eq:yt}) of the fMSV model, $\Eb[f^4_{k,t}] = \Eb[\exp(2h_{k,t})\zeta^4_{k,t}]$, and since $(\zeta_t)$ is independent of $(h_t)$, the moment condition requires $\Eb[\zeta^4_{k,t}] < \infty$ and $\Eb[\exp(2h_{k,t})] < \infty$. Condition (ii) requires that the idiosyncratic variable $\varepsilon_t$ has finite fourth moments. Finally, condition (iii) is verified under pervasive factors, that is the situation where $\lambda_{\max}(p^{-1}\Lambda^\top_0\Lambda_0) \leq \delta$ with $\delta>0$ some constant. Under IC2 and since $\|\Sigma_{\varepsilon}\|_s=O(1)$, $\Lambda_0$ also satisfies the pervasive condition.
First, note that
\begin{equation*}
\forall t, \; \widehat{f}_t-f_t= -(\widehat{\Lambda}^\top\widehat{\Sigma}^{-1}_{\varepsilon} \widehat{\Lambda})^{-1} 
\widehat{\Lambda}^\top \widehat{\Sigma}^{-1}_{\varepsilon}(\widehat{\Lambda}-\Lambda_0)f_t + (\widehat{\Lambda}^\top\widehat{\Sigma}^{-1}_{\varepsilon} \widehat{\Lambda})^{-1} 
\widehat{\Lambda}^\top \widehat{\Sigma}^{-1}_{\varepsilon} \varepsilon_t =:M_{1,t}+M_{2,t}.
\end{equation*}
Take $M_{1,t}$:
\begin{eqnarray*}
\max_{t\leq T}\|M_{1,t}\| \leq O_p(p^{-1})O_p(p^{1/2})O_p(1)\|\widehat{\Lambda}-\Lambda_0\|_F\max_{t\leq T}\|f_t\|_2 \leq O_p(\sqrt{\frac{1}{T}})\max_{t\leq T}\|f_t\|_2,
\end{eqnarray*}
as $\|\widehat{\Lambda}\|_F=O_p(p^{1/2})$, $\|\widehat\Sigma^{-1}_{\varepsilon}\|_s=O_p(1)$. Note that the previous inequality also follows from the condition for identification $p^{-1}\widehat{\Lambda}^\top \widehat{\Sigma}^{-1}_{\varepsilon}\widehat{\Lambda}=I_m$. More generally, the pervasive factor condition $\lambda_{\max}(p^{-1}\widehat{\Lambda}^\top\widehat{\Lambda}) \leq \delta$, with $\delta>0$ some constant, would also imply this inequality. Under Assumption \ref{assumption_factor_moments}-(i), $\max_{t\leq T}\|f_t\|_2 =O_p(T^{1/4})$, which implies $\max_{t\leq T}|M_{1,t}|=O_p(T^{-1/4})$.
Take $M_{2,t}$:
\begin{equation*}
M_{2,t} = (\widehat{\Lambda}^\top\widehat{\Sigma}^{-1}_{\varepsilon} \widehat{\Lambda})^{-1} 
\Lambda^{\top}_0 \Sigma^{-1}_{0,\varepsilon} \varepsilon_t + (\widehat{\Lambda}^\top\widehat{\Sigma}^{-1}_{\varepsilon} \widehat{\Lambda})^{-1} 
(\widehat{\Lambda}^\top \widehat{\Sigma}^{-1}_{\varepsilon}-\Lambda^{\top}_0 \Sigma^{-1}_{0,\varepsilon})\varepsilon_t =: L_{1,t}+L_{2,t}.
\end{equation*}
We have:
\begin{eqnarray*}
\lefteqn{\|L_{2,t}\|_2 \leq  O_p(p^{-1})\|(\widehat{\Lambda}^\top \widehat{\Sigma}^{-1}_{\varepsilon}-\Lambda^{\top}_0 \Sigma^{-1}_{0,\varepsilon})\varepsilon_t\|_F}\\
&  & \leq O_p(p^{-1})\Big(\|(\widehat{\Lambda}^\top-\Lambda^*_0)\widehat{\Sigma}^{-1}_{\varepsilon}\varepsilon_t\|_F+\|\Lambda^{\top}_0(\widehat{\Sigma}^{-1}_{\varepsilon}-\Sigma^{-1}_{0,\varepsilon})\varepsilon_t\|_F\Big)=:P_{1,t}+P_{2,t}.
\end{eqnarray*}
Define by $\widehat{v}_{k,t}$ the $k$-th entry of $\widehat{\Sigma}^{-1}_{\varepsilon}\varepsilon_t$. Then
\begin{equation*}
\max_{t\leq T} P_{1,t} \leq O_p(p^{-1})\sqrt{\sum^{p}_{k=1}\|(\widehat{\lambda}_{k}-\lambda_{0,k})\widehat{v}_{k,t}\|^2_2} =O_p(p^{-1})\max_{t\leq T}\max_{k\leq p}|\widehat{v}_{k,t}|\|\widehat{\Lambda}-\Lambda_0\|_F.
\end{equation*}
Now note that $P_{2,t} = O_p(p^{-1})\|\Lambda^{\top}_0\widehat{\Sigma}^{-1}_{\varepsilon}(\Sigma_{0,\varepsilon}-\widehat{\Sigma}_{\varepsilon})\Sigma^{-1}_{0,\varepsilon}\varepsilon_t\|_F$. 
Then define $\widehat{\xi}_k$ as the $k$-th column of $\Lambda^{\top}_0\widehat{\Sigma}^{-1}_{\varepsilon}$ and $v_{k,t}$ the $k$-th entry of $\Sigma^{-1}_{0,\varepsilon}\eps_t$. We obtain
\begin{align*}
\|\Lambda^{\top}_0(\widehat{\Sigma}^{-1}_{\varepsilon}-\Sigma^{-1}_{0,\varepsilon})\varepsilon_t\|_F &= \sqrt{\sum^{p}_{k=1}\|\widehat{\xi}_k(\Sigma_{0,\varepsilon,kk}-\widehat{\Sigma}_{\varepsilon,kk})v_{k,t}\|^2_2}\\
&\leq \sqrt{\sum^{p}_{k=1}\|\widehat{\xi}_k\|^2|v_{k,t}|^2|\widehat{\Sigma}_{\varepsilon,kk}-\Sigma_{0,\varepsilon,kk}|^2} \leq  O_p(1)\sqrt{\sum^{p}_{k=1}|v_{k,t}|^2|\widehat{\Sigma}_{\varepsilon,kk}-\Sigma_{0,\varepsilon,kk}|^2} 
\end{align*}
Since $\max_k\|\widehat{\xi}_k\|_2=O_p(1)$, we deduce
\begin{equation*}
\|L_{2,t}\|_2 \leq O_p(p^{-1})\max_{t\leq T}\max_{k \leq p}|\widehat{v}_{k,t}|\|\widehat{\Lambda}-\Lambda_0\|_F+O_p(p^{-1})O_p(1)\max_{t\leq T}\sqrt{\sum^{p}_{k=1}|v_{k,t}|^2|\widehat{\Sigma}_{\varepsilon,kk}-\Sigma_{0,\varepsilon,kk}|^2}.
\end{equation*}
Under Assumption \ref{assumption_factor_moments}-(ii), $\max_{t\leq T}\max_{k \leq p}|\widehat{v}_{k,t}|=O_p((pT)^{1/4})$ and $\max_{t\leq T}\max_{k \leq p}|v_{k,t}|=O_p((pT)^{1/4})$. So $\max_{t\leq T}P_{1,t} =O_p((pT)^{-1/4})$, $\max_{t\leq T}P_{2,t} =O_p((pT)^{-1/4})$.
Furthermore, under Assumption \ref{assumption_factor_moments}-(iii), with $\Lambda_{k,0}$ the $k$-th row of $\Lambda_0$, we have
\begin{equation*}
\max_{t\leq T}\|L_{1,t}\|_2 \leq O_p(p^{-1})
\max_{t\leq T}\sum^{p}_{k=1}\|\varepsilon_{k,t}\xi_k\|_2=O_p(p^{-1})O_p(\sqrt{p}T^{1/4}),
\end{equation*}
where $\xi_k$ the $k$-th column of $\Lambda^{\top}_0\Sigma^{-1}_{0,\varepsilon}$.
We deduce $\max_{t\leq T}\|\widehat{f}_t-f_t\|_2 = O_p(\frac{1}{T^{1/4}}+\frac{T^{1/4}}{\sqrt{p}})$.}

\section{Derivative formulas}\label{appendix_derivative}

In this section, we derive the gradients and Hessians for both $\Lb_T(\mathbf{F};\cdot)$ and $\Gb_T(\mathbf{F};\cdot;\cdot)$. 

\vspace*{0.3cm}

\noindent\textbf{Computation of $\nabla_{\theta_1}\Lb_T(\mathbf{F};\theta_1), \nabla^2_{\theta_1\theta^\top_1}\Lb_T(\mathbf{F};\theta_1)$.} To derive the gradient function $\nabla_{\theta_1}\Lb_T(\mathbf{F};\theta_1)$, we consider the the differential $\dd \Lb_T(\mathbf{F};\theta_1)$ with respect to $\underline{\Psi}$, which is
{\small{\begin{eqnarray*}
\lefteqn{\dd\Lb_T(\mathbf{F};\theta_1)=\dd\Big[\frac{1}{2}\overset{T}{\underset{t=1}{\sum}}\big(x_t-\underline{\Psi}Z_{q,t-1}\big)^\top \big(x_t-\underline{\Psi}Z_{q,t-1}\big) \Big]}\\
& = & \frac{1}{2}\overset{T}{\underset{t=1}{\sum}} \tr\big(-x^\top_t (\dd\underline{\Psi})Z_{q,t-1}-Z^\top_{q,t-1}(\dd\underline{\Psi})^\top x_t+Z^\top_{q,t-1}(\dd\underline{\Psi})^\top\underline{\Psi}Z_{q,t-1}+Z^\top_{q,t-1}\underline{\Psi}^\top(\dd\underline{\Psi})Z_{q,t-1}\big).
\end{eqnarray*}}}
By identification, using Chapter 13.2 of \cite{abadir2005}, $\nabla_{\theta_1} \Lb_T(\mathbf{F};\theta_{1}) =  - \text{vec}\big(\overset{T}{\underset{t=1}{\sum}} \big(x_t-\underline{\Psi} Z_{q,t-1}\big)Z^\top_{q,t-1}\big)$. To get the Hessian, we compute the differential operator twice with respect $\underline{\Psi}$, which is
\begin{equation*}
\dd^2\Lb_T(\mathbf{F};\theta_{1}) = \overset{T}{\underset{t=1}{\sum}} \text{tr}\big(Z^\top_{q,t-1}(\dd \underline{\Psi})^\top(\dd\underline{\Psi})Z_{q,t-1}\big) = \overset{T}{\underset{t=1}{\sum}} \text{tr}\big(Z_{q,t-1}Z^\top_{q,t-1}(\dd \underline{\Psi})^\top I_m(\dd\underline{\Psi})\big). 
\end{equation*}
Hence, using 13.49 of \cite{abadir2005}, we get $\nabla^2_{\theta_1 \theta^\top_1} \Lb_T(\mathbf{F};\theta_{1}) = \overset{T}{\underset{t=1}{\sum}} (Z_{q,t-1} Z^\top_{q,t-1} \otimes I_m)$.

\vspace*{0.3cm}

\noindent \textbf{Computation of $\nabla_{\theta_2}\Gb_T(\mathbf{F};\theta_1;\theta_2), \nabla^2_{\theta_2\theta^\top_2}\Gb_T(\mathbf{F};\theta_1;\theta_2)$.} Let us now consider the second step loss function. \textcolor{black}{To simplify the notations, we write $K_{t-1}=K_{t-1}(\theta_{01})$}. The score $\nabla_{\theta_2}\Gb_T(\mathbf{F};\theta_1;\theta_2)$ and Hessian $\nabla^2_{\theta_2\theta^\top_2}\Gb_T(\mathbf{F};\theta_1;\theta_2)$ can be obtained in a similar manner, where $\theta_2 = \text{vec}(\Gamma)$. The first order differential is
\begin{align*}
\dd \Gb_T(\mathbf{F};\theta_1;\theta_2)
&= \frac{1}{2}\overset{T}{\underset{t=1}{\sum}}\tr\big(-x^\top_t(\dd\Gamma)K_{t-1}-K^\top_{t-1}(\dd\Gamma)^\top x_t + K^\top_{t-1}(\dd\Gamma)^\top \Gamma K_{t-1} + K^\top_{t-1}\Gamma^\top (\dd\Gamma)K_{t-1}\big)\\
& = \frac{1}{2}\overset{T}{\underset{t=1}{\sum}}\tr\big(-2 x^\top_t(\dd\Gamma)K_{t-1} + 2 K^\top_{t-1}\Gamma^\top(\dd\Gamma)K_{t-1}\big).
\end{align*}
As a consequence, by identification, the first order derivative is $\nabla_{\theta_2}\Gb_T(\mathbf{F};\theta_1;\theta_2)= - \text{vec}\big(\overset{T}{\underset{t=1}{\sum}}\big(x_t-\Gamma K_{t-1}\big)K^\top_{t-1}\big)$.
The second order differential and Hessian are
\begin{equation*}
\dd^2\Gb_T(\mathbf{F};\theta_1;\theta_2)=
 \overset{T}{\underset{t=1}{\sum}}\tr\big(K^\top_{t-1}(\dd\Gamma)^\top (\dd\Gamma)K_{t-1}\big), \; \nabla^2_{\theta_2\theta^\top_2}\Gb_T(\mathbf{F};\theta_1;\theta_2) = \overset{T}{\underset{t=1}{\sum}}\big( K_{t-1}K^\top_{t-1}\otimes I_m\big).
\end{equation*}

\section{\textcolor{black}{VARMA(1,1) representation for $x_t$}}\label{appendix_varma}

\textcolor{black}{Rewriting equation (\ref{eq:ssf}), we obtain:
\[
\forall t, \; x_t = \nu^* + \Phi(x_t -\nu^*) + w_t, \quad w_t = \xi_t^* - \Phi \xi_{t-1}^* + \eta_{t-1}.
\]
We use the approach of Subsection 4.7 of \cite{hamilton1994} to show that $w_t$ follows the VMA(1) (vector moving-average process of order 1), which indicates that $x_t$ follows the VARMA(1,1) process. The autocovariance function of $w_t$ is given by
\begin{align*}
\Eb[w_t w_{t-i}^\top] = \left\{ \begin{array}{ll}
\Sigma_{\xi} + \Phi \Sigma_{\xi} \Phi^\top + \Sigma_{\eta} & \mbox{for } i=0,\\
- \Phi \Sigma_{\xi} & \mbox{for } i= \pm 1, \\
0 & \mbox{otherwise},
\end{array} \right.
\end{align*}
showing that $w_t$ is covariance stationary, and its autocovariances are zero beyond one lag, as are those for VMA(1). Hence, we consider the VMA(1) representation of $w_t$, as 
\[\forall t,\;  w_t = u_t - \Upsilon u_{t-1}, \quad u_t \sim WN(0,\Sigma_u), \]
which satisfies the moment conditions,
\begin{align*}
\Eb[w_tw_t^\top] &= \Sigma_u + \Upsilon \Sigma_u \Upsilon^\top 
=  \Sigma_{\xi} + \Phi \Sigma_{\xi} + \Phi^\top + \Sigma_{\eta}, \\
\Eb[w_t w_{t-1}^\top] &= - \Upsilon \Sigma_u = - \Phi \Sigma_{\xi}.
\end{align*}
By the second equation, we obtain $\Sigma_u = \Upsilon^{-1} \Phi \Sigma_{\xi}$.
Substituting $\Sigma_u$ in the first equation with it, we obtain,
\[
\Upsilon (\Phi \Sigma_{\xi}) \Upsilon^\top - \Upsilon (\Sigma_{\xi} + \Phi \Sigma_{\xi} \Phi^\top + \Sigma_{\eta})
+ \Phi \Sigma_{\xi} =0.
\]
Thus the $\Upsilon$ is the solution which satisfies the invertible condition.
}

\textcolor{black}{Hence $x_t$ follows the VAMA process (\ref{eq:varma}) with $u_t = (I_m - \Upsilon L)^{-1}(I_m - \Phi L) \xi_t^* + (I_m - \Upsilon L)^{-1} \eta_{t-1}$, where $L$ is the lag operator.
In spite of the structure, $u_t$ and $u_{t-i}$ ($i=\pm 1,\pm 2, \ldots $) are uncorrelated.
}

\section{\textcolor{black}{Implementation details for MSV estimation}}\label{appendix_sec:implementation_MSV}

\textcolor{black}{
In this subsection, we discuss the computational issues relating to criterion (\ref{obj_crit_first_step}) in \textbf{Step 1} for the estimation of the MSV parameters. Hereafter, $\gamma>0$ is fixed. 
The tuning parameter $\lambda_T$ controls the model complexity and must be calibrated.}
\textcolor{black}{To select the latter, several procedures have been proposed to take into account the dependency among observations, but there is no consensual approach. Based on thorough simulated and real data experiments, \cite{cerqueira2020} concluded that good estimation performances are obtained by the ``blocked'' cross-validation when the process is stationary; under non-stationary time series, the most accurate estimations are produced by ``out-of-sample'' methods, where the last part of the data is used for testing. In our experiments, we follow the out-of-sample approach for cross-validation. To be specific, we split the full sample into training and test sets. The training sample corresponds to the first $75\%$ of the entire sample and the test sample to the last $25\%$. For given values of $\lambda_T$, we solve (\ref{obj_crit_first_step}) in the training set and then compute the loss function in the test set using the estimator computed on the training set. This procedure is repeated for all the $\lambda_T$ candidates and we select the one minimizing the loss. Then, we estimate the model over the full sample using the optimal $\lambda_T$. }

\textcolor{black}{Formally, the procedure is as follows: divide the data into two sets, the training set and test set; define the cross-validation score as $\text{CV}(\lambda_T) = T^{-1}_{\text{oos}}\sum_{i \in I_{\text{oos}}}\ell(\widehat{F}_{t};\widehat{\theta}^{\text{in}}_{1}(\lambda_T))$, where $T_{\text{oos}}$ is the size of the test set, $I_{\text{oos}}$ is the index of observations in the test set, $\widehat{\theta}^{\text{in}}_{1}(\lambda_T)$ is the estimator based on the training set with regularization parameter $\lambda_T$, and $\ell(\widehat{F}_t;\cdot)$ is the non-penalized loss. The optimal regularization parameter $\lambda^*_T$ is then selected according to: $\lambda^*_T = \arg \,\min_{\lambda_T}\, \text{CV}(\lambda_T)$. Then, $\lambda^*_T$ is used to obtain the final estimator over the whole data. 
Here, the minimization of the cross-validation score is performed over a grid of $\lambda_T$ candidates, which is set as follows: first, we solve problem (\ref{obj_crit_first_step}) using this out-of-sample technique by the LASSO, i.e., with $\tau(\widetilde{\theta}_{1,k})=1$ for all $k$, and where the smallest (resp. largest) $\lambda_T$ candidate is set as $0$ (resp. the value for which all the estimated coefficients are set to $0$); denote by $\lambda^*_{1,T}$ the corresponding optimal tuning parameter; then the grid to solve the adaptive LASSO problem (\ref{obj_crit_first_step}) is set as $c\lambda^*_{1,T}$ with $c=0.1,0.2,\ldots,5$.  
One advantage of our proposed procedure is that the penalized estimation can be performed equation-by-equation. \textcolor{black}{We rely on CVX (\cite{cvx,gb08}) to solve the penalized problem.} }

\textcolor{black}{It is worth mentioning that the cross-validation procedure may accommodate the selection of $\gamma$. However, this would increase the computational complexity. \cite{zou2006adaptive} considered the values $\{0.5,1,2\}$ and $\gamma\geq 1$ provided the best performances; \cite{fan2009} set $\gamma=0.5$. For simplicity, we set $\gamma=1$ in all our experiments.
Finally, $q$ must be chosen to solve problem (\ref{obj_crit_first_step}). The larger, the better the approximation is: since $\|\mathbf{\Psi}_ {0q}\|_{s,1}\leq \kappa \frac{\beta^{q}}{1-\beta}$ with $0 \leq \beta <1, \kappa>0$ under exponential decay, the truncation bias $\|\mathbf{\Psi}_ {0q}\|_{s,1}$ will tend to zero as $q \rightarrow \infty$. In our simulations and real data experiments, we work with $q = 10$, a value of the same order as $L_1 \log(T)$ ($T=2000$ in our simulations; the sample sizes for real data are of the same order).
}

\section{Competing models}

\subsection{DCC model}\label{dcc_model}

Rather than a direct specification of the variance-covariance matrix process $(H_t)$ of $(y_t)$, an alternative approach is to split the task into two parts: individual conditional volatility dynamics on one side;  conditional correlation dynamics on the other side. The most commonly used correlation process is the Dynamic Conditional Correlation (DCC) of \cite{engle2002}. In its BEKK form, the general DCC model is specified as: $y_t = H^{1/2}_t \eta_t$, with $H_t := \Eb[y_t y^\top_t | \Fc_{t-1}]$ positive definite, where:
\begin{equation*}
H_t = D_t R_t D_t, \; R_t =  Q^{\star-1/2}_t Q_t Q^{\star-1/2}_t, \;
Q_t = \Omega + \overset{q}{\underset{k=1}{\sum}} M_k Q_{t-k} M^\top_k + \overset{r}{\underset{l=1}{\sum}} W_l u_{t-l} u^\top_{t-l} W^\top_l,
\end{equation*}
where $D_t = \text{diag}\left(\sqrt{h_{1,t}},\sqrt{h_{2,t}},\ldots,\sqrt{h_{p,t}}\right)$ with $(h_{j,t})$ the univariate conditional variance dynamic of the $j$-th component, $u_t = \left(u_{1,t},\ldots,u_{p,t}\right)$ with $u_{j,t} = y_{j,t}/\sqrt{h_{j,t}}$, $Q_t = \left[q_{t,uv}\right]$, $Q^{\star}_t = \text{diag}\left(q_{11,t},q_{22,t},\ldots,q_{pp,t}\right)$. $\Fc_{t-1}$ is the sigma field generated by the past information of the $p$-dimensional process $(y_t)_{t \in \mathbb{Z}}$ until (but including) time $t-1$. The model is parameterized by some deterministic matrices $(M_k)_{k = 1,\ldots,q}$, $(W_l)_{l = 1,\ldots,r}$ and a positive definite $p \times p$ matrix $\Omega$. In the original DCC of \cite{engle2002}, the process $(Q_t)$ is defined by $Q_t = \Omega^\star +\sum^q_{k=1} B_k \odot Q_{t-k} +\sum^r_{l=1} A_l \odot u_{t-l} u^\top_{t-l}$, where the deterministic matrices $(B_k)_{k = 1,\ldots,q}$ and $(A_l)_{l = 1,\ldots,r}$ are positive semi-definite. Since the number of parameters of the latter models is of order $O(p^2)$, our applications are restricted to scalar $B_k$'s and $A_l$, say $b_k, a_l$, and we work with $q=r=1$. Moreover, we employ a correlation targeting strategy, that is we replace $\Omega^\star$ by $(1-a_1-b_1)\overline{Q}$ with $\overline{Q}$ the sample covariance of the standardized returns $u_t$. We specify a GARCH(1,1) model for each individual conditional variance dynamic $h_{j,t}$. The estimation is performed by a standard two-step Gaussian QML for the MSCI and S\&P 100 portfolio. For the S\&P 500 portfolio, we employ the composite likelihood method proposed in~\cite{pakel2021} to fix the bias problem caused by the two-step quasi likelihood estimation and to make the likelihood decomposition plausible for large vectors. The composite likelihood method consists of averaging the second step likelihood (correlation part) over $2 \times 2$ correlation-based likelihoods. To be precise, using the $n$-sample of observations, the second step composite likelihood becomes:
\begin{equation*}
L_{2n}(a_1,b_1) = \frac{1}{n} \overset{n}{\underset{t=1}{\sum}}\overset{L}{\underset{l=1}{\sum}} \Big[\log(|R^{(l)}_t|)+u^{(l)\top}_t\big(R^{(l)}_t\big)^{-1}u^{(l)}_t\Big],
\end{equation*}
where $R^{(l)}_t$ is the $2 \times 2$ correlation matrix with $l$ corresponding to a pre-specified pair of indices in $\{1,\ldots,p\}$, $u^{(l)}_t$ is a $2 \times 1$ sub-vector of the standardized residuals $u_t$, with $k$ selecting the pre-specified pair of indices in $\{1,\ldots,p\}$. There are several ways of selecting the pairs: every distinct pair of the assets, i.e, $u^{(1)}_t=(u_{1,t},u_{2,t})^\top,u^{(2)}_t=(u_{1,t},u_{3,t})^\top,\ldots,u^{p(p-1)/2}_t=(u_{p-1,t},u_{p,t})^\top$, so that $L = p(p-1)/2$ in $L_{2n}(a_1,b_1) $; or contiguous overlapping pairs, i.e., $u^{(1)}_t=(u_{1,t},u_{2,t})^\top,u^{(2)}_t=(u_{2,t},u_{3,t})^\top,\ldots,u^{(p-1)}_t=(u_{p-1,t},u_{p,t})^\top$ so that $L=p-1$. We use the contiguous overlapping pairs for computation gains, where the complexity is $O(p)$, in contrast to the $O(p^2)$ complexity if $L=p(p-2)/2$: this pair construction is the 2MSCLE method of \cite{pakel2021}. 

\subsection{BEKK model}\label{bekk_process}

\cite{engle1995} proposed a new way to generate a positive-definite variance-covariance $(H_t)$: the BEKK model (Baba, Engle, Kraft and Kroner). The BEKK(p,q,K) model specifices the dynamic
\begin{equation*}
H_t = \Omega + \sum^q_{i=1}\sum^K_{k=1} A_{ik}y_{t-i}y^\top_{t-i} A^\top_{ik} + \sum^p_{j=1}\sum^K_{k=1} B_{jk}H_{t-j} B^\top_{jk},
\end{equation*}
with $K$ is an integer, $\Omega \succ 0$, $A_{ik}$ and $B_{jk}$ are square $p \times p$ matrices. Clearly, the number of parameters to estimate explodes with $p$. Therefore, it is common to work with the restricted form scalar BEKK(1,1,1) given as
$H_t = \Omega + a^2 y_{t-1}y^\top_{t-1} + b^2H_{t-1}$,
where $a,b$ are scalar parameters. Variance targeting can be performed by simply setting $(1-a^2-b^2)\widehat{S}$ in lieu of $\Omega$, with $\widehat{S}$ the sample variance-covariance matrix. 

\color{black}

\section{Implementation issues}

\subsection{Computational efficiency}\label{comp_cost}

\begin{figure}[h]
\caption{\label{fig:fig01} Computation time 
}
\begin{center}
\includegraphics[width=16cm]{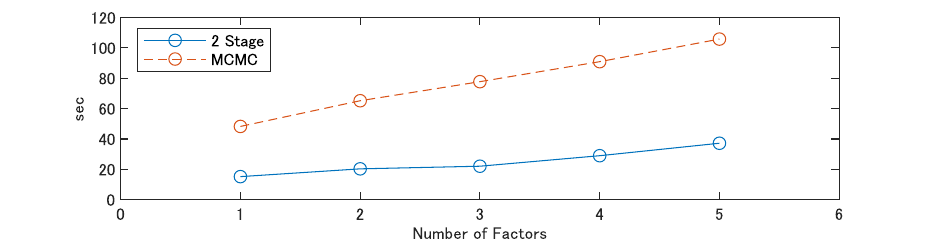}
\end{center}
\end{figure}

We compare the computational time of our two-stage procedure and the MCMC method suggested by \cite{kastner2017}.
For the MCMC estimation, we utilized the R package {\texttt{factorstochvol}}, which is available on CRAN;
All R-based computations were performed using R version 4.5.2.
In contrast, our proposed method was implemented in MATLAB R2025b.
We conducted the computations on a personal computer equipped with an Intel Core i7-12700 processor (2.1 GHz), 32 GB of RAM, and an NVIDIA RTX 3070 GPU.
We focus on the MSCI data ($p=23$, $T=3900$) with small number of factors ($m=1,\ldots,5$).
For the implementation of the MCMC algorithm, we specify $\Phi$ as diagonal, and set the length of MCMC sampler as 2500. Note that $\Phi$ is non-diagonal for the two-stage procedure.
Figure \ref{fig:fig01} shows that the computational time for the two methods increases linearly, as $m$ increases. 
Our two-stage procedure is 3.2 times faster than the MCMC method on average in Figure \ref{fig:fig01}. Although our procedure is faster than the MCMC method in our current setup, the execution speed of the latter could potentially be improved if implemented in MATLAB instead of R.

\subsection{Portfolios}\label{portfolio}

We omit the subscript $t$ here.
For a return vector $y$ with corresponding covariance matrix $H$, we can define a portfolio retunrn as $y_p = w^\top y$ with the portfolio variance $h_p = w^\top Hw$, where $w$ is the $p \times 1$ vector of portfolio weights. Note that $\iota^\top w=1$, where $\iota$ is the $p \times 1$ vector of ones.

The GMVP defined by (\ref{portprob}) can be obtained by equalizing all the marginal risk contributions, $\frac{\partial h_p}{\partial w_i} = \frac{\partial h_p}{\partial w_j}$ $(\forall i,j)$.  On the other hand, the idea behind risk parity portfolio (RPP) is equalizing all the total risk contributions, $w_i \frac{\partial h_p}{\partial w_i} = w_j \frac{\partial h_p}{\partial w_j}$ $(\forall i,j)$.
For this purpose, the RPP is used in the context of long-only portfolios. 
As discussed in \cite{xbai2016}, define $w^c=(w_1^c,\ldots,w_p^c)^\top$ as
\[
w^c = \mbox{arg} \mathop{\mbox{min}}_{w} w^\top Hw - c \sum_{i=1}^p \log w_i, 
\]
where $c$ is a positive constant.
Note that the gradient is given by $Hw - c w^{-1}$, where $w^{-1} = (1/w_1,\ldots,1/w_p)^\top$.
The weights for the RPP is given by $w_{\mbox{\tiny RPP}}=\{ w_i^* \}$ with $w_i^* = w_i^c/\sum_{i=1}^p w_i^c$, which is uniquely determined to be irrelevant to the constant $c$ (See Lemma 2.2 of \cite{xbai2016}).

\end{document}